%% file: graph_main_arxiv.tex
\documentclass{IEEEtran}
\usepackage[T1]{fontenc}
\usepackage[latin1]{inputenc}
\usepackage{float}
\usepackage{amsmath}
\usepackage[noadjust]{cite}
\usepackage{amssymb}
\usepackage{enumerate}
\usepackage{amsthm}
\usepackage{graphics,epsfig,psfrag}
\usepackage{subfigure,epsf,pstricks,subfloat}
\usepackage{bm}
\usepackage{url}
\floatstyle{ruled}
\newfloat{algorithm}{tbp}{loa}
\floatname{algorithm}{Routine}

\usepackage{algorithmic}
\algsetup{linenodelimiter= }

\renewcommand{\algorithmicreturn}{\textbf{Return:}}
\def\bfa{{\mathbf a}}

\def\bfw{{\mathbf w}}
\def\bfx{{\mathbf x}}
\def\bfy{{\mathbf y}}
\def\bfz{{\mathbf z}}
\title{Sparse Recovery with Graph Constraints}
\author{\IEEEauthorblockN{
Meng Wang, %\IEEEauthorrefmark{2},
Weiyu Xu, %\IEEEauthorrefmark{1},
Enrique Mallada,
Ao Tang
}
\thanks
{M. Wang is with Rensselaer Polytechnic Institute.  W. Xu is with the University of Iowa, Iowa City, IA. E. Mallada and A. Tang are with Cornell University, Ithaca, NY.

Partial and preliminary results have appeared in
\cite{WXMT12}.
}
}

\usepackage{graphicx}
\usepackage{epstopdf}
\newtheorem{theorem}{Theorem}
\newtheorem{lemma}{Lemma}
\newtheorem{cor}{Corollary}

\newtheorem{defi}{Definition}
\newtheorem{prop}{Proposition}

\begin{document}

\maketitle \thispagestyle{empty} \pagestyle{empty}

%%%%%%%%%%%%%%%%%%%%%%%%%%%%%%%%%%%%%%%%%%%%%%%%%%%%%%%%%%%%%%%%%%%%%%%%%%%%%%%%
\begin{abstract}
%Sparse recovery can recover high-dimensional sparse vectors from low-dimensional non-adaptive measurements. %Unlike existing results in sparse recovery,
Sparse recovery can recover sparse signals from a set of underdetermined linear measurements.
Motivated by the need to monitor large-scale networks from a limited number of measurements, this paper addresses the problem of recovering sparse signals in the presence of network topological  constraints. %which are characterized by a graph.
 Unlike conventional sparse recovery  where a measurement can contain any subset of the unknown variables, we use a graph to characterize the topological constraints and allow an additive measurement over nodes (unknown variables) only if they induce a connected subgraph.  %Instead of using random measurements corresponding to random walks on the graph as \cite{XMT11}, we provide explicit measurements construction for any given graph.
We provide explicit measurement constructions for several special graphs, and the number of measurements by our construction is less than that needed by existing random constructions. Moreover, our construction for a line network is provably optimal in the sense that it requires the minimum number of measurements. %and later extend the results to general graphs.
A  measurement construction algorithm for general graphs is also proposed and evaluated. For any given graph $G$ with $n$ nodes, we derive   bounds of the minimum number of measurements needed to recover any $k$-sparse vector over $G$ ($M^G_{k,n}$). Using the Erd\H{o}s-R\'enyi random graph %$G(n,p)$
as an example, we characterize the dependence of $M^G_{k,n}$ on the graph structure. % which is captured by edge probability $p$. %Surprisingly, the number of measurements constructed by our methods on a general graph is not far away from that in a complete graph, even though the number of edges in a general graph is much less than that in a complete graph.
Our study suggests that $M^G_{k,n}$ may serve as a graph connectivity metric.
\end{abstract}

%%%%%%%%%%%%%%%%%%%%%%%%%%%%%%%%%%%%%%%%%%%%%%%%%%%%%%%%%%%%%%%%%%%%%%%%%%%%%%%%
\input{intro.tex}
\input{model_arxiv.tex}
\input{special_arxiv.tex}

\input{general.tex}
\input{erdos.tex}
\input{extension_arxiv.tex}
\input{simu_arxiv.tex}
\section{Conclusion}\label{sec:con}

This paper addresses the sparse recovery problem with graph constraints. %. Instead of random arguments, we
%For any given graph,
We provide explicit measurement constructions for special graphs, and propose measurement design algorithms for general graphs. Our construction for a line network is optimal in terms of the number of measurements needed. The constructions on other graphs also improve over the existing results. We characterize the relationship between the number of measurements for sparse recovery and the graph topology. We also derive upper and lower bounds of the minimum number of measurements needed for sparse recovery on a given graph. It would be interesting to tighten such bounds, especially the lower bounds.
%By providing explicit measurement constructions for different graphs, we derive upper bounds of the minimum number of measurements needed to recover vectors up to certain sparsity. It would be interesting to explore corresponding tight lower bounds. %Further efforts are also needed to empirically evaluate the performance of different recovery themes, especially when the measurements are noisy.

%We need to remark that this paper is only the first step towards network measurement constructions with topological constraints, and several practical concerns have not been taken into account yet. For instance,
We have not considered the effect of the measurement noise. Also, we assume full knowledge of the fixed network topology, and  measurement construction when the topology is time-varying or partially known is an open question. %%We also assume that any number of nodes can be measured together as long as they form a connected
\bibliographystyle{IEEEtranS}
%\bibliography{../IEEEabrv,../ref,../MengWangPub}
%\bibliography{IEEEabrv,ref,MengWangPub}

\begin{IEEEbiography}
%[{\includegraphics[width=1in,height=1.25in,clip,keepaspectratio]{figures/MengWang.jpg}}]
{Meng Wang} (S'06-M'12) received   B.E. (Hon.) and M.S. (Hon.) from Tsinghua University, Beijing, China, and Ph. D. from Cornell University, Ithaca, NY, in
2005, 2007, and 2012, all in
Electrical Engineering. %, respectively.

She is   currently an Assistant Professor in the Department of Electrical, Computer, and Systems Engineering at Rensselaer Polytechnic Institute. Her research interests include communication
networks, signal processing, and nonconvex optimization and its
applications.

Dr. Wang is a recipient of Jacobs Fellowship of Cornell University in 2008 and 2010.
\end{IEEEbiography}

\begin{IEEEbiography}
%[{\includegraphics[width=1in,height=1.25in,clip,keepaspectratio]{figures/WeiyuXu.eps}}]
{Weiyu Xu} received his B.E.  in Information Engineering from
Beijing University of Posts and Telecommunications in 2002, and a
M.S. degree in Electronic Engineering from Tsinghua University in
2005. He received a M.S.  and a Ph.D. degree in Electrical Engineering in 2006 and 2009   from California Institute of
Technology, with a minor in Applied and Computational Mathematics.

He is currently an Assistant Professor at the Department of Electrical
and Computer Engineering at the University of Iowa.  His
research interests are in signal processing, compressive sensing,
communication networks, information and coding theory.

Dr. Xu is a recipient of the Information Science and Technology
Fellowship at Caltech,  and the
recipient of Charles and Ellen Wilts doctoral research award  in 2010.
\end{IEEEbiography}

\begin{IEEEbiography}
%[{\includegraphics[width=1in,height=1.25in,clip,keepaspectratio]{figures/Enrique.png}}]
{Enrique Mallada} (S'09) received his Ingeniero en Telecomunicaciones degree from Universidad ORT, Uruguay, in 2005.  From 2006 to 2008 he was
teaching and research assistant in the Department of Telecommunications at Universidad ORT. He is currently  pursuing the Ph.D. degree in electrical and computer engineering at Cornell University, Ithaca, NY. His research interests include communication networks, control, nonlinear dynamics and optimization.
\end{IEEEbiography}
\begin{IEEEbiography}
%[{\includegraphics[width=1in,height=1.25in,clip,keepaspectratio]{figures/KevinTang.jpg}}]
{Ao Tang} (S'01-M'07-SM'11) received the B.E. (Honors.) in electronics
engineering from Tsinghua University, Beijing, China, and the M.S. and Ph.D.
degrees in electrical engineering with a minor in applied and computational
mathematics from the California Institute of Technology, Pasadena, CA, in 1999, 2002, and 2006, respectively.

He is currently an Assistant Professor in the School of Electrical
and Computer Engineering at Cornell University, where he conducts
research on the control and optimization of engineering networks including
communication networks, power networks and on-chip networks.

Dr. Tang was a recipient of the Cornell Engineering School Michael Tien' 72 Excellence in Teaching Award in 2011, and the Young Investigator Award from the Airforce Office of Scientific Research (AFOSR) in 2012.
\end{IEEEbiography}

\end{document}

%% file: intro.tex
\section{Introduction} \label{sec:intro}

In the monitoring of engineering networks,   one often needs to extract network state parameters from indirect observations. In network tomography  \cite{BDPT02,BCGRS01,CHNY02,CBSK07,Duffield06,GR12,GS09,NT06,ZCB06},   since measuring each component (e.g., router) in the communication network directly can be operationally costly, if feasible at all, the goal is to infer system internal characteristics   such as link bandwidth utilizations and  link queueing delays from indirect aggregate measurements. %In sensor networks, compared with directly transmitting all sensor values to the fusion center, the total transmission cost can be significantly reduced if the sensor values are aggregated together along the forwarding paths \cite{HBRN08}.

In many cases, it is desirable to reduce the number of measurements without sacrificing the monitoring performance.  %the total number of aggregate measurements is much smaller than the number of components in a network. But we still hope to extract the status of each individual component with some prior knowledge of the unknown signal to recover.
For example, %instead of monitor all source-destination path delays,
network kriging \cite{CKC06} uses the fact  that different paths experience the same delay on the same link, and shows that by measuring delays on $n$ linearly independent paths, one can recover delays on all $n$ links in the network, and thus identify the delays on possibly exponential number of paths.
 %provided that we have prior knowledge of the status (the unknown signal to be recovered). For example,
%In practice, the total number of aggregate measurements we can take is small compared with the size of the network. % many as the objects in the networks, especially when the size of the network is large.
%However, we can still extract the most dominating elements of a high-dimensional signal from low-dimensional non-adaptive measurements. % with efficient recovering schemes.
Surprisingly, the number of path delay measurements needed to   recover $n$ link delays  can be further reduced by exploiting the fact that  only a small number of bottleneck links %in the communication networks
experience large delays, while the delay is approximately zero elsewhere. % reduced to $m$  with $m \ll n$ . This seems infeasible  since $n$ measurements are required to determine a general $n$-dimensional signal. However,
\textit{Sparse Recovery} theory promises that if the signal of interest is sparse, i.e., its most entries are zero, $m$ measurements are sufficient to correctly recover the signal, even though $m$ is much smaller than the signal dimension.
%For a sparse signal such as the vector of link delays, i.e. most entries are zero,  it can be  exactly recovered from $m$ measurements even though $m$ is much smaller than signal dimension. %For example, transmission delays in the communication networks can be represented by an approximately sparse signal, since only a small number of bottleneck links %in the communication networks
%experience large delays, while the delay is approximately zero elsewhere.
%\textit{Sparse Recovery} addresses the problem of recovering sparse  signals from a set of underdetermined linear measurements, and
Since many network parameters are sparse, e.g., link delays, these network tomography problems can be formulated as a sparse recovery problem with the goal of minimizing the number of indirect observations.

Sparse recovery has two different  but closely related problem formulations.
%There are two different but closely related problem formulations for recovering sparse high-dimensional signal from low-dimensional measurements.
One is \textit{Compressed Sensing} \cite{CaT05,CaT06,DoT05,Don06,BGIKS08}, where the signal is represented by a high-dimensional real vector, %denoting e.g., transmission delays, packet loss rates, etc.
and an aggregate measurement is the arithmetical sum of the corresponding real entries. The other is \textit{Group Testing} \cite{Dorfman43, DH00},
%In compressed sensing, %Since most objects are in normal condition, the corresponding values are all zero, and the goal is to identify a small number of non-zero values from grouped measurements.
%This topic with this formulation is known as Compressed Sensing or Sparse Recovery, which has been extensively studied recently. %Compressed sensing focuses on recovering a high-dimensional sparse signal vector from low-dimensional measurements, and has been extensively studied recently. \cite{CaT05,CaT06,DoT05,Don06,BGIKS08}
%In group testing,
where the high-dimensional signal is binary and a measurement is a logical disjunction (\textbf{OR}) on the corresponding binary values.
%In the second discrete formulation, the high-dimensional vector is logical. For example, in all-optical networks, the existence of switch node failure is denoted by `1', and '0' otherwise.  A measurement is a logical disjunction (\textbf{OR}) on the corresponding logical values, and thus is also a logical value. The goal is to identify the location of the small number of `1's among most `0's. This topic is known as group testing, which has been studied since 1943.

One key question in sparse recovery is to design a small number of non-adaptive measurements (either real or logical) such that all the vectors (either real or logical) up to certain sparsity (the support size of a vector) can be correctly recovered. Most existing results, however, rely critically on the assumption that any subset of the values can be aggregated together \cite{CaT05,Don06}, which is not realistic  in network monitoring problems where  %Measurement design in networks require
 only objects that form a path or a cycle on the graph \cite{ARK09,GR12}, or induce a connected subgraph can be aggregated together in the same measurement.  Only a few recent works consider graph topological constraints, either in group testing \cite{CKMS12} setup, especially motivated by link failure localization in all-optimal networks \cite{BTH11,CKMS12,HPWYC07,TWHR11,WHTJ10}, or in compressed sensing setup, with application in estimation of network parameters \cite{CPR07,HBRN08, XMT11}. % and group testing \cite{BTH11,CKMS12,HPWYC07,TWHR11,WHTJ10}.

%In
%One major difficulty in designing aggregate measurement for networks comes from the topological constraints.
%This paper addresses the sparse recovery problem with graph constraints. % mostly in the compresses sensing setup and sometimes provide some results in group testing for comparison.
We design measurements for  recovering sparse signals in the presence of graph topological constraints, and characterize the minimum number of measurements required to recover sparse signals when the possible measurements should satisfy graph constraints.
Though   motivated by network applications,  graph constraints abstractly models scenarios when certain elements cannot be measured together in a complex system. These constraints can result from various reasons, not necessarily lack of connectivity. Therefore, our results can be potentially useful to other applications besides network tomography.

Here are the main contributions of this paper.

\noindent{\bf (1)} We provide explicit measurement constructions for various graphs. Our construction for line networks is optimal in the sense that it requires the minimum number of measurements. For other special graphs, the number of  measurements by our construction is less than the existing estimates (e.g. \cite{CKMS12,XMT11}) of the measurement requirement. (Section \ref{sec:special})

\noindent{\bf (2)} %Starting with measurement constructions on special graphs,
For general graphs, we propose a measurement design guideline based on \textit{$r$-partition}, and further propose a simple measurement design algorithm. (Section \ref{sec:general}) %for general graphs and further show some of its properties. (Section \ref{sec:bound})

%\noindent{\bf (3)} A simple measurement design algorithm is proposed for general graphs, and %(Section \ref{sec:algo})
%we evaluate its performance both theoretically and numerically. (Section \ref{sec:algo} and \ref{sec:simu})
%Our problem formulation is very similar to that in \cite{XMT11}. \cite{XMT11} uses measurements based on random walks and shows that any $k$-sparse link vector (with no more than $k$ non-zero elements) can be recovered with overwhelming probability using $O(k \log(n))$ measurements. This paper considers node vectors instead of link vectors. Moreover, instead of random arguments, this paper provides constructive ways to design measurements satisfying graph constraints. We first construct measurements for several special networks and then extend the ideas to general networks and provide an algorithm to design the measurements for any give graph. Interesting, though the number of links in a general graph can be much smaller than that in the complete graph with the same number of nodes, the number of measurements designed by our algorithm to recover $k$-sparse vectors in fact is close to that needed in a complete graph.

\noindent{\bf (3)} Using  Erd\H{o}s-R\'enyi random graphs as an example, we characterize the dependence of the number of measurements for sparse recovery on the graph structure. % represented by the link probability $p$.  %in Erd\H{o}s-R\'enyi random graphs.
(Section \ref{sec:erdos})

Moreover,  we also propose measurement construction methods under additional practical  constraints such that the length of a measurement is bounded, or each measurement should pass one of a fixed set of nodes.   The issue of measurement error is also addressed. (Sections  \ref{sec:extension},\ref{sec:extension})

%\noindent{\bf (5)} Motivated by practical needs, we further propose measurement construction methods under additional graph constraints including measurement length constraints, and the requirement that each measurement should pass one of a fixed set of nodes. (Section \ref{sec:extension}) We also address the issue of sparse recovery when some critical measurements may contain errors.  (Section \ref{sec:huberror})

%We now start with Section \ref{sec:model} to introduce the model and problem formulation.

%The rest of the paper is organized as follows. We introduce the model and problem formulation in Section \ref{sec:model}. In Section \ref{sec:special}, we construct measurements on several special graphs. We then extend the obtained construction ideas to general graphs in Section \ref{sec:general}. We provide a design guideline to reduce the number of measurements, and also propose an algorithm with performance guarantee to construct the measurements on any given graph. In Section \ref{sec:simu}, we use numerical results to illustrate the effectiveness of our algorithm and the impact of graph topologies. We also study the sparse recovery performance when the measurements are erroneous and noisy. Section \ref{sec:con} concludes the paper. 

%% file: model_arxiv.tex
\section{Model and Problem Formulation}\label{sec:model}

We use a graph $G=(V,E)$ to represent the topological constraints, where $V$ denotes the set of nodes with cardinality $|V|=n$, and $E$ denotes the set of edges. Each node $i$ is associated with a real number $x_i$, and we say vector $\bfx=(x_i, i=1,...,n)$ is associated with $G$. $\bfx$ is the unknown signal to recover. %Given $\bfx$, let $T=\{i~|~ x_i \neq 0\}$ denote the support of $\bfx$, and let $\|\bfx\|_0=|T|$\footnote{The $\ell_p$-norm ($p\geq 1$) of $\bfx$ is $\|\bfx\|_p=(\sum_i |x_i|^p)^{1/p}$,  $\|\bfx\|_{\infty}=\max_i |x_i|$, and $\|\bfx\|_{0}=|\{i:x_i \neq 0\}|$.} denote the number of non-zero entries of $\bfx$.
We say $\bfx$ is a $k$-sparse vector if $\|\bfx\|_0=k$\footnote{The $\ell_p$-norm ($p\geq 1$) of $\bfx$ is $\|\bfx\|_p=(\sum_i |x_i|^p)^{1/p}$,  $\|\bfx\|_{\infty}=\max_i |x_i|$, and $\|\bfx\|_{0}=|\{i:x_i \neq 0\}|$.} i.e.,    the number of non-zero entries of $\bfx$ is $k$.

Note that in the monitoring of the link delays of a communication network represented by $\mathcal{N}_G$,  the graph model we consider is  the line graph \cite{HN60} (also known as interchange graph or edge graph) $L(\mathcal{N}_G)$ of $\mathcal{N}_G$. According to the definition of a line graph, every node in $G=L(\mathcal{N}_G)$ corresponds to a link in network $\mathcal{N}_G$, and the node value corresponds to the link delay. Two nodes in  $G$ are connected with an edge if and only if the corresponding links in network $\mathcal{N}_G$ are connected to the same router. See  Fig. \ref{fig:trans} (a) (b) as an example of a network and its line graph considered here.  %When we consider sensor network application, the nodes in $G$ represent sensor values, and the edges represent feasible direct communication between pairs of sensors. Thus, our graph model can be viewed as  the line graph of a communication network\footnote{In \cite{CKMS12,XMT11}, the authors have considered designing measurements that correspond to walks in the network directly, while our analysis focuses on the line graph of a given network.}, as well as an abstraction of a sensor network.
%Given $\bfx$, let $T=\{i~|~ x_i \neq 0\}$ denote the support of $\bfx$, and let $\|\bfx\|_0=|T|$\footnote{The $\ell_p$-norm ($p\geq 1$) of $\bfx$ is $\|\bfx\|_p=(\sum_i |x_i|^p)^{1/p}$,  $\|\bfx\|_{\infty}=\max_i |x_i|$, and $\|\bfx\|_{0}=|\{i:x_i \neq 0\}|$.} denote the number of non-zero entries of $\bfx$. We say $\bfx$ is a $k$-sparse vector if $\|\bfx\|_0=k$.
Since large delays only occur at a small number of bottleneck links, the link delays in a network can be represented by a sparse vector $\bfx$ associated with $G$.  %In sensor networks, we assume that there is a known estimate of sensor values, and the number of sensor measurements that are different from the estimate is small. E.g., when as sensor network is deployed to measure temperature at various locations, most values will be close to the temperature estimate obtained from either previous observations or other methods, while a small number of values are significantly different from the estimate. Then in a sensor network, we can use a sparse vector to represent the deviations of sensor values to a given estimate.

Let $S \subseteq V$ denote a subset of nodes in $G$. Let $E_S$ %=\{(u, v) \in
%E~|~ u \in S, v\in S\}$
denote the subset of edges with both ends in
$S$, then $G_S=(S, E_S)$ is the induced subgraph of $G$. We have the following two assumptions on graph topological constraints:\\
\noindent {\bf (A1)}: A set $S$ of nodes can be measured together in one measurement if and only if  $G_S$ is connected.

\noindent{\bf (A2)}:
 The measurement is an additive sum of values at the corresponding nodes.

  %For example, nodes 1, 2, 3, 5, and 6 in Fig. \ref{fig:example} can be measured together by one measurement, and the measurement $y$  satisfies $y=x_1+x_2+x_3+x_5+x_6$.  Take a practical example, measurements in the sensor networks where the nodes represent sensors and the edges represent feasible communication between sensors satisfy  the above assumptions. Since the set of nodes $S$ to be measured induces a connected subgraph, then we can find a sub spanning tree $T$ only containing all the nodes in $S$. The root of $T$, say node $u$, monitors the values of nodes in $S$. Every leaf node in $T$ passes its value to its parent on $T$. Every intermediate node in $T$ sums up the values obtained from its children, adds its own value and passes it to its parent. Root $u$ sums up the obtained values and adds its own value, and the measurement at root $u$ is indeed the sum of values corresponding all the nodes in $S$.

%\begin{figure}[h]
%\begin{center}
%  \includegraphics[scale=0.22]{./Figures/example1}
%  %\includegraphics[scale=0.25]{./Figures/dIdv}
%  \caption{Network Example}\label{fig:example}
%  \end{center}
%\end{figure}
%Note that %in Compressed Sensing, the measurement matrix $A$ can be any real matrix, however,

 (A2) follows from the additive property of many network characteristics\footnote{Compressed sensing can also be applied to cases where (A2) does not hold, e.g., the measurements can be nonlinear as in \cite{WWGZMM11,Blumensath10}.}, e.g. delays and packet loss rates \cite{GR12}. (A1) captures the topological constraints.
 In link delay monitoring problem where $G$ corresponds to the line graph of a communication network, (A1) is equivalent to  that the set of communication links  that correspond to nodes in $S$ should be connected in the communication network $\mathcal{N}_G$. If (A1) is satisfied, one can find a cycle that traverses each link in this set exactly twice (one for each direction).  One router in this cycle %can be chosen as an ``agent'' and
 sends a packet along this cycle and measures
 the  total transmission delay experienced by this packet.  This total delay is twice the sum of average delays on this set of links, and an average delay of a link is the average of its delays in both directions.    For example, Fig. \ref{fig:trans} shows   the correspondence between  assumptions (A1) (A2) in the line graph model and the monitoring in the original network.   % In sensor networks, a measurement obeying (A1) and (A2) is obtained  as follows. For a set $S$ of nodes that induce a connected subgraph, one sensor node $u$ in $S$ is chosen as an ``agent'' to monitor the sum of node values in $S$. Every node in $S$ obtains values from its children, if any, on the spanning tree rooted at $u$, sums them up with its own value and sends the sum to its parent. Then node $u$ will obtain the sum of node values in $S$.

Let vector $\bfy \in \mathcal{R}^m$   denote $m$ measurements satisfying (A1) and (A2).  $A$ is  an $m \times n$ measurement matrix %with its $i$th row corresponding to the $i$th measurement, i.e.,
with $A_{ij}=1$ ($i=1,...,m$, $j=1,...,n$) if and only if node $j$ is included in the $i$th measurement and $A_{ij}=0$ otherwise. We can write it in the compact form that  $\bfy=A\bfx$.
With the requirements  (A1) and (A2), $A$ must be a $0$-$1$ matrix, and for each row of $A$, the set of nodes that correspond to `1' must form a connected induced subgraph of $G$. For the graph in Fig. \ref{fig:example}, we can measure the sum of nodes in $S_1$ and $S_2$ by two separate measurements, and the %corresponding
measurement matrix is
\begin{equation}\nonumber
A=\left[ \begin{array}{cccccccc}
1 &  1 &1 &0 & 1 & 1 &0 & 0\\
0 & 0& 1 & 1 & 0 & 0& 1 & 1
\end{array}
 \right].
\end{equation}

  \begin{figure*}[ht]
\begin{center}
\begin{minipage}{4.6in}
%\begin{figure}
\begin{center}
  \includegraphics[scale=0.45]{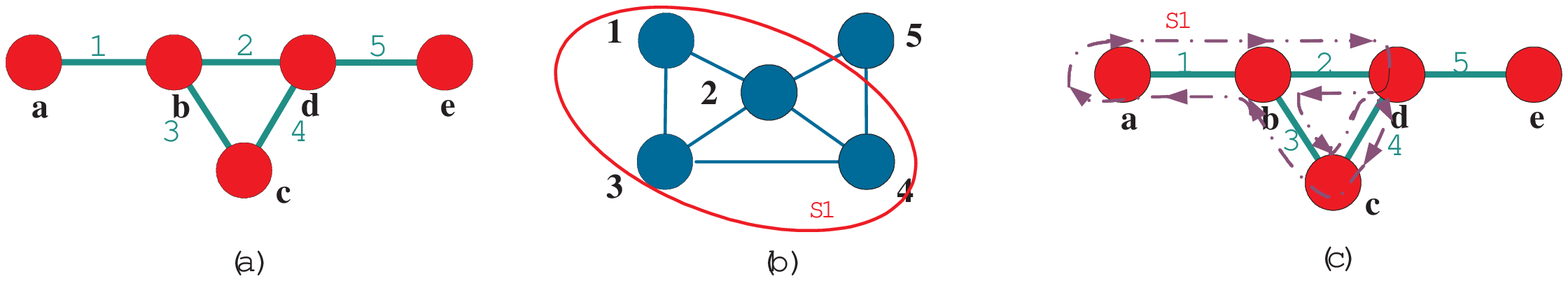}
  \caption{(a) Network $\mathcal{N}_G$ with five links, (b) Its line graph $L(\mathcal{N}_G)$ that we consider in this paper.  Since the links 1, 2, 3, and 4 are connected in  $\mathcal{N}_G$,  the induced subgraph of nodes 1, 2, 3, and 4 in  $L(\mathcal{N}_G)$  is connected. (c) Since the induced subgraph of nodes 1, 2, 3, and 4 is connected, one can find a cycle passing each of the corresponding links in network $\mathcal{N}_G$ exactly twice. }\label{fig:trans}
  \end{center}
\end{minipage}
\begin{minipage}{1.8in}
\begin{center}
  \includegraphics[scale=0.28]{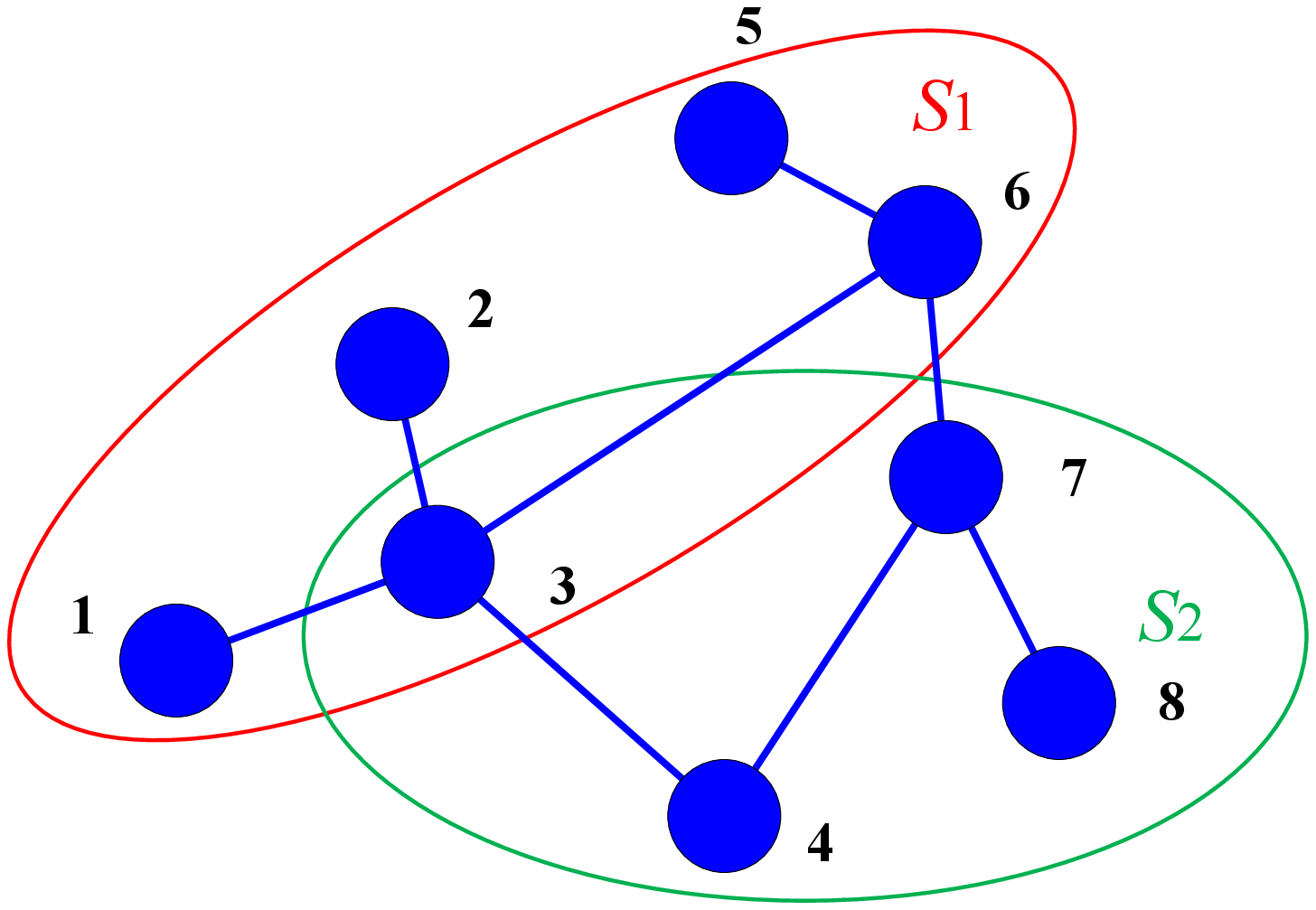}
  \caption{Graph example}\label{fig:example}
%\end{center}
%\end{figure}
\end{center}
\end{minipage}
%\hfill
\end{center}
\end{figure*}

 We say a measurement matrix $A$ can \emph{\textbf{identify all $k$-sparse vectors }}if and only if $A\bfx_1 \neq A\bfx_2$ for every two different vectors $\bfx_1$ and $\bfx_2$ that are \emph{\textbf{at most $k$-sparse}}. This definition indicates that every $k$-sparse vector $\bfx$  is the unique solution to the following $\ell_0$-minimization problem
 \begin{equation}\label{eqn:ell0}
 \min_{\bfz} \|\bfz\|_0  \quad \textrm{s.t. } A\bfz=A\bfx.
  \end{equation}
 Note (\ref{eqn:ell0}) is a combinatorial problem in general.
%   can be recovered from $A\bfx$ via $\ell_0$-minimization, which returns the sparsest vector among all vectors that can produce the same observation $A\bfx$.  %Sparse recovery theory promises that one can identify $n$-dimensional sparse vectors from $m$ measurements even if  $m \ll n$. %  provided that the vectors are sparse enough.
% The advantage of sparse recovery is that with the non-adaptive measurement matrix $A$, it can identify $n$-dimensional vectors from $m$ ($m \ll n$) measurements as long as the vectors are sparse.
  %If we take two measurements of nodes $\{1, 2,3, 5,6\}$ and nodes $\{3,4,7,8\}$

Then, given topological constraints represented by $G$, we want to design non-adaptive measurements satisfying (A1) and (A2) such that one can identify all $k$-sparse vector $\bfx$, and the total number of measurements is minimized.
Given a graph $G$ with $n$ nodes, let $M^G_{k,n}$ denote the minimum number of  measurements satisfying (A1) and (A2)  to identify all $k$-sparse vectors associated with $G$. The questions we would like to address in the paper are:
\begin{itemize}
\item Given  $G$, what is the corresponding $M_{k,n}^G$? What is the dependence of $M_{k,n}^G$ on $G$?
\item How can we explicitly design measurements such that the total number of measurements is close to $M_{k,n}^G$?
\end{itemize}

Though motivated by network applications, we use   graph $G$  to characterize the topological constraints and study a general problem of recovering sparse signals   from measurements satisfying graph constraints. For the majority of this paper, we assume a measurement is feasible as long as (A1) and (A2) are satisfied, and we attempt to minimize the total number of   measurements for identifying sparse signals.  Some additional constraints on the measurements such as bounded measurement length will be discussed %the such that a measurement only measures a small number of nodes
in Section \ref{sec:extension}.

%One important question to address in compressed sensing is that how many measurements are needed in order to identify every $k$-sparse vector $\bfx$ for given $k$? Or equivalently, how to design the measurement matrix $A$ such that for every two different vectors $\bfx_1$ and $\bfx_2$ that are at most $k$-sparse, at least one measurement is different, and $m$ is as small as possible? There is an explosion of research on this topic, however, very few papers consider graph constraints. Here we want to design the measurement matrix $A$ satisfying graph constraints such that the number of measurements needed to identify all $k$-sparse vectors is as small as possible.

%Given a graph $G$ with $n$ nodes, let $M^G_{k,n}$ denote the minimum number of non-adaptive measurements needed  to identify all $k$-sparse vectors associated with $G$. Since
If $G$ is a complete graph, then any subset of nodes  forms a connected subgraph, %then the graph constraints can be neglected and
and every $0$-$1$ matrix is a feasible measurement matrix. Then the problem reduces to the conventional compressed sensing where one wants to identify sparse signals from linear measurements.  %From compressed sensing theory \cite{CaT05}, the measurement matrix $A$ can identify all $k$-sparse vectors if and only if $A\bfx_1 \neq A\bfx_2$ for every two different $k$-sparse vectors $\bfx_1$ and $\bfx_2$, and equivalently, every non-zero vector $\bfz$ satisfying $A\bfz=0$ should have at least $2k+1$ non-zero entries.
Existing results \cite{CaT06,BGIKS08,XH07} show that with overwhelming probability a random $0$-$1$ $A$ matrix with $O( k\log (n/k))$ rows\footnote{We use the notations  $g(n)\in O(h(n))$, $g(n) \in \Omega(h(n))$, or $g(n)=\Theta(h(n))$ if as $n$ goes to infinity, $g(n) \leq c h(n)$, $g(n) \geq c h(n)$ or $c_1 h(n) \leq g(n) \leq c_2 h(n)$ eventually holds for some positive constants $c$, $c_1$ and $c_2$ respectively.}
can identify all $k$-sparse vectors $\bfx$ associated with a complete graph, and $\bfx$ is the unique solution to the $\ell_1$-minimization problem
\begin{equation}
\label{eqn:ell1}
\min_{\bfz} \|\bfz\|_1 \quad \textrm{s.t. } A\bfz =A \bfx.
\end{equation}
 (\ref{eqn:ell1}) can be recast as a linear program, and thus it is computationally more efficient to solve (\ref{eqn:ell1}) than  (\ref{eqn:ell0}). Thus, we have
% and we can recover the sparse vector by $\ell_1$-minimization, which returns the vector with the least $\ell_1$-norm among those that can produce the obtained measurements. Let $M^C_{k,n}$ denote the minimum number of non-adaptive measurements needed in a complete graph (i.e., without graph constraints) with $n$ nodes, then %we have
\begin{equation}\label{eqn:MC}
M^C_{k,n} = O( k\log (n/k)).
 \end{equation}
Note that $ O( k\log (n/k)) \ll n$ for   $k \ll n$, thus,   the number of measurements can be significantly reduced for sparse signals.
 %We  use (\ref{eqn:MC}) for the analysis of construction methods.
 Explicit constructions of measurement matrices for complete graphs also exist, e.g., \cite{AHSC09,BGIKS08,CM06,DeVore07,XH07}.  %provide explicit constructions of the measurement matrices to identify $k$-sparse vectors in complete graphs,
%but the number of measurements by construction in general is much greater than $O( k\log (n/k))$.
 We  use $f(k,n)$ to denote the number of measurements to recover $k$-sparse vectors associated with a complete graph of $n$ nodes by a particular measurement construction method. $f(k,n)$ varies for different construction methods, and clearly $f(k,n) \geq M^C_{k, n}$. Table \ref{tbl:notation} summarizes the key notations. % are summarized in

 \begin{table}
\caption{summary of key notations} \label{tbl:notation}
\begin{tabular}{|c|p{7.1cm}|} %
 \hline
Notation & Meaning\\
\hline
%\normalsize
$G_S$ & Subgraph of $G$ induced by $S$\\
\hline
%\normalsize
$M^G_{k,n}$ & Minimum number of measurements needed to identify  $k$-sparse vectors associated with $G$ of $n$ nodes.\\
\hline
%\normalsize
$M^C_{k,n}$ & Minimum number of  measurements needed to identify  $k$-sparse vectors associated with a complete graph of $n$ nodes.\\
\hline
%\normalsize
$f(k,n)$ & Number of measurements constructed to identify $k$-sparse vectors associated with a complete graph of $n$ nodes\\
%\hline
\hline
\end{tabular}
\end{table}

For a general graph $G$ that is not complete, existing results do not hold any more. Can we still achieve a significant reduction in the number of measurements? This is the focus of this paper. We remark here that in group testing with graph constraints, the requirements for the measurement matrix $A$ are the same, while group testing differs from compressed sensing only in that (1) $\bfx$ is a logical vector, and (2) the operations used in each group testing measurement are the logical ``AND'' and ``OR''. % each measurement is a logical summation in group testing.
     Here we consider  compressed sensing if not otherwise specified, and the main results are stated in theorems. We sometimes discuss group testing for comparison, and the results are stated in propositions. %sometimes provide results in group testing for comparison.
Note that for recovering $1$-sparse vectors, %the measurements of compressed sensing and group testing are essentially the same.
the numbers of measurements required by compressed sensing and group testing are the same.

%% file: special_arxiv.tex
\section{Sparse Recovery over Special Graphs}\label{sec:special}

In this section, we consider four kinds of special graphs: one-dimensional line/ring, ring with each node connecting to its four closest neighbors, two-dimensional grid and a tree.
 The measurement construction method for a line/ring is different from those for the other graphs, %and we discuss them separately. And
 and our construction is optimal (or near optimal) for a line (or ring). For other special graphs, we construct measurements based on the ``hub'' idea and will later extend it  to general graphs in Section \ref{sec:general}.

\subsection{Line and Ring}\label{sec:line}

First consider a line/ring as shown in Fig. \ref{fig:line}. Note that a line/ring is the  line graph of a line/ring network. When later comparing the results here with those in Section \ref{sec:ring4}, one can see that the number of measurements required for sparse recovery can be significantly different in two graphs that only differ from each other with a small number of edges.

In a line/ring, %there is not much freedom in the measurement design since
only consecutive nodes can be measured together from  (A1).  Recovering 1-sparse vectors associated with a line (or ring)  with $n$ nodes is considered in \cite{HPWYC07,TWHR11}, which shows that $\lceil\frac{n+1}{2}\rceil$ (or $\lceil\frac{n}{2}\rceil$) measurements
are both necessary and sufficient  in this case. Here, we consider recovering $k$-sparse vectors for $k\geq2$. %associated with the line/ring. %to recover  $1$-sparse vectors associated with a line (or ring) network with $n$ nodes. Therefore, $\Theta(n)$ measurements are required to recover even one non-zero element associated with a line/ring.
\begin{figure}
\begin{center}
  % Requires \usepackage{graphicx}\
    \includegraphics[scale=0.25]{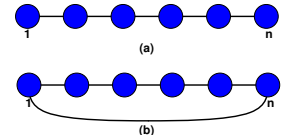}
    \caption{(a) line (b) ring }\label{fig:line}
\end{center}
\end{figure}
%The following result \cite{HPWYC07,TWHR11} is known for the line/ring.
%\begin{prop}\label{thm:line1}
% To recover a $1$-sparse vector associated with a line (ring) network with $n$ nodes, $\lceil\frac{n+1}{2}\rceil$ ($\lceil\frac{n}{2}\rceil$) measurements
%are both necessary and sufficient respectively.
%\end{prop}

%We next consider recovering $k$-sparse vectors ($k\geq2$) associated with the line/ring. We construct $n+1-\lfloor\frac{n+1}{k+1}\rfloor$ measurements as follows. Let $t=\lfloor \frac{n+1}{k+1}\rfloor$.
Our construction %of measurements for recovering $k$-sparse vectors
works as follows. Given $k$ and $n$, let $t=\lfloor \frac{n+1}{k+1}\rfloor$. We construct $n+1-\lfloor\frac{n+1}{k+1}\rfloor$ measurements with the $i$th measurement passing all the nodes from $i$ to $i+t-1$. Let $A^{(n+1-t)\times n}$ be the   measurement matrix, then its $i$th row has `1's from  entry $i$ to entry $i+t-1$ and `0's elsewhere. For example, when $k=3$ and $n=11$, we have $t=3$, and %the measurement matrix
\begin{equation}\label{eqn:Aex}
 {\small A= \left[ \begin{array}{ccccccccccc}  1 & 1 & 1 & 0& 0& 0 & 0& 0& 0 & 0& 0\\0 &1 & 1 & 1 & 0& 0& 0 & 0& 0& 0 & 0\\ 0 &0&1 & 1 & 1 & 0& 0& 0 & 0& 0& 0\\ 0 & 0 &0&1 & 1 & 1 & 0& 0& 0 & 0& 0 \\0 & 0 &0 & 0&1 & 1 & 1 & 0& 0& 0 & 0\\0 &0 & 0 &0 & 0&1 & 1 & 1 & 0& 0& 0 \\0& 0 &0 & 0 &0 & 0&1 & 1 & 1 & 0& 0\\0 &0& 0 &0 & 0 &0 & 0&1 & 1 & 1 & 0\\0 &0 &0& 0 &0 & 0 &0 & 0&1 & 1 & 1 \end{array}  \right].  } \end{equation}

 %We shall show that such measurement construction can recover $k$-sparse vectors  associated with the line/ring. Moreover,
 Let $M^L_{k,n}$ and $M^R_{k,n}$ denote the minimum number of measurements required to recover $k$-sparse vectors  in a line/ring respectively.
% Therefore we have, %, and  is also the minimum number of measurements needed to recover $k$-sparse vectors associated with a line network, and is less than the minimum number of measurements needed to recover $k$-sparse vectors associated with a line network% $\min (i+t-1, n)$. %for $1\leq i\leq kt$, and the last measurement goes through nodes $kt+1$ to $n$.
%Let $M^L_{k,n}$ and $M^R_{k,n}$ denote the minimum number of measurements required in a line/ring respectively.
 We have

\begin{theorem}\label{thm:linekcs}
%To recover up to $k \geq 2$ errors in a line network with $n$
%network nodes,
Our constructed $n+1-\lfloor\frac{n+1}{k+1}\rfloor$  measurements  can identify all $k$-sparse vectors associated with a line/ring with $n$ nodes. Moreover, the sparse signals can be recovered from $\ell_1$-minimization (\ref{eqn:ell1}). Furthermore, it holds that %This number is the minimum number of measurements needed to recover $k$-sparse vectors  in a line network, and also differs from  the minimum number of measurements needed to recover $k$-sparse vectors  in a ring network by at most 1, i.e., % (\ref{eqn:tight}) holds.%Therefore $M^L_{k,n} \leq k\lceil \frac{n}{k+1}\rceil+1$ and $M^R_{k,n} \leq k\lceil \frac{n}{k+1}\rceil+1$.
\begin{equation}\label{eqn:tight}
M^L_{k,n}=n+1-\lfloor\frac{n+1}{k+1}\rfloor \leq M^R_{k,n}+1.
 \end{equation}
\end{theorem}

In our early work \cite{WXMT12} (Theorem 1), we have proved that our constructed $n+1-\lfloor\frac{n+1}{k+1}\rfloor$ measurements can identify $k$-sparse signals, which indicate that one can recover the signal via solving $\ell_0$-minimization (\ref{eqn:ell0}). (\ref{eqn:ell0})  is in general  computationally inefficient to solve.  Here we further show that with these measurements, one can recover the signal via solving a computationally efficient  $\ell_1$-minimization (\ref{eqn:ell1}). %$\ell_1$-minimization can be recast as a linear program, and thus solved efficiently.

Furthermore, (\ref{eqn:tight}) indicates that our construction is optimal for a line   in the sense that the number of measurements is equal to the minimum needed to recover $k$-sparse vectors. %no other method can identify associated $k$-sparse vectors with a smaller number of measurements. Similarly,
%It is  near optimal
For a ring, this number  is no more than the minimum   plus one. This improves over our previous result, which does not have optimality guarantee.

\begin{proof}
(of Theorem \ref{thm:linekcs})
\textit{We first prove that one can recover $k$-sparse signals from our constructed $n+1-\lfloor\frac{n+1}{k+1}\rfloor$ measurements via $\ell_1$-minimization (\ref{eqn:ell1}).}

Let $A$ be the measurement matrix.  When $t=1$, $A$ is the identity matrix, and the statement holds trivially. So we only consider the case $t \geq 2$. %Note that in this case $n \geq kt+t-1$.
It is well known in compressed sensing (see e.g., \cite{EB02}) that a $k$-sparse vector $\bfx$ can be recovered from $\ell_1$-minimization, i.e., it is the unique solution to (\ref{eqn:ell1}), if and only if for every vector $\bfw \neq \bm 0$ such that $A\bfw=\bm 0$, and for every set $T \subseteq \{1,...,n\}$ with $|T| \leq k$, it holds that
\begin{equation}\label{eqn:nullspace}
\|\bfw_{T}\|_1 < \|\bfw\|_1 /2,
\end{equation}
where $\bfw_T$ is a subvector of $\bfw$ with entry indices in $T$. Thus, we only need to prove that (\ref{eqn:nullspace}) holds for our constructed $A$.

From the construction of $A$, one can  check that for every $\bfw \neq \bm 0$ such that $A\bfw=\bm 0$, and for every $j \in \{1,...,n\}$, %it holds that
\begin{equation}\label{eqn:equal}
w_j=w_{j-\lfloor\frac{j}{t}\rfloor t}
\end{equation}
holds. For example,
\begin{equation}\nonumber
w_1=w_{t+1}=w_{2t+1}=\cdots =w_{(k-1)t+1}.
\end{equation}
%and
%\begin{equation}\nonumber
%w_t=w_{2t}=\cdots=w_{kt},
%\end{equation}

Let $w^* := \arg \max_{j=1}^{t} |w_j|$. From (\ref{eqn:equal}), it also holds that
\begin{equation}\label{eqn:wstar}
w^* =\arg \max_{j=1}^n |w_j|.
\end{equation}
From  the first row of $A$, we have
\begin{equation}\label{eqn:one}
\sum_{i=1}^{t} w_i =0,
\end{equation}
From the definition of $w^*$, (\ref{eqn:one}) implies
\begin{equation}\label{eqn:half}
w^* \leq \frac{1}{2}\sum_{j=1}^t |w_j|.
\end{equation}

 Since $n \geq kt+t-1$ from the definition of $t$,  we have
 \begin{equation}\label{eqn:positive}
\sum_{j=kt+1}^{n} |w_j| \geq \sum_{j=kt+1}^{kt+t-1} |w_j|=\sum_{j=1}^{t-1} |w_j| >0,
 \end{equation}
 where the equality follows from (\ref{eqn:equal}). The last inequality holds since $w_j \neq 0$ for at least one $j$ in $1,...,t-1$. Suppose  $w_j=0$ for all $j=1,..., t-1$, then $w_t=0$ from (\ref{eqn:one}), which then leads to $\bfw=\bm 0$ through (\ref{eqn:equal}), contradicting the fact that $\bfw\neq \bm 0$.

Now consider any $T$ with $|T|\leq k$, combining (\ref{eqn:equal}), (\ref{eqn:wstar}), (\ref{eqn:half}),  and (\ref{eqn:positive}), we have
\begin{equation}\label{eqn:terms}\nonumber
%\|\bfw_T\|_1 & \leq & k w^* \label{eqn:wstardf}\\
%& \leq & \frac{k}{2}\sum_{j=1}^t |w_j| \label{eqn:wstardf2}\\
%& =& \frac{1}{2}\sum_{j=1}^{kt} |w_j| \label{eqn:ext}\\
%<\|\bfw\|_1/2,\label{eqn:terms}
\|\bfw_T\|_1  \leq  k w^*  \leq  \frac{k}{2}\sum_{j=1}^t |w_j|
 = \frac{1}{2}\sum_{j=1}^{kt} |w_j|
<\|\bfw\|_1/2.
\end{equation}
%where (\ref{eqn:wstardf}) follows from (\ref{eqn:wstar}), (\ref{eqn:wstardf2}) follows from (\ref{eqn:half}), (\ref{eqn:ext}) follows from (\ref{eqn:equal}).
Thus, $\bfx$ can be correctly recovered via $\ell_1$-minimization (\ref{eqn:ell1}).

\textit{We next prove that the number of measurements needed to identify $k$-sparse vectors associated with a line (or ring)  is at least $ n+1-\lfloor \frac{n+1}{k+1}\rfloor$ (or  $n-\lfloor \frac{n}{k+1}\rfloor$.)}

%We start with the line network with $n$ nodes.
 Let $A^{m \times n}$ denote a measurement matrix with which one can recover $k$-sparse vectors associated with a line of $n$ nodes. Then every $2k$ columns of $A$ must be linearly independent. %Let $m$ denote the number of rows in $A$, i.e., $m$ is the number of measurements.
 We will prove that $m \geq n+1-\lfloor \frac{n+1}{k+1}\rfloor$.

Let $\bm \alpha^i$ denote the $i$th column of $A$. Define $\bm \beta^1=\bm \alpha^1$, $\bm \beta ^i= \bm \alpha^{i}-\bm \alpha^{i-1}$ for all $2 \leq i \leq n$, and $\bm \beta^{n+1}=-\bm \alpha^n$. Define matrix $P^{m\times (n+1)}=(\bm \beta ^i, 1\leq i \leq n+1)$. Since $A$ is a measurement matrix for a line  network, %$P$ only has $`0'$, $`1'$ and $`-1'$ entries. Moreover,
each row of $P$ contains  one $`1'$ entry and   one $`-1'$ entry, and all the other entries must be $`0'$s.

Given $P$, we construct a graph $G_{eq}$ with $n+1$ nodes as follows. For every row $i$ of $P$, %if $P_{ij}=1$ and $P_{ik}=-1$ for some integers $j$ and $k$ between $1$ and $n$, then
there is an edge $(j, k)$ in $G_{eq}$, where $P_{ij}=1$ and $P_{ik}=-1$. Then $G_{eq}$ contains $m$ edges, and $P$ can be viewed as the transpose of an  oriented incidence matrix of $G_{eq}$. Let $S$ denote the set of indices of nodes in a component of $G_{eq}$, then one can check that %the sum of the columns of $P$ corresponding to all the indices in $S$ is a zero vector.
 \begin{equation}\label{eqn:comgeq} \sum_{ i\in S} \bm \beta^i= \bm 0.
 \end{equation}
  Since every $2k$ columns of $A$ are linearly independent, every $k$ columns of $P$ are linearly independent, which then implies that the sum of any $k$ columns of $P$ is not a zero vector.
 %Since the sum of any $k$ columns of $P$ is a non-zero vector, then
 With (\ref{eqn:comgeq}), we know that any component of $G_{eq}$ should have at least $k+1$ nodes. %Thus, $G_E$ has at most $\lfloor \frac{n+1}{k+1}\rfloor$ components.
Since a component with $r$ nodes contains at least $r-1$ edges, and $G_{eq}$ has at most $\lfloor \frac{n+1}{k+1}\rfloor$ components, then $G_{eq}$ contains at least $n+1-\lfloor \frac{n+1}{k+1}\rfloor$ edges. The claim follows.

We next consider the ring.  Let $\tilde{A}$ denote the measurement matrix with which one can recover $k$-sparse vectors on a ring  with $n$ nodes. Let $\bm \tilde{\alpha}^i$ denote the $i$th column of $\tilde{A}$. Define $\bm \tilde{\beta} ^1= \bm \tilde{\alpha}^{1}-\bm \tilde{\alpha}^n$, and $\bm \tilde{\beta} ^i= \bm \tilde{\alpha}^{i}-\bm \tilde{\alpha}^{i-1}$ for all $2 \leq i \leq n$. Define matrix $\tilde{P}^{m\times n}=(\bm \tilde{ \beta} ^i, 1\leq i \leq n)$. %Similarly one can argue that the sum of any $k$ columns of $\tilde{P}$ is a non-zero vector.
Similarly, we construct a graph $\tilde{G}_{eq}$ with $n$ nodes based on $\tilde{P}$, and each component of $\tilde{G}_{eq}$ should have at least $k+1$ nodes. Thus, $\tilde{G}_{eq}$ contains at most $\lfloor \frac{n}{k+1}\rfloor$ components and  at least $n-\lfloor \frac{n}{k+1}\rfloor$ edges. Then
\begin{equation}\nonumber
M^R_{k,n} \geq n-\lfloor \frac{n}{k+1}\rfloor \geq n-\lfloor \frac{n+1}{k+1}\rfloor,
\end{equation}
and the inequality of (\ref{eqn:tight}) holds.
\end{proof}

We can save about $\lfloor \frac{n+1}{k+1}\rfloor-1$ measurements but
still be able to recover $k$-sparse vectors in a line/ring
via compressed sensing. But for group testing on a line/ring,  $n$ measurements are necessary to recover more than one non-zero element. The key is that every node should be the \textit{endpoint} at least twice, where the endpoints are the nodes at the beginning and the end of a measurement. %The endpoints of a measurement can be a same node.
If node $u$ is an endpoint for at most once, then it is always measured together with one of its neighbors, say $v$, if ever measured. %or always measured together with at least one of its neighbors, say node $v$.
Then when $v$ is `1', we cannot determine the value of $u$, either '1' or '0'. %If $u$ is an endpoint for only once, then the particular measurement also measures one of its neighbors, and we still cannot determine the value of $i$ if the neighbor is `1'.  or both $i$ and  of a measurement thus there is no saving in the number of measurements with group testing when $k>1$. %
Therefore, to recover more than one non-zero element, we need at least $2n$ endpoints, and thus $n$ measurements. %Moreover, we have the following proposition. We skip the proof here and  please refer to \cite{} for details.

\subsection{Ring with nodes connecting to four closest neighbors}\label{sec:ring4}
% We know from Section \ref{sec:line} that
%$n-\lfloor \frac{n}{k+1}\rfloor$ measurements are necessary to recover $k$ non-zero element associated with a ring
%network. Now
Consider a graph with each node directly connecting to its
four closest neighbors as in Fig. \ref{fig:ring4} (a), denoted by $\mathcal{G}^4$. $\mathcal{G}^4$ is important to the study of small-world networks \cite{WS98}.
$\mathcal{G}^4$ has $n$ more edges than the ring, but we will show
that %by only adding $n$ edges to the ring network
the number of
measurements required by compressed sensing to recover $k$-sparse vectors associated with $\mathcal{G}^4$ significantly reduces from $\Theta(n)$ to $O(k \log(n/k))$.
The main idea in the measurement construction is referred to as `` the use of a hub''.

\begin{figure}[h]
\begin{center}
\includegraphics[scale=0.25]{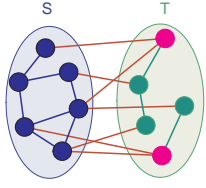}
  \caption{Hub $S$ for $T$}\label{fig:hub}
  \end{center}
\end{figure}

\begin{defi}
Given $G=(V,E)$, $S \subseteq V$ is a \textbf{\textit{hub}} for  $T\subseteq V$ if $G_S$ is connected, and $\forall u \in T$, $\exists s \in S$ s.t. $(u,s) \in E$.
\end{defi}

We first take one measurement of the sum of nodes in $S$, denoted by $s$. For any subset $W$ of $T$, e.g., the pink nodes in  Fig. \ref{fig:hub}, $S\cup W$ induces a connected subgraph from the hub definition and thus can be measured by one measurement. To measure the sum of nodes in $W$, we first measure nodes in $S \cup W$ and then subtract $s$ from the sum. Therefore we can apply the measurement constructions for complete graphs on $T$ with this simple modification, and that requires only one additional measurement for the hub $S$. Thus,

\begin{theorem}\label{thm:hub}
With hub $S$, $M^C_{k,|T|}+1$ measurements are enough to recover $k$-sparse vectors associated with $T$.
%$M^C_{k, \lfloor n/2 \rfloor }+1$ measurements are enough to recover $\bfx_{\textrm{e}}\in \mathcal{R}^{\lfloor n/2 \rfloor}$, where the additional one measurement measures $s_{\textrm{o}}$.
\end{theorem}

The significance of Theorem \ref{thm:hub} is that $G_T$  is not necessarily a complete subgraph, i.e., a clique, and it can even be disconnected. As long as there exists a hub $S$, %the number of measurements needed to recover sparse vectors associated with $T$ is close to that needed on a complete graph with the same number of nodes,
the measurement construction for a complete graph with the same number of nodes can be applied to $T$ with   simple modification. % discussed above, and we only need one additional measurement for the hub.
Our later results rely heavily on Theorem \ref{thm:hub}. % we can still  measure the sum of any subset of nodes in $T$ with the help of the hub $S$.
%Theorem \ref{thm:hub} does not require any
%Throughout the paper, given a graph $G=(V,E)$, we say  $S$ forms a \textbf{\textit{hub}} for $U$ if $G_S$ is connected, and for every $u$ in $U$, there exists $s$ in $S$ such that $(u,s) \in E$.

\begin{figure}
\centering
\begin{tabular}{c c}
%\includegraphics[width=0.5\linewidth]{./Figures/enum3}
%\includegraphics[width=0.4\linewidth]{./Figures/ring4}
%&
%\includegraphics[width=0.2\linewidth]{./Figures/ring4_2}
%&
%\includegraphics[width=0.5\linewidth]{./Figures/enum2}
\includegraphics[width=0.4\linewidth, height=0.3\linewidth]{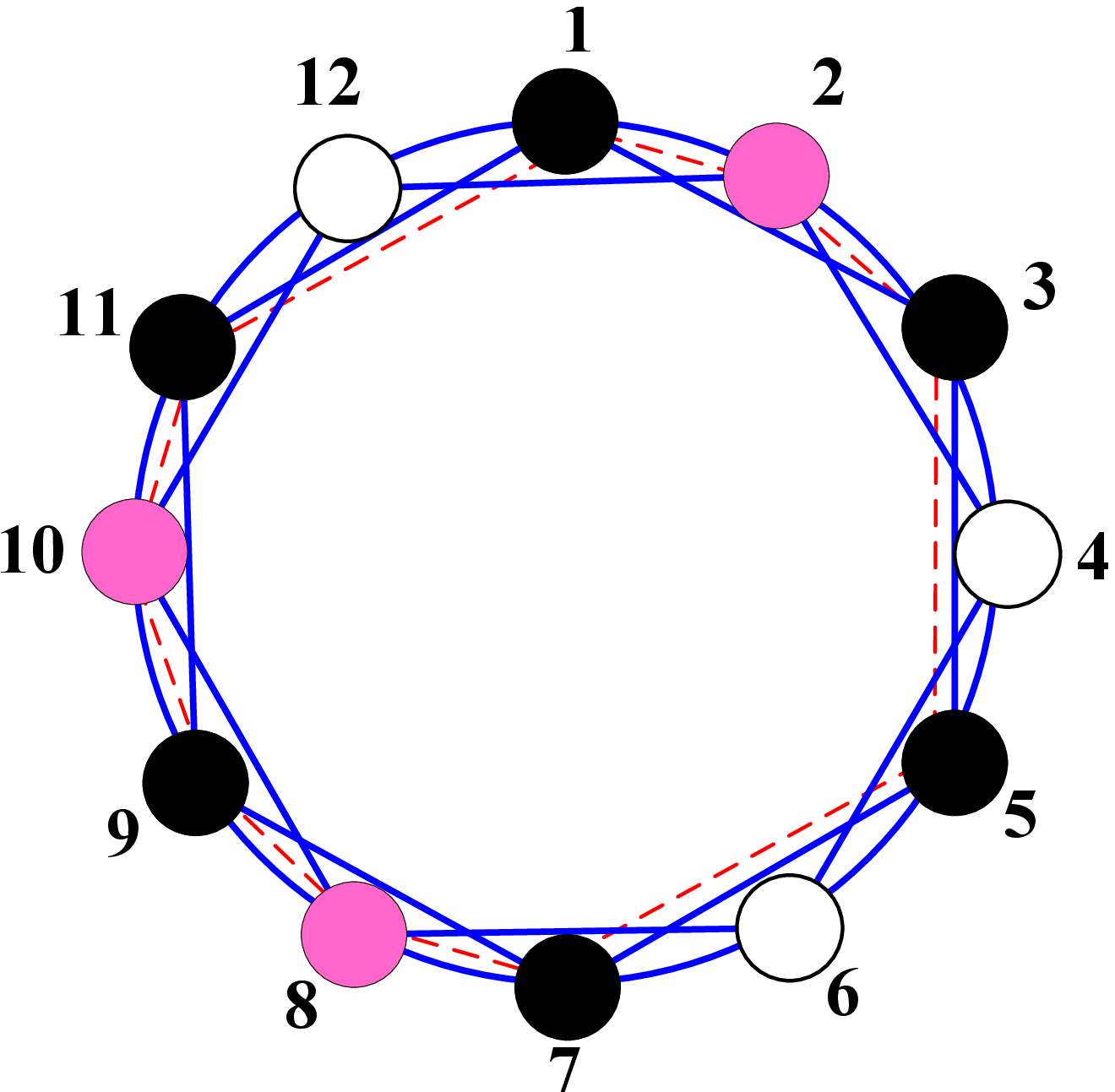}
&
\includegraphics[width=0.4\linewidth, height=0.3\linewidth]{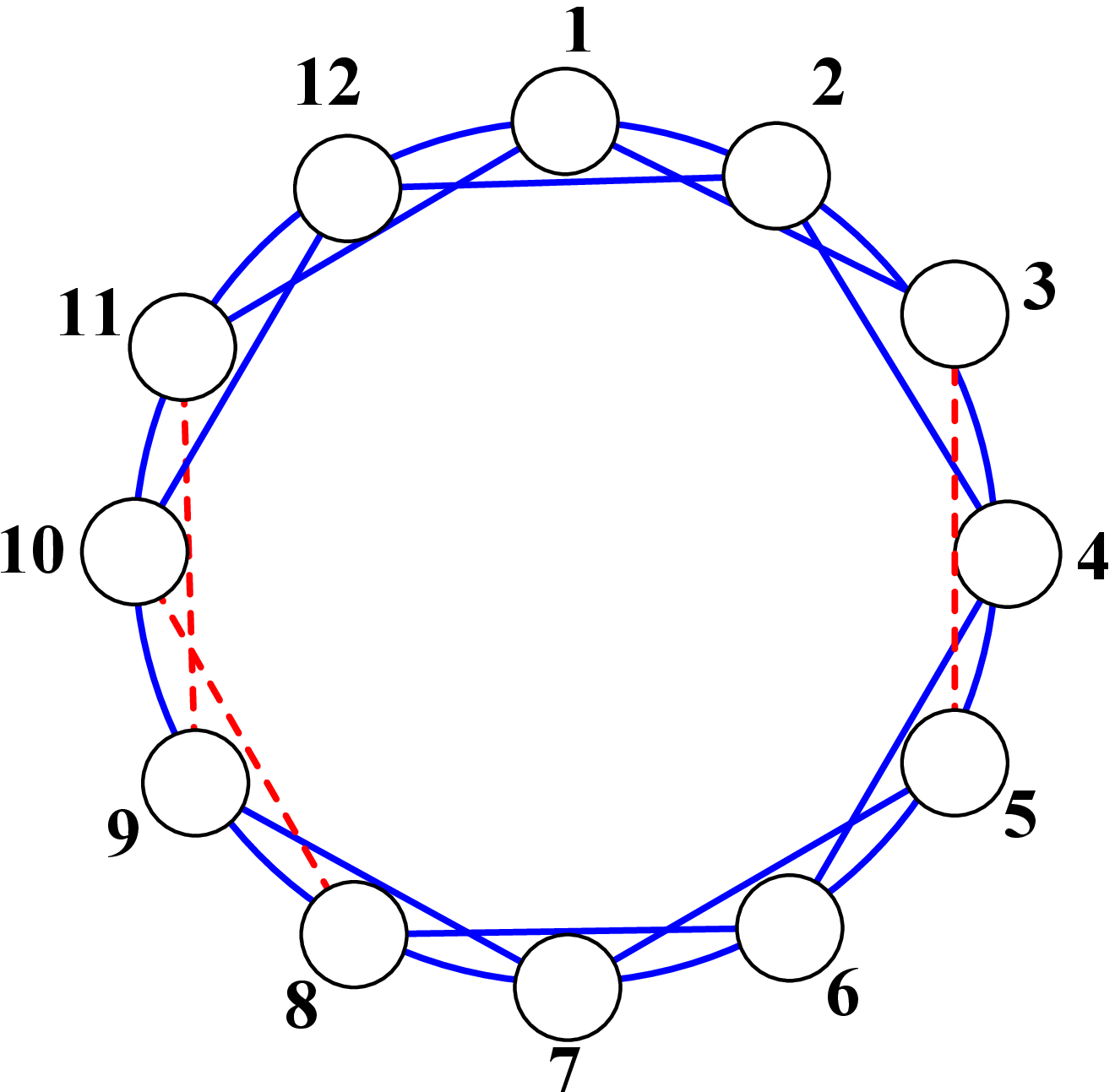}
\\
{\scriptsize (a) %Topology of $\mathcal{G}^4$} & {\scriptsize (b)
Measure nodes 2,8 and 10 via hub $T_{\textrm{o}}$}& {\scriptsize (b) Delete $h$ long links} \\
\end{tabular}
      \caption{Sparse recovery on graph $\mathcal{G}^4$}
      \label{fig:ring4}
   \end{figure}
In $\mathcal{G}^4$, if nodes are numbered consecutively around the ring, then the set of all the odd nodes, denoted by $T_{\textrm{o}}$, form a hub for the set of all the even nodes, denoted by $T_{\textrm{e}}$. Given a $k$-sparse vector $\bfx$, let $\bfx_{\textrm{o}}$ and $\bfx_{\textrm{e}}$ denote the subvectors of $\bfx$ with odd and even indices. Then $\bfx_{\textrm{o}}$ and $\bfx_{\textrm{e}}$ are both at most $k$-sparse. %The sum of entries in $\bfx_{\textrm{o}}$, denoted by $s_{\textrm{o}}$, can be obtained by one measurement, and similarly for the sum $s_{\textrm{e}}$ of the entries of $\bfx_{\textrm{e}}$. For any subset $W$ of $T_{\textrm{e}}$, $T_{\textrm{o}} \cup W$ induces a connected subgraph and thus can be measured by one measurement. We can obtain the sum of values corresponding to  nodes in $W$ by measuring nodes in $T_{\textrm{o}} \cup W$ and then subtracting $s_{\textrm{o}}$ from the sum.  For example in Fig. \ref{fig:ring4} (b) and (c), in order to measure the sum of the pink nodes 2, 8 and 10, we  measure the sum of pink nodes and all the black odd nodes, and then subtract $s_{\textrm{o}}$ from the obtained summation. Though the subgraph induced by $T_{\textrm{e}}$ is not complete, we can still freely measure nodes in $T_{\textrm{e}}$ with the help of the hub $T_{\textrm{o}}$. Therefore
From Theorem \ref{thm:hub}, $M^C_{k, \lfloor n/2 \rfloor }+1$ measurements are enough to recover $\bfx_{\textrm{e}}\in \mathcal{R}^{\lfloor n/2 \rfloor}$. % where the additional one measurement measures $s_{\textrm{o}}$. %the sum of all odd nodes.
Similarly, we can use $T_{\textrm{e}}$ as a hub to recover the subvector $\bfx_{\textrm{o}} \in \mathcal{R}^{\lceil n/2 \rceil}$ with $M^C_{k, \lceil n/2 \rceil }+1$ measurements, and thus $\bfx$ is recovered. %From above, we have
\begin{cor}\label{thm:ring4}
 All $k$-sparse vectors associated with $\mathcal{G}^4$ can be recovered with $M^C_{k, \lfloor n/2 \rfloor}+M^C_{k, \lceil n/2 \rceil}+2$ measurements, which is $O( 2k\log (n/(2k)))$. % from (\ref{eqn:MC}). %the measurement matrix $A$ defined is in (\ref{eqn:A1}) to (\ref{eqn:A4}). %Moreover, the number of measurements needed to identify $k$-sparse vectors on $\mathcal{G}^4$ satisfies
%\begin{equation}\label{eqn:g4}
%M^{\mathcal{G}^4}_{k,n} \leq M^C_{k, \lfloor n/2 \rfloor}+M^C_{k, \lceil n/2 \rceil}+2.
%\end{equation}
 \end{cor}
From a ring  to $\mathcal{G}^4$, although the number of edges only increases by $n$, the number of measurements required to recover $k$-sparse vectors significantly reduces from $\Theta(n)$ to $O( 2k\log (n/(2k)))$. This value %$k$-sparse vectors associated with $\mathcal{G}^4$ can be recovered with $O( 2k\log (n/(2k)))+2$ measurements, which
is in the same order as $M^C_{k,n}$, while %the number of measurements required in a complete graph with size $n$. Note that
the number of edges in $\mathcal{G}^4$ is only $2n$, compared with  $n(n-1)/2$ edges in a complete graph. % while the number of edges in a complete graph is $n(n-1)/2$. Besides,

Besides the explicit measurement construction based on the hub idea, %described before Theorem \ref{thm:ring4},  %we will show that
we can also recover $k$-sparse vectors associated with $\mathcal{G}^4$  with $O( \log n)$ random measurements. %h a simple random construction of the measurement matrix .
We need to point out that %unlike other measurement constructions in this paper,
these random measurements do not depend on the measurement constructions for a complete graph.

%Consider a $n$-step Markov chain $\{ X_k, 1 \leq k \leq n\}$ with
%initial distribution $P(X_1=1)=1$, $P(X_1=0)=0$ and transition
%matrix $P$ in (\ref{eqn:p}). Then $P(X_{k+1}=0~|~X_{k}=0)=0$,
%$P(X_{k+1}=1~|~X_{k}=0)=1$, $P(X_{k+1}=0~|~X_{k}=1)=0.5$, and
%$P(X_{k+1}=1~|~X_{k}=1)=0.5$.
%\begin{equation}\label{eqn:p}
%P = \begin{array}{*{20}c}
%   0  \\
%   1  \\
%\end{array}\left[ {\begin{array}{*{20}c}
%   0 & 1  \\
%   {0.5} & {0.5}  \\
%\end{array}} \right]
%\end{equation}
Consider an $n$-step Markov chain $\{ X_k, 1 \leq k \leq n\}$ with $X_1=1$. For any $k\leq n-1$, if $X_k=0$, then $X_{k+1}=1$; if $X_k=1$, then $X_{k+1}$ can be 0 or 1 with equal probability.
%Since $P(X_{k+1}=0~|~X_{k}=0)=0$, then
Clearly any realization of this Markov chain does not contain two or more consecutive
zeros, and thus is a feasible row of the measurement matrix. %For every realization $\{x_k\}$ ($1 \leq k \leq n$) of the
%Markov chain,  define $U:=\{ 1 \leq k \leq n: x_k=1\}$. Nodes in $U$ induce a connected subgraph of $\mathcal{G}^4$, and thus can be measured together in one measurement. Therefore, every realization of the Markov chain corresponds to a feasible row of the measurement matrix.
We have the following result, please see the conference version \cite{WXMT12} for its proof.  % $O(\log n)$ such
%measurements are sufficient to recover all $k$-sparse vectors with high
%probability.
%then $x_k=1$
%for every $k \in U$, and for two consecutive elements (in increasing
%order) $i$ and $j$ of $U$, either $j=i+1$ or $j=i+2$ holds. Thus,
%the corresponding node $i$ and node $j$ are connected by edge $(i,
%j)$. Then we can measure the sum of the nodes with $U$ as the set of
%indices directly by one measurement.
%
%Now that each realization of the Markov chain corresponds to a
%feasible measurement, we will show that $O(\log n)$ such
%measurements are sufficient to recover $k$-sparse vectors with high
%probability.
\begin{theorem}\label{thm:random}
%If each measurement corresponds to an independent realization of the
%above defined Markov chain, then
With high probability all $k$-sparse vectors associated with $\mathcal{G}^4$ can be recovered with $O(g(k) \log n)$
measurements obtained from the above Markov chain, where $g(k)$ is a function of
$k$.
\end{theorem}

%\begin{proof}
%See Appendix.
%\end{proof}
%

%At least $(n/2)$ measurements are necessary to recover a $1$-sparse vector in
%a ring network, while for $\mathcal{G}^4$, with $O(2k \log (n/k))+2$
%measurements we recover all $k$-sparse vectors. Thus,
Adding $n$ edges in the
form $(i, i+2 (\textrm{mod } n))$ to the ring   greatly
reduces the number of measurements needed from $\Theta(n)$ to $O(\log n)$. %One
%may naturally wonder whether it
%Is it necessary to add all these $n$
%edges to achieve such a reduction in the number of required
%measurements? Would we still achieve such a reduction by adding a smaller number of edges
%positive fraction of these $n$ edges?
Then how many edges in the form $(i, i+2(\textrm{mod } n))$ shall we add to the ring  such that the minimum number of measurements required to recover $k$-sparse vectors is exactly $\Theta (\log n)$? The answer is $n- \Theta(\log n)$.
%Our next result implies that we need to add such $n- \Theta(\log n)$ edges so that
%Our next result implies that
%adding $(1-\rho) n$ ($\rho \in (0,1)$) edges cannot reduce the
%number of measurements significantly, and $\Theta(n)$ measurements are still necessary to recover $1$-sparse vectors. We use
To see this, let $\mathcal{G}^4_h$  denote the graph obtained by deleting $h$ edges in the form $(i, i+2 (\textrm{mod } n))$ from
$\mathcal{G}^4$. For example in Fig. \ref{fig:ring4} (b), we delete edges $(3,5)$, $(8,10)$ and $(9,11)$ in red dashed lines from $\mathcal{G}^4$. Given $h$, %there are many choices of the edges to remove,
our following results do not depend on the specific choice of edges to remove.  We have%and hold for all graphs obtained by deleting $h$ edges in the form $(i, i+2 (\textrm{mod } n))$ from
%$\mathcal{G}^4$.

\begin{theorem}\label{thm:ring4h}
 The minimum number of measurements required to recover $k$-sparse vectors associated with $\mathcal{G}^4_h$ is lower bounded by $\lceil h/2 \rceil$, and upper bounded by $2M^C_{k,\lceil \frac{n}{2}\rceil}+h+2$.
 \end{theorem}
%We first provide a lower bound of the number of measurements needed to recover a $1$-sparse vector associated with $\mathcal{G}^4_h$.
%
%\begin{figure}[h]
%\begin{center}
%  % Requires \usepackage{graphicx}
%  \includegraphics[scale=0.4]{./Figures/ring4_4}\\
%  \caption{Delete $h$ long edges from $\mathcal{G}^4$}\label{fig:ring4h}
%\end{center}
%\end{figure}
%
%\begin{theorem}\label{thm:ring4rho}
%At least $\lceil  h/2 \rceil $ measurements
%are necessary to recover a $1$-sparse vector associated with $\mathcal{G}^4_h$.
%\end{theorem}
\begin{proof}
%Let $D=\{n_{e_i}, 1 \leq i \leq h\}$, where edge $(n_{e_i}-1, n_{e_i}+1)$, $\forall i$
%($1\leq i \leq h$)
Let $D$ denote the set of nodes such that for every $i \in D$, edge $(i-1, i+1)$
is removed from $\mathcal{G}^4$.
The proof of the lower bound follows the proof of Theorem 2 in \cite{TWHR11}.
The key idea is that recovering one non-zero element in $D$ is equivalent to recovering one non-zero element in a ring   with $h$ nodes, and thus $\lceil h/2 \rceil$ measurements are necessary.

For the upper bound, we %construct the measurements as follows. We
first measure nodes in $D$ separately with $h$ measurements.  Let $S$ contain the even nodes in $D$ and all the odd nodes. %One can check that $G_S$ is connected, and
$S$ can be used as a hub to recover the $k$-sparse subvectors associated with the even nodes that are not in $D$, and the number of measurements used is at most $M^C_{k,\lfloor \frac{n}{2}\rfloor}+1$. We similarly recover $k$-sparse subvectors associated with odd nodes that are not in $D$ using the set of the odd nodes in $D$ and all the even nodes as a hub.  The number of measurements is at most $M^C_{k,\lceil \frac{n}{2}\rceil}+1$. Sum them up and the upper bound follows.
%
%If edge $(i, i+2 (\textrm{mod } n))$ is removed, we measure node
%$i+1(\textrm{mod } n)$ directly. Then $h$ nodes are measured
%separately by $h$ measurements, and let $H$ denote the set of such
%nodes. We can recover up to $k$ non-zero entries corresponding to the remaining nodes by
%$2 f(k, n/n)+2$ measurements. This is because the sum of any combination
%of even nodes in the remaining $n-h$ nodes can be measured by using
%all the odd nodes and some nodes in $H$ as a hub. Then $f(k,n/2)+1$
%measurements are enough to recover up to $k$-sparse subvectors associated with even nodes.
%Similarly, $f(k, n/2)+1$ measurements are enough for odd nodes.
%%Thus, the total number of measurements is $O(2k \log n+h)$.
%Thus, the statement follows.
%
\end{proof}

Together with (\ref{eqn:MC}), Theorem \ref{thm:ring4h}   implies that
%Corollary \ref{cor:ring4h} implies that
 if $\Theta(\log n)$ edges in the form $(i, i+2 (\textrm{mod } n))$ are deleted from $\mathcal{G}^4$, then $\Theta( \log n)$ measurements are   necessary and sufficient to recover associated $k$-sparse vectors % with $\mathcal{G}^4_{\Theta(\log n)}$
  for   constant $k$. %Moreover, the lower bound in Theorem \ref{thm:ring4h} implies that if the number of edges removed is $\Omega( \log n)$, then the number of measurements required for sparse recovery is also $\Omega( \log n)$. Thus, we need to add $n-\Theta(\log n)$ edges to a ring network such that the number of measurements required for sparse recovery is exactly $\Theta(\log n)$.
 %Therefore, after adding $n-\Theta(\log n)$ edges to a ring network, the number of measurements
%needed to recover $k$-sparse vectors significantly
%reduces from $\Theta(n)$ to $O(2k\log (n/(2k)))$.  Note that at this point, the number of edges in the graph is only $2n-\Theta(\log n)$, while the number of edges in the complete graph is $n(n-1)/2$.

Since the number of measurements required by compressed sensing is greatly reduced when we add $n$ edges to a ring, one may wonder whether
the number of measurements needed %to locate $k$ non-zero elements
by group
testing can  be greatly reduced or not. Our next result shows that
this is not the case for group testing, please refer to the conference version \cite{WXMT12} for its proof.

\begin{prop}\label{thm:ring4gt}
$\lfloor n/4\rfloor$ measurements are necessary to locate two non-zero elements associated with $\mathcal{G}^4$ by group testing.
\end{prop}

By Corollary \ref{thm:ring4} and Proposition \ref{thm:ring4gt}, we
observe that in $\mathcal{G}^4$, with compressed sensing the number
of measurements needed to recover $k$-sparse vectors is $O(2k \log (n/(2k)))$,
while with group testing, $\Theta(n)$ measurements are required if
$k \geq 2$. %We can also show that $3n/4$ measurements are enough to locate two non-zero elements in $\mathcal{G}^4$ with group testing, and we skip the proof here.

\subsection{Line graph of a ring network with each router connecting to four routers}

\begin{figure}[h]
\begin{center}
  \includegraphics[width=0.7\linewidth, height=0.3\linewidth]{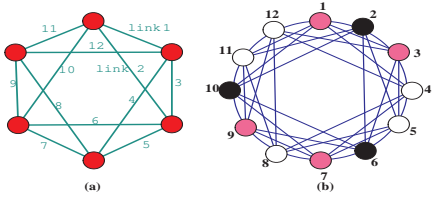}
  \caption{(a) a communication network with $n=12$ links, (b) the line graph of the network in (a). Measure the sum of any subset of odd nodes (e.g., 1, 3, 7, and 9) using nodes 2,6, and 10 as a hub }\label{fig:lineg4}
  \end{center}
\end{figure}

Here we compare our construction methods with those in   \cite{CKMS12,XMT11} on recovering link quantities in a network with each router connecting to four routers in the ring. Fig \ref{fig:lineg4} (a)\footnote{Fig \ref{fig:lineg4} (a) is a communication network with nodes representing routers and edges representing links, while Fig. \ref{fig:ring4} (a) is a graph model capturing topological constraints with nodes representing the quantities to recover.} shows such a network with $n$ links with $n=12$. As discussed in Section \ref{sec:model}, we analyze the line graph of the communication network  in Fig. \ref{fig:lineg4} (a). In its line graph, as shown in  Fig.   \ref{fig:lineg4} (b),  node $i$ (representing the delay on link $i$) is connected to nodes $i-3$, $i-2$, $i-1$, $i+1$, $i+2$, and $i+3$ (all mod $n$) for all odd $i$; and node $i$  is connected to nodes $i-4$, $i-3$, $i-1$, $i+1$, $i+3$, and $i+4$ (all mod $n$) for all even $i$. % See Fig.   \ref{fig:lineg4} (b) as an example.

 With the hub idea, we can recover $k$-sparse link delays in this network from $O(2k \log (n/(2k)))$ measurements. Specifically, we use the set of all the odd nodes as a hub to recover the values associated with the even nodes, and it takes $O(k \log (n/(2k)))$ measurements. %The odd nodes are divided into two groups, and each group is recovered using $O(k \log n)$ measurements.
 We then use the set of nodes $\{4j+2, j=0,..., \lceil \frac{n-2}{4} \rceil\}$ as a hub to recover the values associated with the odd nodes, see Fig.  \ref{fig:lineg4} (b) as an example. And it takes another $O(k \log (n/(2k)))$ measurements. %the nodes  $\{4j+1, j=0,..., \lceil \frac{n-2}{4} \rceil\}$, see Fig.  \ref{fig:lineg4} (b) as an example. Similarly, the set of nodes $\{4j, j=1,..., \lceil \frac{n-2}{4} \rceil\}$ as a hub to recover the nodes  $\{4j-1, j=1,..., \lceil \frac{n-2}{4} \rceil\}$.

%For the network in Fig \ref{fig:lineg4} (a),
Our construction of $O(2k\log (n/(2k)))$   measurements   to recover $k$-sparse link delays in the network in Fig \ref{fig:lineg4} (a) greatly improves over
the existing results in \cite{CKMS12,XMT11},   which  are based on the mixing time of a random walk.
%Theorem \ref{thm:ring4} implies that $M^{\mathcal{G}^4}_{k,n} \leq M^C_{k, \lfloor n/2 \rfloor}+M^C_{k, \lceil n/2 \rceil}+2$.
%Since  $M^C_{k,n} \leq O( k\log (nk))$, then $M^{\mathcal{G}^4}_{k,n} \leq O( 2k\log (n/(2k)))+2$. %We provide a constructive way to take $O(k \log n)$ measurements so
%as to identify up to $k$ errors in the proof of Theorem
%\ref{thm:ring4}. In order to decode the errors in even nodes, we
%first deduct $s_1$ from all the measurements on the even nodes, then
%apply $\ell_1$-minimization. To decode the errors in the odd nodes,
%we deduce $s_2$ from all the measurements on the odd nodes, and then
%apply $\ell_1$-minimization.
%If we randomly generate a measurement matrix for the complete graph with $\lfloor n/2 \rfloor$ nodes, then with overwhelming probability that $ O( k \log (n/2k)$ random measurements are enough to identify a $k$-sparse vector in $\mathcal{R}^{\lfloor n/2 \rfloor}$. Similarly, we randomly generate a measurement matrix for the complete graph with $\lceil n/2 \rceil $ nodes. Therefore from Theorem \ref{thm:ring4}, with overwhelming probability, $k$-sparse vectors associated with $\mathcal{G}^4$ can be recovered with $O(2k \log (n/2k))+2$ measurements. Thus, $M^{\mathcal{G}^4}_{k,n} \leq O(2 k\log ( n/2k )) +2$. In \cite{XMT11}, the measurement matrix corresponds to random walks on the graph, and
%$T(n)$ is the $\delta$-mixing time defined as follows.
%Let $\mu$ denote the stationary distribution over the nodes of a standard random walk
%over the graph $G$.
The %$\delta$-
mixing time $T(n)$ can be roughly interpreted as the minimum length of a random walk on a graph  such that its distribution is close to its stationary distribution. %is  the smallest $t'$ such that
%a random walk of length $t'$ starting at any node in $G$ ends up
%having a distribution $\mu'$ with $\|\mu-\mu'\|_{\infty} \leq 1/(2cn)^2$ %where the degree of every node in $G$ is between $D$ and $cD$
%for some $c\geq 1$, where $\mu$ is the stationary distribution over the nodes of a standard random walk
%over the graph $G$. %$T (n)$ is the $\delta$-mixing time of $G$ for  $\delta = 1/(2cn)^2$, and the degree of every node in $G$ is between $D$ and $cD$ for some $c\geq 1$.
\cite{XMT11} proves that $O( k T^2(n) \log n)$ measurements %from random walks
can identify $k$-sparse vectors with overwhelming probability by compressed sensing.  \cite{CKMS12} needs $O(k^2 T^2(n) \log (n/k))$  measurements to identify $k$ non-zero elements by group testing.  %one can easily see that %the stationary distribution is $\mu_i=1/n, \forall i$. A random walk with length at most  $\lfloor n/4 \rfloor$ starting at node 1 can never reach %the stationary distribution on
%node $\lfloor n/2 \rfloor$ is 0, therefore
$T(n)$ should be at least $n/8$ for the network in Fig \ref{fig:lineg4} (a). Then both results provide no saving in the number of measurements, while our construction reduces this number to  $O(2k\log (n/(2k)))$.  %as the mixing time is $\Theta(n)$. %However, our method can recover all $k$-sparse vectors with overwhelming probability with  $O(2k \log (n/2k))+2$ measurements .
%Besides, the number of measurements used to identify $k$ non-zero elements by group testing in \cite{CKMS12} is $O(k^2 T^2(n) \log (n/d)$, much greater than our result.

\subsection{Two-dimensional grid}\label{sec:2d}
Next we consider the two-dimensional grid, denoted by $\mathcal{G}^{2d}$. %As shown in
%Fig. \ref{fig:2d} (a),
$\mathcal{G}^{2d}$ has $\sqrt{n}$ rows and $\sqrt{n}$ columns. We assume $\sqrt{n}$ to be even here, and also skip `$\lceil\cdot \rceil$' and `$\lfloor\cdot \rfloor$' for notational simplicity. % but note that the number of nodes should always be an integer.

%\begin{figure}
%\centering
%\begin{tabular}{c c}
%%\includegraphics[width=0.5\linewidth]{./Figures/enum3}
%%\includegraphics[width=0.\linewidth]{./Figures/twoD1}
%%&
%\includegraphics[width=0.35\linewidth]{./Figures/twoD2}
%&
%%\includegraphics[width=0.5\linewidth]{./Figures/enum2}
%\includegraphics[width=0.35\linewidth]{./Figures/twoD3}
%\\
%%{\scriptsize (a) Two dimensional grid }&
%{\scriptsize (a) The set of black nodes as a hub }& {\scriptsize (b) Measure pink nodes via the hub}\\
%\end{tabular}
%      \caption{Sparse recovery on two-dimensional grid}
%      \label{fig:2d}
%   \end{figure}

\begin{figure}
\begin{minipage}{1.6in}
\begin{center}
\includegraphics[width=0.75\linewidth,height=0.5\linewidth]{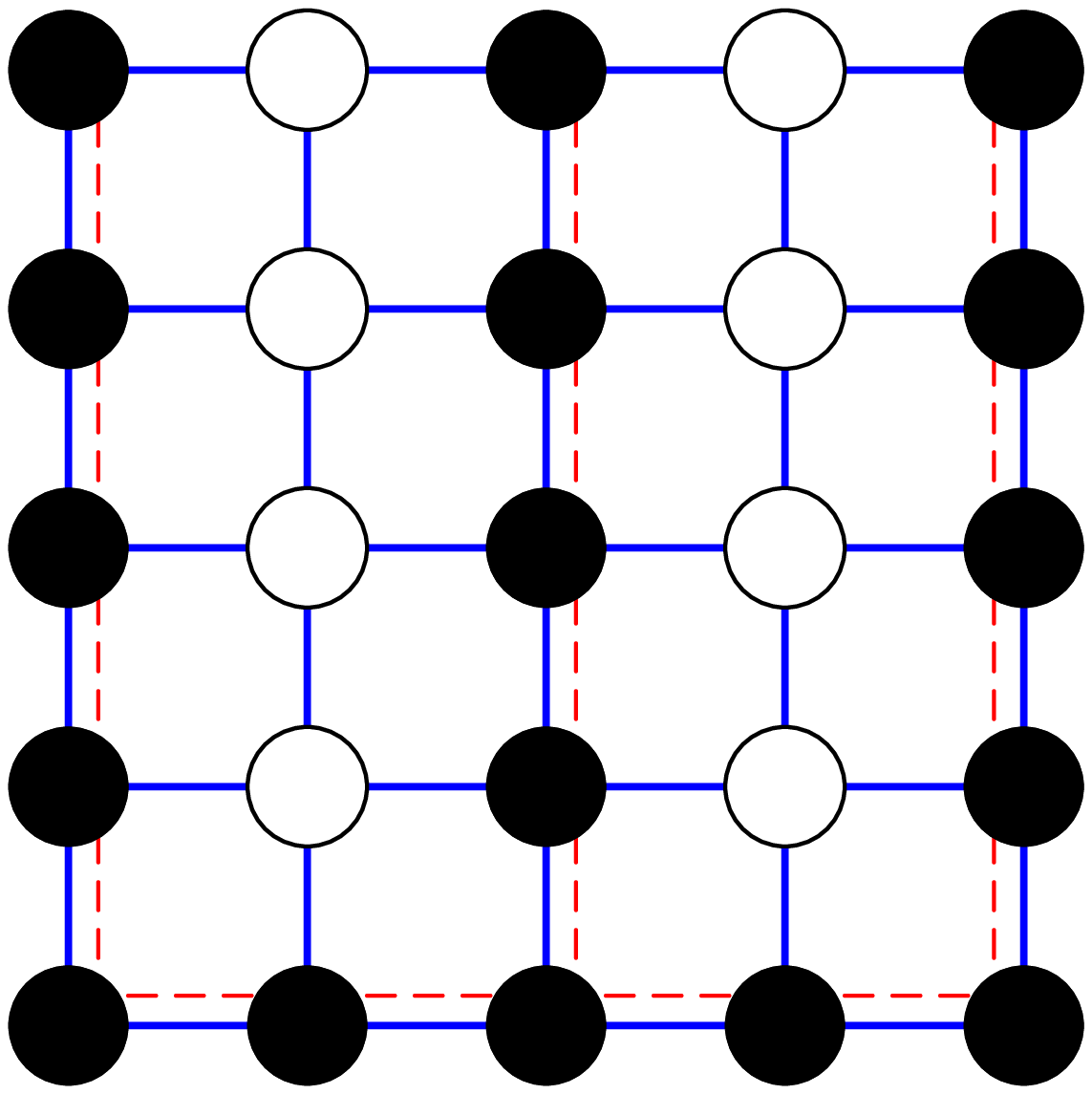}
\caption{Two-dimensional grid}\label{fig:2d}
\end{center}
\end{minipage}
\begin{minipage}{1.6in}
\begin{center}
\includegraphics[width=0.75\linewidth, height=0.5\linewidth]{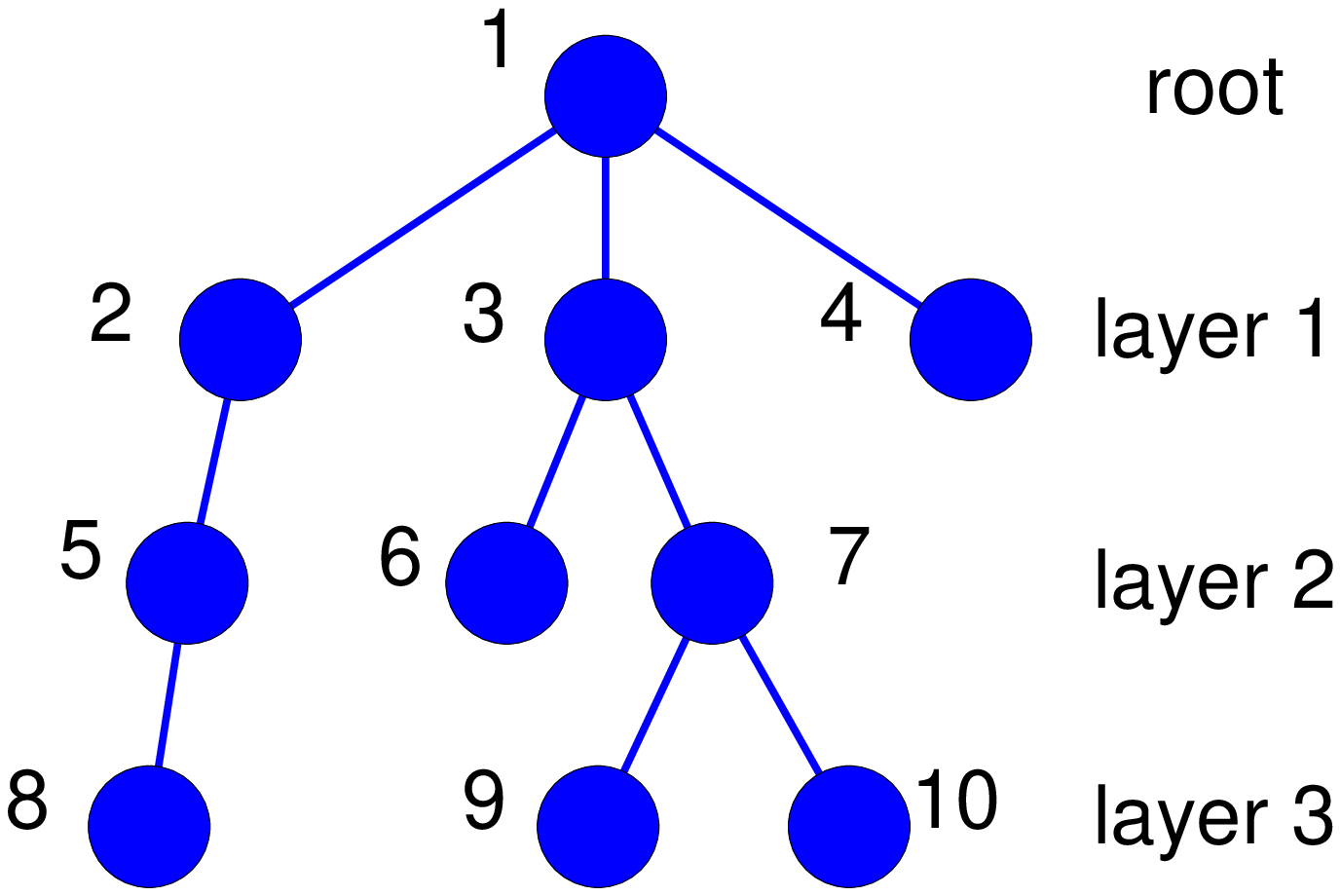}
\caption{Tree topology} \label{fig:tree}
\end{center}
\end{minipage}
\end{figure}

%\begin{figure}
%\begin{tabular}{c c}
%%\includegraphics[width=0.5\linewidth]{./Figures/enum3}
%%\includegraphics[width=0.\linewidth]{./Figures/twoD1}
%%&
%\includegraphics[width=0.35\linewidth]{./Figures/twoD2}
%&
%%\includegraphics[width=0.5\linewidth]{./Figures/enum2}
%\includegraphics[width=0.35\linewidth]{./Figures/twoD3}
%\\
%%{\scriptsize (a) Two dimensional grid }&
%{\scriptsize (a) The set of black nodes as a hub }& {\scriptsize (b) Measure pink nodes via the hub}\\
%\end{tabular}
%      \caption{Sparse recovery on two-dimensional grid}
%      \label{fig:2d}
%   \end{figure}

%We construct $2f(k, n/2-\sqrt{n}/2)+f(k,\sqrt{n})+2$ measurements %.provide a constructive way
%to recover $k$-sparse vectors associated with $\mathcal{G}^{2d}$.
%We assume $\sqrt{n}$ to be even here for notational simplicity, and the result can be easily modified for the case that $\sqrt{n}$ is odd.
The idea of measurement construction is still the use of a hub. %similar to that for graph $\mathcal{G}^4$.
First, Let $S_1$ contain the nodes in the first row and all the nodes in the odd columns, i.e., the black nodes in Fig. \ref{fig:2d}. Then $S_1$ can be used as a hub to measure $k$-sparse subvectors associated with nodes in $V \backslash S_1$. %, as shown in Fig. \ref{fig:2d}.
The number of measurements %used to recover up to $k$-sparse subvectors associated with  $V \backslash S_1$
is $M^C_{k, (n/2-\sqrt{n}/2)}+1$. Then let $S_2$ contain the nodes in the first row and all the nodes in the even columns, and use $S_2$ as a hub to recover up to $k$-sparse subvectors associated with nodes in $V \backslash S_2$. Then number of measurements required is also $M^C_{k, (n/2-\sqrt{n}/2)}+1$. Finally, use nodes in the second row as a hub to recover sparse subvectors associated with nodes in the first row. Since nodes in the second row are already identified in the above two steps, then we do not need to measure the hub separately in this step. The number of measurements here is  $M^C_{k, \sqrt{n}}$. %Thus, the total number of measurements by this construction is $2M^C_{k, (n/2-\sqrt{n}/2)}+M^C_{k, (\sqrt{n})}+2$.
Therefore,
 %We pick a subset of nodes that can induce a connected subgraph as a hub, like the black nodes in Fig. \ref{fig:2d} (b). Then for those nodes that are have direct edges to at least one node in the hub, we can measure the sum of any subset of nodes, like the pink nodes in  Fig. \ref{fig:2d} (c) by first %any two nodes are connected via the hub, and we can apply the result for complete graph with simple modification. For example, although the sum of the pink nodes cannot be measured directly since they do not induce a connected subgraph, we can
%measuring the sum of the nodes to be measured and the nodes in the hub, and then subtract the sum of the nodes in the hub. We need to pick several hubs such that every node in the graph is at least directly connected to one hub.
%\begin{prop}

\vspace{0.05in}
\textit{With $2M^C_{k, n/2-\sqrt{n}/2}+M^C_{k,\sqrt{n}}+2$ measurements one can recover $k$-sparse vectors associated with $\mathcal{G}^{2d}$.} %The number of measurements needed to recover $k$-sparse vectors associated with $\mathcal{G}^{2d}$ is at most $2M^C_{k, n/2-\sqrt{n}/2}+M^C_{k,\sqrt{n}}+2$.}

\subsection{Tree}\label{sec:tree}
Next we consider a tree topology as in Fig. \ref{fig:tree}. For a given tree, the
root is treated as the only node in layer 0. The nodes that are $t$
steps away from the root are in layer $t$. We say the tree has depth
$h$ if the farthest node is $h$ steps away from the root. Let $n_i$ denote the number of nodes on layer $i$, and $n_0=1$. We %show
%that $\sum_{i=0}^{h} M^C_{k,n_i}$ measurement are enough
construct measurements to recover vectors associated with a tree by the following \textit{tree approach}.
%in a tree with depth $h$ by compressive sensing, and provide two
%different ways to construct the measurement matrix.
%\begin{figure}[h]
%\begin{center}
%  % Requires \usepackage{graphicx}
%  \includegraphics[scale=0.25]{./Figures/tree}
%  \caption{Tree topology}\label{fig:tree}
%  \end{center}
%\end{figure}
%\noindent {\bf Approach 1}:
% We first measure the root directly.  Then, in
%order to measure the nodes in lay $i$ ($2 \leq i \leq h$), we use all the nodes in layers smaller than $i$ as a hub, and let
%$s_{i-1}$ denote the sum of these nodes, then we use the hub to measure the nodes in layer $i$. We can measure the sum of any
%subset of nodes in layer $i$ by measuring the sum of these nodes and all the nodes in layers smaller than $i$ and
%subtracting $s_{i-1}$ from the obtained sum. Since $f(k,n_i)$ measurements are enough to recover $k$-sparse vectors associated with $n_i$ nodes if
%any combination of the nodes can be measured, and we also know that measuring each node separately will cost $n_i$ measurements, then $\min (f(k,n_i)+1, n_i)$ measurements
%are enough to recover up to $k$-sparse vectors associated with nodes in layer $i$. Then the total  number of measurements to recover $k$-sparse vectors associated with the tree is $\sum_{i=0}^h \min (f(k,n_i)+1, n_i)$.
%
%Note that in this approach, identifying nodes in one layer is
%independent of identifying nodes in another layer, since the only
%information it requires is the sum of all the nodes above. Thus we
%can start with identifying nodes in any layer.

%\noindent {\bf Tree Approach}:
We recover the nodes layer by layer
starting from the root, and recovering nodes in layer $i$ requires
that all the nodes above layer $i$ should already be recovered. First measure the root separately.  When recovering the subvector associated with nodes in  layer $i$ ($2\leq i \leq h$), %note that $M^C_{k,n_i}$ measurements are enough to recover $k$-sparse vectors associated with a complete graph of $n_i$ nodes. %, or equivalently, when any combination of the nodes can be measured together.
we can measure the sum of any subset of nodes in layer $i$ using some nodes in the upper layers as a hub and then delete the value of the hub from the obtained sum.
%For those $n_i$ nodes in layer $i$, although there is no direct edge between any two nodes, we can use the nodes in the above lays to route among nodes in layer $i$.   Since all
%the nodes above layer $i$ are already recovered, we only need to subtract the value of these nodes from the obtained measurement.  Given a subset of nodes in layer $i$ to be measured together,
One simple way to find a hub %obtain the nodes in upper layers as a hub
is to trace back from  nodes to be
measured
on the tree simultaneously until they reach one same node. For example in Fig. \ref{fig:tree}, %in order to measure nodes 9 and 10 in layer 3 together, we can measure the sum of nodes 7, 9 and 10, and then subtract node 7 from the sum since the value at node 7 has already been identified when recovering nodes in layer 2. If we want to
  to measure the sum of nodes 5 and 7, we  trace back to the root and measure the sum of nodes 1, 2, 3, 5, and 7  and then subtract the values of nodes 1, 2, and 3, which are already identified when we recover nodes in the upper layers. %For layer $i$, since measuring each node separately will cost $n_i$ measurements, then $\min (f(k,n_i), n_i)$ measurements
%are enough to recover up to $k$-sparse vectors associated with nodes in layer $i$.
%Since we use Then
%$\min(f(k,n_i),n_i)$ measurements are enough to recover up to $k$-sparse vectors associated %with nodes in
%layer $i$.
%Then the tree approach uses $\sum_{i=0}^h M_{k,n_i}$ measurements in
%total. Then
With this approach, we have,
%In the process of
%recovering nodes in layer $i$, we will still use the $f(k, n_i)$ measurements for the complete graph. %apply the  need to
%%be able to measure any combination of nodes in layer $i$.
%Instead of using all the nodes above layer $i$ as a hub, here since all
%the nodes above layer $i$ are already identified, we can use any
%subset of them to route among the nodes to be measured in layer $i$
%and then simply delete the value of these nodes from the obtained sum.  One simply way to route is tracing back from the nodes to be
%measured in layer $i$ on the tree at the same time until they reach one same node. Then
%$\min(f(k,n_i),n_i)$ measurements are enough to recover up to $k$-sparse vectors associated with nodes in
%layer $i$. This approach uses $\sum_{i=0}^h\min(f(k,n_i),n_i)$ measurements in
%total.
%\begin{prop}\label{thm:tree}

\vspace{0.05in}
\textit{
$\sum_{i=0}^h M^C_{k,n_i}$ measurements are enough to recover $k$-sparse vectors associated with a tree with depth $h$, where $n_i$ is the number of nodes in layer $i$.} %is at most $\sum_{i=0}^h M_{k,n_i}$.
%\vspace{0.05in}

%\begin{equation}
%M^T_{k,n} \leq \sum_{k=0}^h \min (M^C_{k, n_i}, n_i).
%\end{equation}
%\end{prop}
%
%\begin{proof}
%%We have two different ways to construct such $O(k h \log n)$
%%measurements. Both measure the nodes by layers.
%It follows from the discussion of the tree approach.
%\end{proof}

%% file: general.tex
\section{Sparse Recovery over General Graphs}\label{sec:general}

In this section we consider recovering $k$-sparse vectors associated with general graphs. The graph is assumed to
be connected. If not, %we simply  treat each component as a connected subgraph and
 we design measurements to recover $k$-sparse subvectors associated with each component separately. %Then the total number of measurements for the whole
%graph is the sum of the measurements on each component.
%all the components is the number of measurements for.

%Inspired by the construction methods in Section \ref{sec:special}, %for measurements design on general graphs,
In Section \ref{sec:bound}
we propose a general design guideline based on ``$r$-partition''. The key idea is to divide the nodes into a small number of groups such that each group can be measured with the help of a hub. %nodes in the same group are connected to one hub, and thus %Then nodes in the same group
%can be measured freely with the help of the hub.  $\mathcal{G}^4$ and  Erd\H{o}s-R\'enyi random graphs are examples of graphs having an $r$-partition. % as an example to illustrate the design guideline based on $r$-partition.
Since finding the minimum number of such groups turns out to be NP-hard  in general, in Section \ref{sec:algo} we propose a simple algorithm to design measurements on any given graph. %a small number of measurements to recover $k$-sparse vectors associated with any given graph.

%
%
% In Section \ref{sec:bound} we will provide two general designing guidelines %inspired from previous results on special networks
% together with their corresponding number of measurements. One idea (Theorem \ref{thm:general}) is to obtain a spanning tree of the graph and then apply the tree approach introduced in Section \ref{sec:special}.  %This idea does not take full advantage of the edges in a general graph since we in fact only use the edges in the spanning tree.
%
% The other idea (Theorem \ref{thm:component}) is to divide nodes into several groups such that nodes in the same group can be freely combined in one measurement via a hub, which induces a connected subgraph. Since the number of measurements needed for a given graph can be much smaller that the general upper bounds in Section \ref{sec:bound}, then in Section \ref{sec:algo}, we provide an algorithm to design measurements as few as possible to recover $k$-sparse vectors associated with any given graph.
%
%%Our methods to recover $k$-sparse vectors in the special networks in the
%%previous section inspires the following two different ways to
%%identify sparse vectors associated with a general network.
%\subsection{Guidelines for Measurement Construction on General Graphs}\label{sec:bound}
\subsection{Measurement Construction Based on $r$-partition}\label{sec:bound}
 \begin{defi}[$r$-partition]\label{def:rp}
Given $G=(V,E)$, disjoint subsets $N_i$ ($i=1,...,r$) of $V$ form an \textbf{$r$-partition} of $G$ if and only if these two conditions both hold: (1) $\cup_{i=1}^r N_i=V$, and (2) $\forall i$, $V \backslash N_i$ is a hub for $N_i$.
%all the three condition holds:
%\begin{description}
%\item[(1)] $\cup_{i=1}^r N_i=V$,
%\item[(2)] $G_{V \backslash N_i}$ is connected for every $i=1,...,r$, and
%\item[(3)]  $\forall$ $u\in N_i$, $\exists v\in V \backslash N_i$ such that $(u,v)\in E$.
% \end{description}
 \end{defi}
Clearly, $T_{\textrm{o}}$ and $T_{\textrm{e}}$ form a $2$-partition of graph $\mathcal{G}^4$. With Definition \ref{def:rp} and Theorem \ref{thm:hub}, we have %of $r$-partition, we have
%\begin{theorem}\label{thm:component}
%Given graph $G=(V,E)$, let $c_i$ ($i=1,...,r$) denote $r$ disjoint
%subsets of $V$, and let $c= \cup_{i=1}^r c_i$. If $G_{c_i}$ is
%connected for every $i=1,...,r$, and for every $u \in V$, there
%exists some $v \in c$ such that $(u,v) \in E$, then $O(r k \log n)$
%measurements are enough to identify up to $k$ errors in the graph.
%\end{theorem}
\begin{theorem}\label{thm:component}
If $G$ has an $r$-partition $N_i$ ($i=1,...,r$), then the number of measurements needed to recover $k$-sparse vectors associated with $G$ is at most  $\sum_{i=1}^r M^C_{k,|N_i|}+r$, which is $O(r k\log (n/k))$. %measurements are enough to recover $k$-sparse vectors associated with $G$, where $n_i=|N_i|$. %is the cardinality of set $N_i$.
\end{theorem}
Another example of the existence of an $r$-partition is the Erd\H{o}s-R\'enyi random graph $G(n,p)$ with $p>\log n/n$. %With $r$-partition,
The number of our constructed measurements on $G(n,p)$ is less than the existing estimates in \cite{CKMS12,XMT11}. Please refer to %the detailed discussion of $G(n,p)$ in
Section \ref{sec:erdos} for the detailed discussion.

 Clearly, if an $r$-partition exists, the number of measurements also depends on $r$.  In general one  wants to reduce $r$ so as to reduce the number of measurements. %However, finding an $r$-partition of graph $G$ with the smallest $r$ turns out to be NP-hard. %the $r$-partition problem turns out to be $NP$-complete
Given graph $G$ and integer $r$, the question that whether or not $G$ has an $r$-partition is called \textit{$r$-partition problem}.
%We have the following proposition.
In fact,
\begin{theorem}\label{thm:np}
$\forall r \geq 3$, $r$-partition problem is NP-complete.
\end{theorem}
%Please refer to Appendix for its proof.
Please see the conference version \cite{WXMT12} for its proof. We remark that we cannot prove the hardness of the $2$-partition problem though we conjecture it is also a hard problem.

Although  finding an $r$-partition with the smallest $r$ in general is NP-hard, it still provides a guideline that one can reduce the number of measurements by constructing a small number of hubs such that all the nodes are connected to at least one hub. Our measurement constructions for some special graphs in Section \ref{sec:special} are also based on this guideline. We next provide efficient measurement design methods for a general graph $G$ based on this guideline.
\subsection{Measurement Construction Algorithm for General Graphs}\label{sec:algo}
%Section \ref{sec:bound} proposes the $r$-partition concept as a measurement design guideline. %and uses it to reduce the number of measurements needed by construction.
%But finding an $r$-partition with the smallest $r$ in general is NP-hard. %provides two general ideas in designing measurements.
%Given a connected graph $G$, how shall we efficiently design a small number of measurements %matrix such
%%that the number of measurements is as small as possible while at the same time we can still
%to recover $k$-sparse vectors associated with $G$?
%%However, the spanning tree idea does not take full advantage of the edges in a general graph since we only use the edges in the spanning tree.  and it is NP-hard to find an $r$-partition of a general graph with smallest $r$.

%(1) $G_S$ is connected for each set $S$ of nodes measured
%together, (2) the number of measurements are relatively small, and
%(3) we can still identify up to $k$ errors?
%We assume the graph to
%be connected without loss of generality. If the graph is not
%connected, we treat each component as a connected subgraph and
%design measurements that can identify up to $k$ errors separately
%for each component. Then the total number of measurements needed for
%all the components is the number of measurements for the whole
%graph.
%Since $G$ is connected,
One simple way is to find the spanning tree of $G$, and   use the tree approach in Section \ref{sec:tree}. %For a connected
% graph $G=(V,E)$ with $|V|=n$,
The depth of the spanning tree is at least $R$, where $R=\min_{u \in V}
\max_{v \in V} d_{uv}$ is the radius of $G$ with $d_{uv}$ as the
length of the shortest path between $u$ and $v$. % Let $u^*$ denote the node such that $\max_{v \in V} d_{u^* v}= R$.
%Pick $u^*$ as the root and obtain a spanning tree of $G$ by
%breadth-first search. Then the depth of the spanning tree is
%$R$. Apply the tree approach on the spanning tree rooted at $u^*$, from Theorem \ref{thm:tree}, we know that the number of measurements used is at most $ RM^C_{k,n}$, which is no greater than $Rk \log (n/k)$.
This approach only uses edges in the spanning tree, and the number of measurements needed %is at least $R$, one for each layer, and thus
is large when the radius $R$ is large. For example, the radius of  $\mathcal{G}^4$   is $n/4$, then the   tree approach uses at least $n/4$ measurements, while %one for each layer. However,  %if we can take advantage of the additional edges in a general connected graph, the number of measurements can be less than that in the spanning tree approach. Like in $\mathcal{G}^4$,
%the number of measurements can be as small as
$O(2k \log (n/2k))$ measurements are already enough if we take advantage of the additional edges not in the spanning tree.

Here we propose a simple algorithm to design the measurements for general graphs. The algorithm combines the ideas of the tree approach and the $r$-partition. %and is very simple to implement.
We still  divide nodes into a small number of groups such that each group can be identified via some hub. Here nodes in the same group are the leaf nodes of a spanning tree of a gradually reduced graph. A leaf node has no
children on the tree.

Let $G^*=(V^*,E^*)$ denote the original graph. The algorithm is built on the following two subroutines. %\textbf{Leaves}($G$, $u$), and \textbf{Reduce}($G$, $u$, $H$).
 \textbf{Leaves}($G$, $u$) returns the set of leaf nodes of a spanning tree of $G$ rooted at $u$.  \textbf{Reduce}($G$, $u$, $K$) deletes $u$ from $G$ and fully connects all the neighbors of $u$. Specifically, for every two neighbors $v$ and
$w$ of $u$, %if edge $(v,w)$ does not already exist,
we add a edge
$(v,w)$, if not already exist, %Moreover, for each edge $(s,t)$, set $H_{(s,t)}$ contains the nodes in $G^*$ that serve as a hub for $s$ and $t$. For the new added edge $(v,w)$,
and let $K_{(v,w)}=K_{(v,u)}\cup K_{(u,w)} \cup \{u \}$, where for each edge $(s,t)$, $K_{(s,t)}$ denotes the set of nodes, if any,  that connects $s$ and $t$ in the original graph $G^*$. We record $K$ such that measurements constructed on a reduced graph $G$ can be feasible in $G^*$.

%and add an mark $u$ to that edge. We delete a node and add
%edges sequentially for all the nodes in $L_f$,
%
% Let $u^*$ denote the node
%such that $\max_{v \in V} d_{u^* v}= R$. Pick $u^*$ as the root and
%obtain a spanning tree of $G$ by breadth-first search. Then clearly
%the depth of the spanning tree is $R$. We divide nodes into
%different groups and two nodes are in the same graph if and only if
%they are on the same level of the tree.  Nodes in the same group can
%be measured freely using all the nodes in the upper levels as a hub,
%and $O(k \log n)$ measurements are enough to identify $k$ errors in
%each group. We next provide an improved design of measurements based
%on the spanning tree idea.
 \floatname{algorithm}{Subroutine}
\begin{algorithm}
\begin{algorithmic}[1]
\REQUIRE graph $G$, root $u$
 \STATE  Find a spanning tree $T$ of $G$ rooted at $u$ by breadth-first search, and let $S$ denote the set of leaf nodes of $T$.
 \RETURN $S$
 \end{algorithmic}
 \caption{\textbf{Leaves}($G$, $u$)}
 \end{algorithm}
 \begin{algorithm}
\begin{algorithmic}[1]
\REQUIRE  $G=(V, E)$,  $H_e$ for each $e \in E$, and node $u$
 \STATE  $V=V\backslash u$.
 \FORALL {two different neighbors $v$ and
$w$ of $u$}
 \IF {$(v,w) \notin E$}
 \STATE  $E=E \cup (v,w)$, $K_{(v,w)} =K_{(v,u)} \cup K_{(u,w)} \cup \{u \}$.%add an mark $u$ to that edge.
 \ENDIF
 \ENDFOR
 \RETURN $G$, $K$
 \end{algorithmic}
 \caption{\textbf{Reduce}($G$, $u$, $K$)}
 \end{algorithm}
\setcounter{algorithm}{0}
 \floatname{algorithm}{Algorithm}
 \renewcommand{\algorithmicreturn}{\textbf{Output:}}
  \begin{algorithm}
\begin{algorithmic}[1]
\REQUIRE  $G^*=(V^*, E^*)$.
 \STATE $G=G^*$, $K_e=\varnothing$ for each $e \in E$
\WHILE { $|V| > 1$}
 \STATE  Find the node $u$ such that $\max_{v \in V} d_{u v}= R^G$, where $R^G$ is the radius of $G$.  $S=$\textbf{Leaves}($G$, $u$).
  \STATE Design $f(k, |S|)+1$ measurements to recover $k$-sparse vectors associated with $S$ using nodes in $V\backslash S$ as a hub.
   \FORALL {$v$ in $S$}
 \STATE $G=$ \textbf{Reduce}($G$, $v$, $K$)
 \ENDFOR
 \ENDWHILE
 \STATE Measure the last node in $V$ directly.
 \RETURN All the measurements.
 \end{algorithmic}
 \caption{Measurement construction for graph $G^*$}\label{algo:design}
 \end{algorithm}

Given graph $G^*$, %Given connected $G=(V,E)$ with diameter $R$,
let $u$ denote the
node such that $\max_{v \in V^*} d_{u v}= R$, where $R$ is the radius of $G^*$. Pick $u$ as the root
and obtain a spanning tree $T$ of $G^*$ by breadth-first search. Let
$S$ denote the set of leaf nodes in $T$.  %Since $G^*_{V^* \backslash S}$ is connected, then
With $V^* \backslash S$ as a hub, we
can design $f(k,|S|)+1$ measurements to recover up to $k$-sparse vectors associated with $S$.  %using the nodes in $V^*\backslash S$ as a hub. The additional one measurement measures the sum of the nodes in $V^*\backslash S$. %Since every node in $L_f$ is identified now,
%we can reduce the network as follows. For each node $u \in L_f$,
We then reduce the network by deleting every node $v$ in $S$ and fully connects its neighbors. %and for every two different neighbors $v$ and
%$w$ of $u$, if edge $(v,w)$ does not already exist, we add an edge
%$(v,w)$ and add an mark $u$ to that edge. We delete a node and add
%edges sequentially for all the nodes in $L_f$, and let $G_2=(V_2,
%E_2)$ denote the resulting graph.
For the   reduced network $G$, we repeat the above process until all the nodes are deleted. %Note that for each edge $(v, w)$ in the reduced graph $G$, $H_{(v,w)}$ denotes the set of nodes, if any,  that are already identified in previous steps and connect nodes $v$ and $w$ in the original graph $G^*$. Thus, %when we design measurements for the reduced network,
%a measurement should include nodes in $H_{(v,w)}$ if edge $(v, w)$ is used so as to be feasible in the original graph $G^*$. %we should also measure all the nodes Note that the reduced network may contain some edge $(v,w)$ that is not in $E^*$. In $G^*$, there are several nodes that connects $v$ and $w$ and has
Note that when designing the measurements in a reduced graph $G$, if a measurement passes edge $(v, w)$, then it should also include nodes in $K_{(v,w)}$ so as to be feasible in the original graph $G^*$.

In each step tree $T$ is rooted at node $u$ where $ \max_{v \in {V}} d_{u v}$ equals the radius of the current graph $G$. Since all the leaf nodes
of $T$ are deleted in the graph reduction procedure, the radius of the new obtained graph should be reduced by at least one. %Since the radius of $G^*$ is $R$,
Then we have at most $R$ iterations in Algorithm \ref{algo:design} until only one node is left. Clearly we have,
\begin{theorem}\label{prop:algo}
The number of measurements designed by Algorithm \ref{algo:design} is at most $Rf(k, n)+R+1$, where $R$ is the radius of the graph.
\end{theorem}
We  remark  that the number of measurements by the spanning tree approach %we mentioned at the beginning of Section \ref{sec:algo}
is also no greater than $Rf(k, n)+R+1$.   %in Proposition \ref{prop:algo} is the same as that of
%depends on the radius $R$ of the graph,
However, since Algorithm 1 also considers edges that are not in the spanning tree, we expect that  for general graphs, it uses fewer measurements than the spanning tree approach. This is verified in Experiment 1 in Section \ref{sec:simu}. %and thus the upper bound is loose.  % then the number of measurements designed can be much less than the upper bound in Proposition \ref{prop:algo}.

%% file: erdos.tex
\section{Sparse Recover over   Random Graphs}\label{sec:erdos}

Here we consider measurement constructions over the Erd\H{o}s-R\'enyi random graph $G(n,p)$, which has $n$ nodes and every two nodes are connected by a edge independently with probability $p$. The behavior of $G(n,p)$ changes significantly when $p$ varies. We study the dependence of number of measurements needed for sparse recovery on $p$.

\subsection{$np=\beta \log n$ for some constant $\beta>1$}
Now $G(n,p)$ is connected almost surely \cite{ER60}. Moreover, we have the following lemma regarding the existence of an $r$-partition. %the set of nodes that can be measured together via some hub.

\begin{lemma}\label{lem:huberdos}
When $p=\beta \log n/n$ for some constant $\beta>1$, with probability at least $1-O(n^{-\alpha})$ for some $\alpha>0$, every set $S$ of nodes with size $|S|=n/(\beta-\epsilon)$ for any $\epsilon \in (0, \beta-1)$ forms a hub for the complementary set $T=V \backslash S$, %, where $V$ is the set of all nodes.
%Moreover, with probability at least $1-O(n^{-\alpha})$,
which implies that $G(n,p)$ has a $\lceil\frac{\beta-\epsilon}{\beta-\epsilon-1}\rceil$-partition.
\end{lemma}

\begin{proof}
Note that the subgraph $G_S$ is also Erd\H{o}s-R\'enyi random graph in $G(n/(\beta-\epsilon), p)$. Since $p=\beta\log n/n>\log ( n/(\beta-\epsilon))/(n/(\beta-\epsilon))$,   $G_S$ is connected almost surely.

%We claim that with high probability, for every $u \in T$, there exists $v \in S$ such that $(u,v) \in E$. % and for
%every $v \in N_2$, there exists $u \in N_1$ such that $(u,v) \in E$.
%Then by Corollary \ref{cor:part}, the statement follows.
Let $P_f$ denote the probability that there exists some $u \in
T$ such that $(u, v) \notin E$ for every $v \in S$. Then
\begin{align*}
P_f&= \sum_{u\in T}(1-p)^{|S|}= (1-\frac{1}{\beta-\epsilon})n
 (1-\beta \log n/n)^{n/(\beta-\epsilon)}\\
  &= (1-1/(\beta-\epsilon))n
(1-\beta \log n/n)^{ \frac{n}{\beta \log n}\cdot \frac{ \beta \log n}{\beta-\epsilon}}\\
&\leq (1-\frac{1}{\beta-\epsilon})n e^{-\frac{\beta\log n}{\beta-\epsilon}} \quad
\leq(1-\frac{1}{\beta-\epsilon}) n^{-\epsilon/(\beta-\epsilon)}. % \\
%&\leq n^{-\epsilon/2}/2,
\end{align*}
 %Then $P_f$ goes to zero as $n$ goes to infinity.  %i.e., with high probability for every $u \in N_1$, there exists $v \in N_2$
%such that $(u,v) \in E$.
%According to the definition of a hub, with high probability, $S$ is a hub for $T$.
%Then the first claim follows.
Thus, $S$ is a hub for $T$ with probability at least $1-O(n^{-\alpha})$ for   $\alpha=\epsilon/(\beta-\epsilon)>0$. Since the size of $T$ is $(1-1/(\beta-\epsilon))n$,  $G(n,p)$  has at most  $\lceil \frac{\beta-\epsilon}{\beta-\epsilon-1}\rceil$ such disjoint sets. Then by a simple union bound, one can conclude that $G(n,p)$ has a $\lceil\frac{\beta-\epsilon}{\beta-\epsilon-1}\rceil$-partition with probability at least $1-O(n^{-\alpha})$.
\end{proof}

%From Lemma \ref{lem:huberdos}, every set $T$ of $(1-1/(\beta-\epsilon))n$ ($\epsilon \in (0, \beta-1)$) nodes can be measured freely with the rest of the nodes as a hub, thus $M^C_{k, (1-1/(\beta-\epsilon))n}+1$ measurements are enough to identify $k$-sparse subvectors associated with $T$. Since $G(n,p)$ only has at most  $\lceil \frac{\beta-\epsilon}{\beta-\epsilon-1}\rceil$ such disjoint sets, then  $(\lceil \frac{\beta-\epsilon}{\beta-\epsilon-1}\rceil )(M^C_{k, (1-1/(\beta-\epsilon))n}+1)$ measurements are enough to recover $k$-sparse vectors associated with $G(n,p)$. By (\ref{eqn:MC}), and let $\epsilon \rightarrow 0$, we have
For example, when $\beta>2$, Lemma \ref{lem:huberdos} implies that any two disjoint sets $N_1$ and $N_2$ with $|N_1|=|N_2|=n/2$ form a 2-partition of $G(n,p)$ with probability $1-O(n^{-\alpha})$. From Theorem \ref{thm:component} and Lemma \ref{lem:huberdos}, and let $\epsilon \rightarrow 0$, we have

\vspace{0.05in}
%\begin{prop}
\textit{When $p=\beta \log n/n$ for some constant $\beta>1$, all $k$-sparse vectors associated with %an Erd\H{o}s-R\'enyi random graph
$G(n,p)$ can be identified with $O(\lceil \frac{\beta}{\beta-1}\rceil k \log (n/k))$ measurements with  probability at least $1-O(n^{-\alpha})$ for some $\alpha>0$.}
%\end{prop}
\vspace{0.05in}

\cite{CKMS12} considers group testing over Erd\H{o}s-R\'enyi random graphs and shows that $O(k^2 \log^3 n)$ measurements %corresponding to random walks on the graph
are enough to identify up to $k$ non-zero entries %in an $n$-dimensional logical vector
if it further holds that $p=\Theta (k \log^2 n/n)$. Here with compressed sensing setup and $r$-partition results, we can recover $k$-sparse vectors in $\mathcal{R}^n$ with $O(k\log(n/k))$ measurements when $p>\log n/n$. This result also improves over the previous result in \cite{XMT11}, which requires $O(k \log^3 n)$ measurements for compressed sensing on $G(n,p)$.

\subsection{$np-\log n \rightarrow +\infty$, and $\frac{np-\log n }{\log n}\rightarrow 0$}
Roughly speaking, $p$   is just large enough to guarantee that $G(n,p)$ is connected almost surely \cite{ER60}. The diameter $D=\max_{u,v}d_{uv}$ of a connected graph is the greatest distance between any pair of nodes, and here it is concentrated around $\frac{\log n}{\log \log n}$ almost surely \cite{Bollobas01}. We design measurements on $G(n,p)$ with  Algorithm 1. With Theorem \ref{prop:algo} and the fact that the radius $R$ is no greater than the diameter $D$ by definition, we have
%\begin{prop}

\vspace{0.05in}
\textit{When $np-\log n \rightarrow +\infty$, and $\frac{np-\log n }{\log n}\rightarrow 0$, $O(k\log n \log (n/k)/\log \log n)$ measurements can identify $k$-sparse vectors associated with $G(n,p)$ almost surely.}
%\vspace{0.05in}

%\end{prop}
\subsection{$1<c=np<\log n$}

Now $G(n,p)$ is disconnected and has a unique giant component containing $(\alpha+o(1))n$ nodes almost surely with $\alpha$ satisfying
%\begin{equation*}
$e^{-c \alpha}=1-\alpha,$
%\end{equation*}
or equivalently,
\begin{equation}\nonumber
\alpha=1-\frac{1}{c}\sum_{k=1}^{\infty}\frac{k^{k-1}}{k !}(ce^{-c})^k,
\end{equation}
and all the other nodes belong to small components.
%The diameter $D=\max_{u,v}d_{uv}$ of a connected graph is the greatest distance between any pair of nodes. %the largest number of nodes which must be traversed in order to travel from one node to another.
The expectation of the total number of components in $G(n,p)$ is $(1-\alpha-c(1-\alpha)^2/2+o(1))n$ \cite{ER60}. Since it is necessary to take at least one measurement for each component, $(1-\alpha-c(1-\alpha)^2/2+o(1))n$ is an expected lower bound of measurements required to identify sparse vectors. % associated with $G(n,p)$.

The diameter $D$ of a disconnected graph is defined to be the largest distance between any pair of nodes that belong to the same component. Since   $D$  is now $\Theta(\log n/ \log (np))$ almost surely \cite{CL01}, then for the radius $R$ of the giant component,  $R \leq D = O(\log n/ \log (np))$, where the second equality holds almost surely.
   We use Algorithm 1 to design measurements on the giant component, %and the number of measurements used to identify $k$-sparse subvectors associated with the giant component is at most $ O(k\log^2 n/\log(np))$ almost surely.  We then
   and then measure every node in the small components directly. Thus, $k$-sparse vectors associated with $G(n,p)$ can be identified almost surely with   $ O(k\log n \log(n/k)/\log(np))+(1-\alpha+o(1))n$ measurements.

  Note that here  almost surely the size of every small component is at most $\frac{\log n+ 2\sqrt{\log n}}{np-1-\log(np)}$ (Lemma 5, \cite{CL01}). If $k= \Omega( \log n)$, %for every small component,  its node values are all non-zero, thus,
  %the number of measurements needed to identify the subvector associated with a small component should be no less than the size of the component. Therefore,   when $k= \Omega( \log n)$,
  almost surely
  $(1-\alpha+o(1))n$ measurements are necessary to identify subvectors associated with small components, and thus necessary for identifying $k$-sparse vectors associated with $G(n,p)$. Combing the arguments, we have

 % \begin{prop}
 \vspace{0.05in}
\textit{  When $1<c=np<\log n$ with  constant $c$, we can identify $k$-sparse vectors associated with $G(n,p)$ almost surely with   $O(k\log n \log (n/k)/\log(np))+(1-\alpha+o(1))n$ measurements. $(1-\alpha-c(1-\alpha)^2/2+o(1))n$ is an expected lower bound of the number of measurements needed. Moreover, if $k= \Omega( \log n)$, almost surely $(1-\alpha+o(1))n$ measurements are necessary to identify $k$-sparse vectors.}
 %  \vspace{0.05in}

 % \end{prop}

  \subsection{$np<1$}

  Since the expectation of the total number of components in $G(n,p)$ with $np<1$ is $n-pn^2/2+O(1)$ \cite{ER60}, %and we need to take at least one measurement of each component,
  then  $n-pn^2/2+O(1)$ is an expected lower bound of the number of measurements required. % to identify  vectors associated with $G(n,p)$.
  Since almost surely all components are of size $O(\log n)$, then we need to take $n$ measurements when $k= \Omega( \log n)$. Therefore,
  %\begin{prop}

   \vspace{0.05in}
\textit{  When  $np<1$, we need at least $n-pn^2/2+O(1)$ measurements to identify $k$-sparse vectors associated with $G(n,p)$ in expectation. Moreover, when $k= \Omega( \log n)$, $n$ measurements are necessary almost surely.}%to identify $k$-sparse vectors almost surely.
  %\vspace{0.05in}

 % \end{prop} 

%% file: extension_arxiv.tex
\section{Adding additional graph constraints} \label{sec:extension}

%We discuss the design of the measurement matrices with graph constraints in the  previous sections.
Our   constructions are based on assumptions (A1) and (A2). Here we  consider   additional graph constraints brought by practical implementation. We first consider measurement construction with   length constraints, since measurements with short length are preferred in practice. We then discuss the scenario that each measurement should pass at least one node in a fixed subset of nodes, since in network applications, one may want to reduce the number of routers that initiate the measurements. %when only a  subset of nodes can act as agents and each measurement should pass at least one agent. %In Section \ref{sec:huberror}, we address the issue that a hub measurement is erroneous, and prove that successful recovery is still achievable in this case. % even though hub measurements are not correct.% in practice one may only measure a small number of nodes in one measurement so
%%as to reduce the overhead and the measurement
%%noise.v

 \subsection{Measurements with short length}
 We have not imposed any constraint on the number of nodes in one measurement. %There is no constraint on the length of a measurement in previous sections, and we assume that a measurement can measure any number of nodes as long as they form a connected subgraph. In practice, however, one may want to reduce the number of nodes measured  in one measurement
 In practice, one may want to take short measurements so as to reduce the communication cost and the measurement noise. We next consider sparse recovery with additional constraint on measurement length, and we discuss two special graphs. % that are introduced earlier.

\subsubsection{Line and Ring}
% To recover $k$ errors associated with a line/ring with $n$ nodes, we provide a construction method using $n+1- \lfloor \frac{n+1}{k+1}\rfloor $ measurements
The construction in Section \ref{sec:line} is optimal for a line in terms of the number of measurements needed, %Theorem \ref{thm:linekcs} also shows that this number is the minimum number of measurements needed for a line  and is no greater than the minimum number of measurements needed for a ring plus one.
% However, when $k$ is a constant, the number of nodes in a measurement, or equivalently,
and the length of each measurement is $\lfloor \frac{n+1}{k+1}\rfloor$, which is proportional to $n$ when $k$ is a constant. Here we provide a different  construction such that the total number of measurements needed to recover associated  $k$-sparse vectors is $k \lceil \frac{n}{k+1}\rceil +1$, but each measurement measures at most $k+2$ nodes. %Since
%  \[  (k\lceil \frac{n}{k+1}\rceil +1)-(n+1-\lfloor\frac{n+1}{k+1}\rfloor) \leq  \max( k-1,1),\]
%and
 %\[  (k\lceil \frac{n}{k+1}\rceil +1)-(n-\lfloor\frac{n}{k+1}\rfloor) \leq k,\]
We also remark that the number of measurements by this construction is  within the minimum plus $\max( k-1,1)$ for a line, and  the minimum plus $k$ for a ring.

 We construct the measurements as follows. Given $k$, let $B^k$ be a $k+1$ by $k+1$ square matrix with  entries of `1' on the main diagonal and the first row, i.e. $B^k_{ii}=1$ %for all $i$
 and $B^k_{1i}=1$ for all $i$. If $k$ is even, %the subdiagonal entries of $B^k$ are all `1's, i.e.
 let $B^k_{i(i-1)}=1$ for all $2 \leq i\leq k+1$; if $k$ is odd,  %the subdiagonal entries of $B^k$ are  `1's except  $B^k_{(k+1)k}$, i.e.
 let $B^k_{i(i-1)}=1$ for all $2 \leq i\leq k$.  $B^k_{ij}=0$ elsewhere. Let $t=\lceil \frac{n}{k+1} \rceil$, we construct a %$n+1-\lfloor \frac{n+1}{k+1}\rfloor$ by $n$
 $(kt+1)$ by $(k+1)t$ matrix $A$ based on $B^k$.   Given set %$S\subseteq \{1,...,n+1-\lfloor \frac{n+1}{k+1}\rfloor\}$ and set $T \subseteq \{1,...,n\}$,
 $S\subseteq \{1,...,kt+1\}$ and set $T \subseteq \{1,...,(k+1)t\}$,
 $A_{ST}$ is the submatrix of $A$ with row indices in $S$ and column indices in $T$. %Let $S_1=\{1,...,2k\}$ and $T_1=\{1,...,2k\}$.
 %For all $i=1,...,t$, let $S_i=\{(i-1)k+1,...,ik+1\}$, and let $T_i=\{(k+1)(i-1)+1,...,(k+1)i \}$.
  For all $i=1,...,t$, %$i=1,...,\lfloor \frac{n}{k+1}\rfloor$,
  let $S_i=\{(i-1)k+1,...,ik+1\}$, and let $T_i=\{(k+1)(i-1)+1,...,(k+1)i \}$. Define $A_{S_i T_i}=B^k$ for all $i$. %$=1,...,\lfloor \frac{n}{k+1}\rfloor$. %Let $r=n-(k+1)\lfloor n/(k+1) \rfloor$. Let $I_r$ be an $r \times r$ identity matrix.
 All the other entries of $A$ are zeros. %For measuring   the line/ring  with $n$ nodes,
 We keep the first $n$ columns of $A$ as a measurement matrix for the line/ring  with $n$ nodes.
  Note that   the last one or serval rows of the reduced matrix can be all zeros, and we just delete such rows, and let the resulting matrix be the measurement matrix. For example, when $k=2$ and $n=9$, we have $t=3$, and
   \begin{equation*}  {\small B^2=\left[ \begin{array}{ccc}  1 & 1 & 1  \\ 1& 1 & 0 \\ 0 & 1 & 1\end{array}\right],} \end{equation*}
   and
   \begin{equation}\label{eqn:Aex1}
   {\small A= \left[ \begin{array}{ccccccccc}  1 & 1 & 1 & 0& 0& 0 & 0 & 0 & 0\\ 1 & 1 & 0 &0&0&0 &0&0&0 \\ 0 & 1 & 1 &1 &1 &1&0&0&0 \\ 0& 0&0&1 &1&0&0&0&0  \\0&0&0&0&1&1 &1&1&1 \\ 0 &0&0&  0& 0&0&1 &1&0 \\0 &0&0&  0& 0&0&0 &1&1\end{array}  \right].}
   \end{equation}
When $k=3$, and $n=8$, we have $t=2$ and
  \begin{equation*}   {\small B^3=\left[ \begin{array}{cccc}  1 & 1 & 1 &1  \\ 1& 1 & 0 & 0 \\ 0 & 1 & 1 & 0 \\ 0 & 0 & 0 & 1\end{array}\right],
   A= \left[ \begin{array}{cccccccc}  1 & 1 & 1 &1 & 0& 0& 0 &0 \\ 1 & 1 & 0 &0&0&0 & 0 &0 \\ 0 & 1 & 1 &0 &0 &0 &0 &0\\ 0& 0&0&1 &1&1&1&1 \\0&0&0&0&1&1& 0&0 \\0&0&0&0&0&1&1& 0\\0&0&0&0&0&0& 0&1   \end{array}  \right].}\end{equation*}
%In the above construction, when $k$ is even,
Each measurement measures at most $k+2$ nodes when $k$ is even and at most $k+1$ nodes when $k$ is odd. We have,
   \begin{theorem}
The above construction can recover $k$-sparse vectors associated with a line/ring  %constructed $k\lceil\frac{n}{k+1}\rceil +1$ measurements,
 %one can identify up to $k$ errors in a line/ring  of $n$ nodes
 with at most $k\lceil\frac{n}{k+1}\rceil +1$ measurements, which is within the minimum number of measurements needed plus $k$. And each measurement measures at most $k+2$ nodes.
 \end{theorem}

 \begin{proof}
 We only need to prove that all $k$-sparse vectors in $\mathcal{R}^{(k+1)t}$ can be identified with $A$, which happens if and only if for every  vector $\bfz \neq \bm 0$ such that $A\bfz =\bm 0$, $\bfz$ has at least $2k+1$ non-zero elements \cite{CaT05}.

 If $t=1$, $A$ a $k+1$ by $k+1$ full rank matrix, and the claim holds trivially.
  We next consider $t \geq 2$.  We prove the case when $k$ is even, and skip the similar proof for odd $k$. % The proof for odd $k$ is similar, and is skipped due to space constraint.

%We prove the case that $k$ is even. The proof for the case $k$ is odd is very similar, and we skip its details for simplicity.
%When $k$ is even, consider $k+1$ by $k+1$ matrix $B^k$ and the $(kt+1)$ by $(k+1)t$ matrix $A$ defined as above, and $t=\lceil \frac{n}{k+1} \rceil$. We assume that $t \geq 2$ without loss of generality. When $t=1$, the constructed matrix is a $k+1$ by $k+1$ full rank matrix, and thus can recover all the vectors. The first $n$ columns of $A$ is the measurement matrix. To prove that $k$ errors can be recovered with the measurement matrix, we only need to prove that all $k$-sparse vectors in $\mathcal{R}^{(k+1)t}$ can be identified with $A$, which happens if and only if for every non-zero vector $\bfz$ such that $A\bfz =\bm 0$, $\bfz$ has at least $2k+1$ non-zero elements.

 For each integer $t'$ in $[2, t]$, define a submatrix $A_{t'}$ formed by the first $kt'+1$ rows and the first $(k+1)t'$ columns of $A$. For example, for $A$ in (\ref{eqn:Aex1}), we define
  \begin{equation*}  {\small A_2= \left[ \begin{array}{cccccc}  1 & 1 & 1 & 0& 0& 0 \\ 1 & 1 & 0 &0&0&0 \\ 0 & 1 & 1 &1 &1 &1\\ 0& 0&0&1 &1&0 \\0&0&0&0&1&1 \end{array}  \right], \quad \textrm{and } A_3=A. }
  \end{equation*}

  We will prove by induction on $t'$ that (*) \textit{every non-zero vector $\bfz \in \mathcal{R}^{(k+1)t'}$ such that $A_{t'} \bfz=\bm 0$ holds has at least $2k+1$ non-zero elements for every $t'$ in $[2, t]$}.

 First consider $A_2$, which is a $(2k+1)\times (2k+2)$ matrix. From the last $k$ rows of $A_2$, one can easily argue that for every $\bfz$ such that $A_2 \bfz=\bm 0$, its last $k+1$ entries are either all zeros or all non-zeros. If the last $k+1$ entries of $\bfz$ are all zeros, let  $\bfz'$ denote the subvector containing the first $k+1$ entries of $\bfz$. Then we have $\bm 0=A_2 \bfz=B^k \bfz'$. Since $B^k$ is full rank, then $\bfz'=\bm 0$, which implies that $\bfz =\bm 0$.
 % Let $\bfa_i$ ($i=1,...,2k+1$) denote the $i$th row of $A_2$. For every $\bfz$ such that $A_2 \bfz=\bm 0$, $\bfa_i \bfz=0$ holds for all $i$. We consider two cases: $z_{k+2}=0$ and $z_{k+2} \neq 0$. From $\bfa_i \bfz=0$ for $i=k+2,...,2k+1$ and the construction of $A_2$, one can easily see that if $z_{k+2}$ is zero, then $z_{i}$ should be zero for $i=k+3, ...,2k+2$. In this case, the last $k+1$ entries of $\bfz$ are all zeros. Let $\bfz'$ denote the subvector containing the first $k+1$ entries of $\bfz$, then $\bm 0=A_2 \bfz=B^k \bfz'$. Since $B^k$ is full rank, then $\bfz'=\bm 0$. Therefore, if $z_{k+2}=0$, then $\bfz$ must be a zero vector.

  Now consider the case that  last $k+1$ entries of $\bfz$ are all non-zeros. Since $k+1$ is odd, the sum of these   entries is, %$z_{k+2} \neq 0$. Still from $\bfa_i \bfz=0$ for $i=k+2,...,2k+1$ and the construction of $A_2$, one can check that $z_{i}\neq 0$ for all $i=k+3, ...,2k+2$. Thus, the last $k+1$ entries of $\bfz$ are all nonzero. Moreover, we also check that
 \begin{equation}\label{eqn:med1}
  \sum_{i = k+2}^{2k+2}z_i = z_{k+2} \neq 0.
  \end{equation}
 Let $\bfa_i^T$ ($i=1,...,2k+1$) denote the $i$th row of $A_2$. We have %Since from $\bfa_{k+1} \bfz=0$, we know that
 \begin{equation}
 \label{eqn:med2}
 \bfa_{k+1}^T \bfz= \sum_{i = k}^{2k+2} z_i= 0.
 \end{equation}
 Combining (\ref{eqn:med1}) and (\ref{eqn:med2}), we know that
 \begin{equation}\label{eqn:med3}
 z_{k}+z_{k+1}= - z_{k+2} \neq 0.
 \end{equation}
Thus, at least one of $z_{k}$ and $z_{k+1}$ is non-zero.
Combining (\ref{eqn:med3}) with $\bfa_1^T \bfz=0$, we have
 %\begin{equation}\label{eqn:med4} \sum_{i = 1}^{k+1} = 0.\end{equation}
  %Combining (\ref{eqn:med3}) and (\ref{eqn:med4}),
     one of the first $k-1$ entries of $\bfz$ is non-zero. From $\bfa_i ^T\bfz=0$ for $2 \leq i\leq k-1$, % and the construction of $A_2$,
  one can argue that if one of the first $k-1$ entries of $\bfz$ is non-zero, then all the first $k-1$ entries are non-zero. Therefore, %in the case that  $z_{k+2} \neq 0$,
  %all the first $k-1$ entries and the last $k+1$ entries of $\bfz$ are non-zero, and at least one of  $z_{k}$ and $z_{k+1}$ is nonzero. Thus,
  $\bfz$ has at least $2k+1$ nonzero entries.  (*) holds for $A_2$. %We have thus proved that if $A_2 \bfz=\bm 0$, then $\bfz$ either is a zero vector or has at least $2k+1$ nonzero entries.

 %We have proved the base case $A_2$.
 Now suppose (*) holds for  some $t'$ in $[2, t-1]$. Consider % We next need to prove that the claim holds for
 matrix $A_{t'+1}$. Same as the arguments for $A_2$, one can show that for every  $\bfz \neq \bm 0$ %Let $\bfa_i$ denote the $i$th row of $A_{t'+1}$. For every vector $\bfz$
 such that $A_{t'+1} \bfz=\bm0$, its last $k+1$ entries are either all zeros or all non-zero.  %there are two possibilities, either $z_{(k+1)t'+1}=0$ or $z_{(k+1)t'+1}\neq 0$. If $z_{(k+1)t'+1}=0$, same as the case for $A_2$, one can argue that $z_{i}=0$ for all $i=(k+1)t'+2,...,(k+1)(t'+1)$. The last $k+1$ entries of $\bfz$ are all zeros.
 In the former case,  let $\bfz'$ denote the subvector containing the first $(k+1)t'$ entries of $\bfz$. %Then $\bm 0=A_{t'+1} \bfz=A_{t'} \bfz'$.
 By induction hypothesis, $\bfz'$ %either is a zero vector or
  has at least $2k+1$ nonzero entries, thus so does $\bfz$. %Then $z_{(k+1)t'+1}=0$, $\bfz$ also has either all zeros or at least $2k+1$ nonzero entries.
 %Now consider the case that $z_{(k+1)t'+1}\neq 0$. Same as the arguments for the $A_2$ case, one can check that the

 If the last $k+1$ entries of $\bfz$ are all non-zero, like in the $A_2$ case, we argue that the sum of $z_{(k+1)t'-1}$ and $z_{(k+1)t'}$ is non-zero, which implies that at   least one of them % $z_{(k+1)t'-1}$ and $z_{(k+1)t'}$
 is non-zero. Also consider %We claim that there exist $j$ with $ 1\leq j \leq t'$ such that
% \begin{equation}\label{eqn:med5}\sum_{i=(j-1)(k+1)+1}^{(j-1)(k+1)+k-1} z_i \neq 0 \end{equation}.
 %Suppose not, then for every $ 1\leq j \leq t'$, \begin{equation}\label{eqn:med6}\sum_{i=(j-1)(k+1)+1}^{(j-1)(k+1)+k-1} z_i =0 \end{equation}.
%Combining these equations with
  $\bfa_i^T\bfz=0$ with $i=rk+1$ for every integer $r$ in $[0,t'-1]$, one can argue that there exist $j$ in $[0, t'-1]$ such that
  the sum of all $k-1$ entries from $z_{j(k+1)+1}$ to $z_{j(k+1)+k-1}$ is non-zero.
 %\begin{equation}\label{eqn:med5}\sum_{i=(j-1)(k+1)+1}^{(j-1)(k+1)+k-1} z_i \neq 0. \end{equation}
 Then, from
  % sequentially from $j=1$ to $t'$ that   $z_{(k+1)j-1}+z_{(k+1)j}=0$. This contradicts the above observation that $ z_{(k+1)t'-1}+ z_{(k+1)t'}$ is nonzero. Therefore, the claim follows, and (\ref{eqn:med5}) holds for some $j$.
%From the construction of $A$, and specifically
$\bfa_i^T\bfz=0$ for $i=jk+2,..., jk+k-1$, we know that if the sum of $z_{j(k+1)+1}$ to $z_{j(k+1)+k-1}$ is non-zero, every entry is non-zero. %if one of the entries in the sum of (\ref{eqn:med5}) is non-zero, then all the entries are non-zero. Therefore, all $k-1$ entries from $z_{(j-1)(k+1)+1}$ to $z_{(j-1)(k+1)+k-1}$ are nonzero. %Since we obtained earlier that at least one of  $z_{(k+1)t'-1}$ and $z_{(k+1)t'}$ is nonzero, and the last $k+1$ entries of $\bfz$ are all nonzero,
We conclude that in this case $\bfz$ also has at least $2k+1$ nonzero entries.

By induction over $t'$, every $\bfz \neq \bm 0$ such that $A\bfz=\bm 0$ has at least $2k+1$ non-zero entries, then the result follows.
%
%When $k$ is odd, the proof is basically the same, the only difference in the arguments is that instead of having at least one of  $z_{(k+1)t'-1}$ and $z_{(k+1)t'}$ is nonzero, we have $z_{(k+1)t'}$ is nonzero.
 \end{proof}

%The above construction can recover $k$ errors in a line/ring  with $n$ nodes by $k\lceil\frac{n}{k+1}\rceil +1$ measurements, which is within the minimum plus $k$. The number of nodes measured by each measurement is at most $k+2$.  %In Section \ref{sec:line}, we provide a measurement construction method with each measurement measuring $n+1-\lfloor \frac{n+1}{k+1}\rfloor$ nodes. In order to reduce the length of each measurement, given $k$ and $n$, one can choose a construction method that measures a smaller number of nodes in each measurement. However,
%When $k$ is large, say, $k=\Theta(\log n)$, the two methods measure $\Theta(\log n)$ and $\Theta(n/\log n)$ nodes respectively in each measurement, which are still large when $n$ is large.
This construction measures at most $k+2$ nodes in each measurement. If measurements with constant length are preferred, we provide another construction method such that every measurement only measures at most three nodes. This method requires more measurements,  $(2k-1)\lceil\frac{n}{2k}\rceil +1$ measurements %, instead of $k\lceil\frac{n}{k+1}\rceil +1$ measurements,
to recover $k$-sparse vectors associated with a line/ring.

 Given $k$, let $D^k$ be a $2k$ by $2k$ square matrix having  entries of `1' on the main diagonal and the subdiagonal and `0' elsewhere, i.e. $D^k_{ii}=1$ for all $i$ and $D^k_{i(i-1)}=1$ for all $i\geq 2$, and $D^k_{ij}=0$ elsewhere. Let $t=\lceil \frac{n}{2k} \rceil$, we construct a $(2kt-t+1)$ by $2kt$ matrix $A$ based on $D^k$.  % Given set $S\subseteq \{1,...,2kt-t+1\}$ and set $T \subseteq \{1,...,2kt\}$, $A_{ST}$ is the submatrix of $A$ with row indices in $S$ and column indices in $T$. %Let $S_1=\{1,...,2k\}$ and $T_1=\{1,...,2k\}$.
 Let $S_i=\{(i-1)(2k-1)+1,...,i(2k-1)+1\}$, and let $T_i=\{2k(i-1)+1,...,2ki \}$. Define $A_{S_i T_i}=D^k$ for all $i=1,...,t$, and $A_{ij}=0$ elsewhere. We keep the first $n$ columns of $A$ as the measurement matrix. For example, when $k=2$ and $n=8$, we have
  \begin{equation*}  {\small D^2=\left[ \begin{array}{cccc}  1 &0 & 0& 0 \\ 1& 1 & 0 & 0 \\ 0 & 1 & 1 & 0\\ 0 & 0 & 1 & 1\end{array}\right],} \end{equation*}
   and
   \begin{equation}\label{eqn:Aex2}
   {\small A= \left[ \begin{array}{cccccccc}  1 & 0 & 0 & 0& 0& 0 & 0 & 0 \\ 1 & 1 & 0 &0&0&0 &0&0 \\ 0 & 1 & 1 &0 &0 &0&0&0 \\ 0& 0&1&1 &1&0&0&0  \\0&0&0&0&1&1 &0&0 \\ 0 &0&0&  0& 0&1&1 &0 \\0 &0&0&  0& 0&0&1&1\end{array}  \right].}
   \end{equation}
 % In fact, we have
  \begin{theorem}
 The above constructed $(2k-1)\lceil\frac{n}{2k}\rceil +1$ measurements  can identify $k$-sparse vectors associated with a line/ring  of $n$ nodes, and each measurement measures at most three nodes.
 \end{theorem}
\begin{proof}
%Consider the $(2kt-t+1)$ by $2kt$ matrix $A$ defined as above where $t=\lceil \frac{n}{2k} \rceil$. We assume that $t \geq 2$ without loss of generality. When $t=1$, the matrix constructed is a $2k$ by $2k$ full rank matrix, and thus can recover all the vectors. The first $n$ columns of $A$ is the measurement matrix. To prove that $k$ errors can be recovered with the measurement matrix, we only need to prove that all $k$-sparse vectors in $\mathcal{R}^{2kt}$ can be identified with $A$, which happens if and only if for every non-zero vector $\bfz$ such that $A\bfz =\bm 0$, $\bfz$ has at least $2k+1$ non-zero elements.
When $t=1$, $A$ is a full rank square matrix. We focus on the case that $t \geq 2$.
For each integer $t'$ in $[2, t]$, define a submatrix $A_{t'}$ formed by the first $2kt'-t'+1$ rows and the first $2kt'$ columns of $A$. We will prove by induction on $t'$ that every  $\bfz \neq \bm 0$ such that $A_{t'} \bfz=\bm 0$ holds has at least $2k+1$ non-zero elements for every $t'$ in $[2, t]$.

First consider $A_2$. For $A$ in (\ref{eqn:Aex2}), $A_2=A$. From the first $2k-1$ rows of $A_2$, one can check that for every $\bfz$ such that $A_2 \bfz=\bm 0$, its first $2k-1$ entries are zeros. % Let $\bfa_i$ denote the $i$th row of $A_2$. For every $\bfz$ such that $A_2 \bfz=\bm 0$, $\bfa_i \bfz=0$ holds for all $i=1,...,4k-1$. Since $\bfa_1 \bfz=z_1$, then $z_1=0$. Since $\bfa_2\bfz=z_1+z_2=0$ and $z_1=0$, then we have $z_2=0$. Similarly we can argue that $z_i=0$ for every $i=1,..., 2k-1$, for every $\bfz$ such that $A_2 \bfz=\bm 0$. Note that \[\bfa_{2k}\bfz=z_{2k-1}+z_{2k}+z_{2k+1}=0\] and $z_{2k-1}=0$, then we know
  From the $2k$th row of $A_2$, we know that $z_{2k}$ and $z_{2k+1}$ are either both zeros or both non-zero. In the former case,  the remaining $2k-1$ entries of $\bfz$ must be zeros, thus, $\bfz=\bm0$. In the latter case, one can check that   the remaining $2k-1$ entries are all non-zero, and therefore $\bfz$ has $2k+1$ non-zero entries.  %If $z_{2k}$ and $z_{2k+1}$ are both zero, considering $\bfa_i \bfz=0$ for every $i=2k+1, ..., 4k$, we know that $z_{i}=0$ for every $i$ by similar arguments, thus $\bfz=\bm0$. If $z_{2k}$ and $z_{2k+1}$ are both non-zero, since $\bfa_{2k+1}\bfz=z_{2k+1}+z_{2k+2}=0$, then $z_{2k+2}\neq 0$. We can argue one by one that $z_{i} \neq 0$ for every $i \geq 2k+3$. Then $\bfz$ has $2k+1$ non-zero entries from index $2k$ to $4k$. Therefore, for every non-zero $\bfz$ such that $A_2 \bfz=\bm0$, $\bfz$ has $2k+1$ non-zero elements.

Now suppose the claim holds for some $t'$ in $[2, t-1]$. %We next need to prove that the claim holds for matrix $A_{t'+1}$. Let $\bfa_i$ denote the $i$th row of $A_{t'+1}$. For every
Consider vector $\bfz \neq \bm 0$ such that $A_{t'+1} \bfz=\bm0$. If  $z_{2kt'+1}=0$, it is easy to see that  the last $2k$ entries of $\bfz$ are all zeros. Then by induction hypothesis, at least  $2k+1$ entries of the first $2kt'$ elements of $\bfz$ are non-zero. If  $z_{2kt'+1}\neq 0$, one can check that the last $2k-1$ entries of $\bfz$ are all non-zero, and at least one of  $z_{2kt'-1}$ and $z_{2kt'}$ is non-zero. Thus, $\bfz$ also has at least $2k+1$ non-zero entries in this case. %there are two possibilities, either $z_{2kt'+1}=0$ or $z_{2kt'+1}\neq 0$. If $z_{2kt'+1}=0$, then $\bfa_{2kt'-t'+2}\bfz=z_{2kt'+1}+z_{2kt'+2}=0$, then $z_{2kt'+2}=0$. Since $\bfa_{i}\bfz=0$ for all $i =2kt'-t'+3,..., 2k(t'+1)-t'$, we can obtain that $z_{i}=0$ for all $i \geq 2kt'+3$ by similar arguments. Then the last $2k$ entries of $\bfz$ are all zero. Let vector $\bfz'$ contain the first $2kt'$ entries of $\bfz$, then $A_{t'} \bfz'=\bm0$. By induction hypothesis, either $\bfz'=\bm0$ or $\bfz'$ has at least $2k+1$ non-zero entries. Thus, $\bfz$ is either $\bm0$ or contains at least $2k+1$ non-zero entries. Now consider the case that $z_{2kt'+1}\neq 0$. Since \[\bfa_{2kt'-t'+1}\bfz=z_{2kt'-1}+z_{2kt'}+z_{2kt'+1}=0,\] then at least one of  $z_{2kt'-1}$ and $z_{2kt'}$ is non-zero. Since $\bfa_{2kt'-t'+2}\bfz=z_{2kt'+1}+z_{2kt'+2}=0$, then $z_{2kt'+2}\neq0$. We can argue that $z_{i} \neq 0$ for all $i \geq 2kt'+3$. Thus, $\bfz$ has at least $2k+1$ non-zero entries.

By induction over $t'$, every $\bfz \neq \bm 0$ such that $A\bfz=\bm 0$ has at least $2k+1$ non-zero entries, then the theorem follows.
\end{proof}

The number of measurements %5needed to identify $k$ errors
by this construction %method is  $(2k-1)\lceil\frac{n}{2k}\rceil +1$, which
is greater than those of the previous methods. %the number of measurements by our method in Section \ref{sec:line} and the method at the beginning of this section. However,
But the advantage of this construction is that the number of nodes in each measurement is at most three, no matter how large $n$ and $k$ is.

\subsubsection{Ring with each node connecting to four  neighbors}
We next consider   $\mathcal{G}^4$ in Fig. \ref{fig:ring4} (a). %where each node connects to four closet neighbors, as shown in Fig. \ref{fig:ring4} (a). Compared with the construction in Section \ref{sec:ring4},
 We further impose the constraint that the number of nodes in each measurement cannot exceed $d$ for some predetermined integer $d$. We  neglect  $\lfloor \cdot \rfloor$ and $\lceil \cdot \rceil$ for notational simplicity. %We assume $d>k$ here. %In Section \ref{sec:ring4}, our deterministic construction method uses all the even nodes as a hub to measure the odd nodes and vice versa, and therefore each measurement measures at least $n/2$ nodes. The random construction method based on Markov chain also measures at least $n/2$ nodes in a measurement. Therefore, for any $d<n/2$, our previous construction methods do not apply.

%We still use some even nodes as a hub to measure the odd nodes that are directly connected to them. But in order to make sure that the number of nodes in each measurement does not exceed $d$,  each hub contains at most $\lfloor d/2 \rfloor$ nodes, and the number of nodes measured by each hub is also at most $\lceil d/2 \rceil$. Then
All the even nodes are divided into $n/d$ groups. Each group contains $d/2$ consecutive even nodes and is used as a hub to measure $d/2$ odd nodes that have direct edges with nodes in the hub. %Then we can recover $k$ errors in the  $d/2$ odd nodes with $M^C_{k,d/2}+1$ measurements using the corresponding even nodes as a hub.
Then we can identify the values related to all the odd nodes with   $nM^C_{k,d/2}/d +n/d$ measurements, and the number of nodes in each measurement does not exceed $d$. We then measure the even nodes with groups of odd nodes as hubs. %, and the number of measurements needed is $nM^C_{k,d/2}/d$. % since we know the values of all the odd nodes and do not need to measure the hubs separately.
In total, the number of measurements is $2nM^C_{k,d/2}/d +2n/d$, which is $O(2kn \log(d/2)/d)$. When $d$ equals to $n$, the result coincides with Theorem \ref{thm:ring4}. Since $n/d$ measurements are needed to measure each node  at least once,   we have
\begin{theorem}
The number of  measurements needed to recover $k$-sparse vectors associated with $\mathcal{G}^4$ with each measurement containing at most $d$ nodes is lower bounded by $n/d$, and upper bounded by $O(2kn \log(d/2)/d)$.
\end{theorem}
The ratio of the number of measurements by our construction to the minimum number needed with length constraint is within $Ck\log(d/2)$ for some constant $C$.

\subsection{Measurements passing at least one node in a fixed subset}
%\subsection{Limited number of agents}
\begin{figure}[h]
\begin{center}
\includegraphics[scale=0.2]{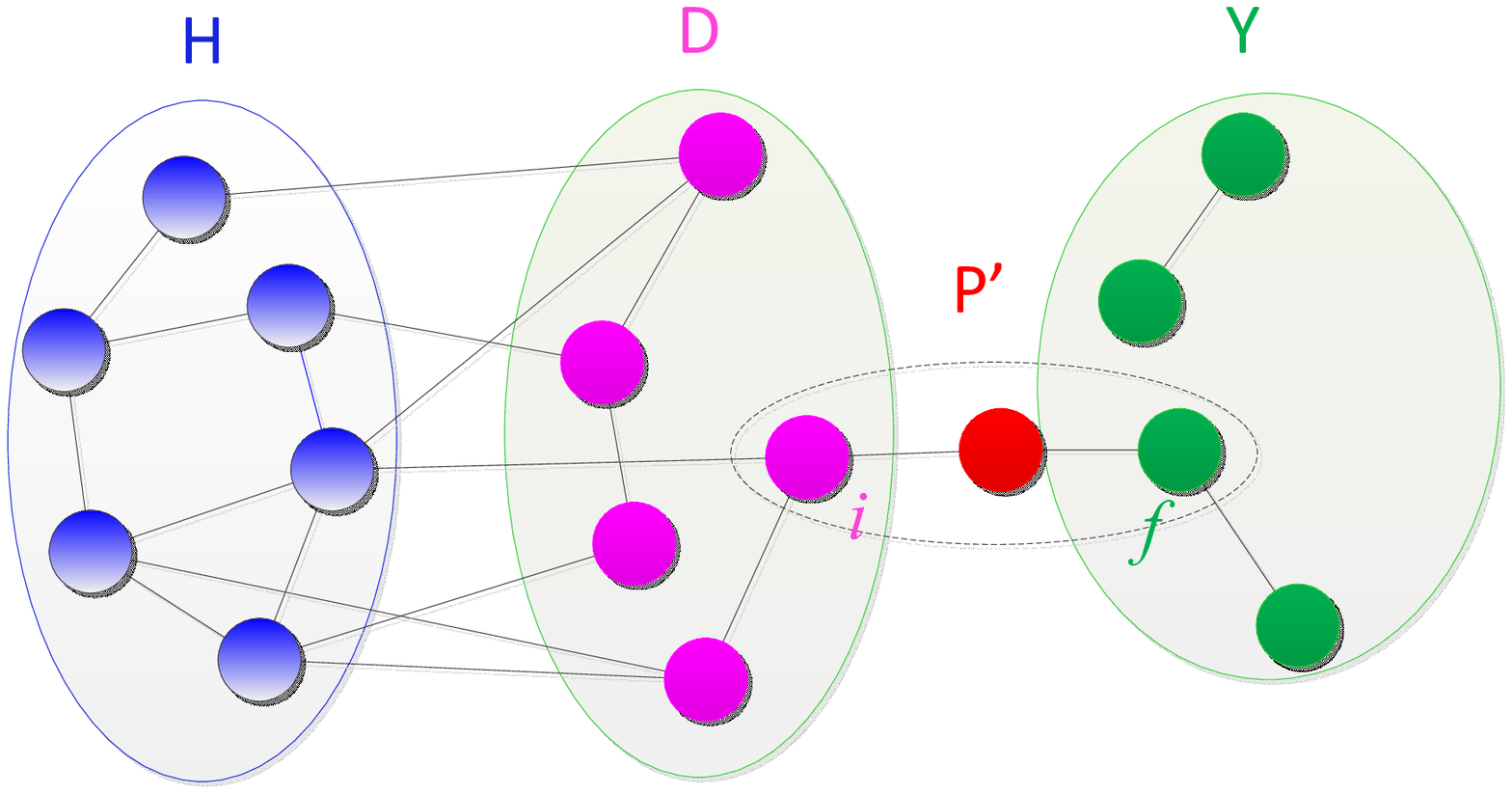}
  \caption{When $H \cap Y =\Phi$, use  hub $H'=H \cup P'$ to measure nodes $D\backslash i$}\label{fig:newhub}
  \end{center}
\end{figure}

Recall that in network delay monitoring, a router sends a probing packet to measure the sum of delays on links that the packet transverses. Then every measurement initiated by this router measures the delay on at least one link that is connected to the router. In order to reduce the monitoring cost,  one may only employ several routers to initiate measurements, thus, each measurement would include at least one link that is connected to these routers. In the graph model $G=(V,E)$ we consider in this paper, it is equivalent to the requirement that every  measurement should contain at least one node in a fixed subset of nodes $Y \subset V$. We will show that this requirement can be achieved with small modifications to Algorithm 1.

%Recall that we assume that one agent is chosen to monitor one aggregate measurement. % for a connected subgraph $G_S$, %in assumption (A1), in order to measure the sum of nodes in set $S$ such that $G_S$ is connected, an
%agent $u$ in $S$ aggregates the sum of all the nodes in $S$, and sends the measurement to the central operator.
%Previously we assume that a measurement can measure any subset of nodes as long as they form a connected subgraph. Recall that the value of each measurement is aggregated at one node in the set of nodes that are measured. We call such a node which finally obtains the value of a measurement as an ``agent''.
%In order to recover the sparse vector,
%The central operator then asks  all the agents for the measurements and run the recovering algorithm.
%One may want to reduce the total number of agents in a network so as to reduce the monitoring and transmission costs. If only a small set $Y$ of nodes can perform as the agents, %that requires each measurement passing at least one of the agents. Therefore, in this setup, we further require that
%every constructed measurement is then required to pass at least one node in $Y$.

%This can be achieved with small modification to the construction algorithm.  %all the measurements will contain at least one node in a fixed subset of nodes.
After step 3 in Algorithm 1, %Let $Y$ denote the set of agents.
let $H$ denote the currently chosen hub, and let $D$ denote the set of nodes that one needs to design measurements via hub $H$. % i.e., $G_H$ is connected, and for every node $i$ in $D$, there exists some node $j$ in $H$ such that edge $(i,j)$ exists. %There are two cases to consider. Case one:
If $H \cap Y$ is not empty,  %That means the hub $H$ contains at least one node in $F$.
since every measurement constructed to measure nodes in $D$ should contain all the nodes in the $H$,
 it contains at least one node in $Y$ automatically.  If $H \cap Y$ is empty, we want to find a new hub that contains at least one node in $Y$. If there exists a path $P$ in $G$ from some node $j$ in $H$ to some node $f$ in $Y$ such that $P$ does not contain any node in $D$, then let  $\hat{H}:=H \cup P$ be the new hub, and design measurements for $D$ using hub $\hat{H}$. Then every measurement contains all nodes in $\hat{H}$ and thus the node $f$. If such a path does not exist, pick any node $i$ in $D$ and any node $f$ in $Y$, find the shortest path $P'$ between $i$ and $f$.  %Since the graph is connected, then every node $f$ in $F$ is connected to some node in $H$ via some path. Since every node in $D$ is directly connected to some node in $H$, then for every $f$ in $F$, there always exists a path $P$ from $f$ to some node $j$ in $H$ such that $P$ contains at most one node in $D$, let $i^*$ denote such a node if exists.
 Let $H':=H \cup P'$ be the hub, and let $D':=D\backslash i$ be the set of nodes that can be measured via $H'$, see Fig. \ref{fig:newhub}. %Since $H'$ contains $f$ in $F$, then with this simple modification,
 Then every measurement containing hub $H'$ contains $f$. %, and the number of nodes that can be measured together is only one less than the number of nodes measured in the original method.
 Since node $i$ belongs to $H'$, we need two additional measurements passing $f$ to measure it. One measures    $P'$, and the other measures $P'\backslash i$. With this simple modification, we can measure  nodes in $D$ with each measurement containing  one node in $Y$, and the total number of measurements increases by at most two.
 \floatname{algorithm}{Subroutine}
\setcounter{algorithm}{2}
 \begin{algorithm}
\begin{algorithmic}[1]
\REQUIRE hub $H$, set $D$ of nodes to measure, set $Y$ of fixed nodes, $G$
 \IF {$H \cap Y \neq \Phi$}
 \STATE Design $f(k,|D|)+1$ measurements to recover $k$-sparse vectors associated with $D$ using $H$ as a hub.
 \ELSE
 \STATE Find the shortest path between every node in $H$ and every node in $Y$.
 \IF {there exists a shortest path $P$ s.t. $P \cap D = \Phi$}
 \STATE  Design $f(k,|D|)+1$ measurements to recover nodes in $D$ using $\hat{H}=H \cup P$ as a hub.
\ELSE
\STATE pick a  node $i$ in $D$ and a node $f$ in $Y$, find the shortest path $P'$ between $i$ and $f$.
\STATE  $D':=D\backslash i$, $H':=H \cup P'$, design  $f(k,|D'|)+1$ measurements to recover  $D'$ with $H'$ as a hub.
\STATE Measure $P'$ and $P'\backslash i$ to recover node $i$.
 \ENDIF
  \ENDIF
 %\RETURN $G$, $H$
 \end{algorithmic}
 \caption{\textbf{Agent}($H$, $D$, $Y$, $G$)}
 \end{algorithm}

We summarize the above modification in subroutine \textbf{Agent}. %($H$, $D$, $Y$, $G$).
For measurement design on general graphs, %with additional constraint that every measurement passes at least one agent in $Y$,
we first replace step 4 in Algorithm 1 in Section \ref{sec:algo} with subroutine \text{Agent($V\backslash S$, $S$, $Y$, $G$)}. %measurements on general graphs. In each iteration, we can apply the above modification to make every measurement contain at least one node in set $F$.
Then in each iteration the number of measurements is increased by at most two. We then replace step 9 %To measure the last node $i$ in the network after several iterations of reduction, we simple need to
with measuring the paths $P^*$ and $P^*\backslash n_{\textrm{last}}$, where $n_{\textrm{last}}$ is the last node in $G$, and $P^*$ connects $n_{\textrm{last}}$ to any node $j$ in $Y$ on the original graph. Therefore, the total number of measurements needed %to recover $k$ errors in a network with $n$ nodes
by the modified algorithm is upper bounded by $Rf(k,n)+3R+2$, and each measurement in the modified version contains at least one node in  $Y$.

 \section{Sensitivity to  hub measurement errors}\label{sec:huberror}
 In constructions based on the use of a hub, in order to measure nodes in $S$ using hub $H$, we
%The key idea in design measurements for a set $S$ of nodes using a hub $H$ is that %on graphs is that we can use a connected subset $H$ of nodes as a hub to freely measure the set $S$ of nodes that are directly connected to the hub. We
%measure the hub with one measurement.
%the sum of any subset of nodes in $S$ is obtained by
first measure the sum of nodes in $H$, and then delete it from other measurements to obtain the sum of some subset of nodes in $S$.  This arises the issue that if the sum of $H$ is not measured  correctly, %and all the measurements we take over $S$ using $H$ a hub are correct,
this single error would be introduced into all the measurements. % since we need to delete the erroneous measurement of $H$ from every other measurement. %in order to recover the unknown vector associated with $S$.
Here we prove that successful recovery is still achievable when a hub measurement is erroneous.

Mathematically, let $\bfx_S$ denote the sparse vector associated with $S$, and let $\bfx_H$ denote the vector associated with $H$ and let $A^{m\times |S|}$ be a measurement matrix that can identify $k$-sparse vectors associated with a complete graph of $|S|$ nodes. We arrange the vector $\bfx$ such that $\bfx=[\bfx_S^T \quad \bfx_H^T]^T$, then
\[F= \left[ \begin{array}{cc}
A &  W^{m \times |H|} \\
\bm0_{|S|}^T & \bm1_{|H|}^T
\end{array}
 \right] \]
  is the measurement matrix for detecting $k$ non-zeros in $S$ using hub $H$, where  $W$ is %an $m$ by $|H|$
  a matrix with all `1's,  %$\bm0_{|S|}$ is a $|S|$-dimensional column vector of all `0's, and $\bm1_{|H|}$ is a $|H|$-dimensional column vector of all `1's.
  $\bm0$ is a  column vector of all `0's, and $\bm1$ is a   column vector of all `1's.
   Let vector $\bfz$ denote the first $m$ measurements, and let $z_0$ denote the last measurement of the hub $H$. Then
  \[\left[ \begin{array}{c}
 \bfz \\
z_0
\end{array}
 \right] = \left[ \begin{array}{c}
A\bfx_S+\bm1^T \bfx_H \bm 1_m \\
\bm1^T \bfx_H
\end{array}
 \right],\] or equivalently
  \begin{equation}\label{eqn:recoverhub}\bfz-z_0 \bm 1_m=A\bfx_S.\end{equation}
    % We choose $A$ and the recovery method such that $\bfx_S$ can be correctly covered given $\bfz$ and $z_0$. % and $A$ by the Compressed Sensing theory.
     If there is some error $e_0$ in the last measurement, i.e., instead of $z_0$, the actual measurement we obtain is
     \[\hat{z}_0= \bm1^T \bfx_H+ e_0,
    \]
    %then when recovering $\bfx_S$, we have
    %\[A\bfx_S=\bfz-\hat{z}_0 \bm1_m= \bfz-z_0\bm1_m-e_0\bm1_m.\]
    % Then error $e_0$ in the hub measurement can in fact lead to errors in every other measurement of nodes in $S$, and finally the error in the recovery of $\bfx_S$.
$e_0$ hurts the recovery accuracy of $\bfx_S$ through (\ref{eqn:recoverhub}).

 To eliminate the impact of $e_0$,  %the errors in the hub measurements on the recovery accuracy,
 we model it as an entry of an augmented sparse signal to recover. %the errors in the hub measurements as entries of an augmented sparse signal to recover. In the above example,
 Let $\bfx'=[\bfx^T \quad e_0]^T$, and $A'=[A \quad -\bm1_m]$,  %let $F'=[F \quad \bfe_{m+1}]$, where $\bfe_i$ is a column vector with `1' on the $i$th entry and `0' elsewhere. Then the measurements and the augmented signal are related by
% \[\left[ \begin{array}{c}
% \bfz \\
%\hat{z}_0
%\end{array}
% \right] = F' \bfx'=F\bfx+ e_0 \bfe_{m+1} =\left[ \begin{array}{c}
%A\bfx_S+\bm1^T \bfx_H  \\
%\bm1^T \bfx_H+e_0
%\end{array}
% \right].\]
% %$[\bfz^T \hat{z}_0]^T=A'\bfx'$.
% When recovering $\bfx_S$, we delete the sum of the hub from every other measurement, and the obtained equations are \[A\bfx_S-e_0 \bm1_m=\bfz-\hat{z}_0 \bm1_m,\] with the equivalent matrix form
we have
\begin{equation}\label{eqn:recoverhub2}
A' \bfx' =\bfz-\hat{z}_0 \bm1_m.
 \end{equation}
 %where $A'=[A \quad -\bm1_m]$, and $\bfx'_S=[\bfx_S^T \quad e_0]^T$. We know that with the measurement matrix $A$, one can recover all $k$-sparse $|S|$-dimensional vectors $\bfx_S$, but with the erroneous hub measurement,
Then, recovering $\bfx_S$   in the presence of hub error $e_0$ is equivalent to recovering $k+1$-sparse vector $\bfx'$ from (\ref{eqn:recoverhub2}). %We next show that %with the measurement matrix $A'$? If the answer is yes, then with the measurement design method we proposed earlier, by augmenting the measurement matrix and the sparse vector to recover, we can easily recover the sparse vector together with the errors, if any,  in the hub measurements. We next show that
%one can indeed correctly recover $\bfx'_S$ from (\ref{eqn:recoverub2}). %under certain conditions, the statement is indeed true, and with the same way of measurement construction as we proposed earlier, one can recover sparse vectors with the presence of errors in hub measurements.

% Since the recovery performance varies for different recovery methods, we focus on the $\ell_1$-minimization method the widely used in Compressed Sensing. Given measurement matrix $A$ and the measurements $\bfy=A\bfx$, $\ell_1$-minimization returns vector $\bfx^*$ with the least $\ell_1$-norm among all the vectors $\bfx'$ such that $A\bfx'=\bfy$ and uses $\bfx^*$ as an estimate the unknown vector $\bfx$. The following lemma provides the equivalent null space condition of successful sparse recovery via $\ell_1$-minimization when the hub error exists.
%
% \begin{lemma}\label{lem:nullspace}
% Given the augmented matrix $A'=[A \quad -\bm1_m]$, $\ell_1$-minimization  successfully recovers $k$-sparse vectors $\bfx_S \in \mathcal{R}^n$ in the presence of some unknown error $e_0$ in the hub measurement if and only if for every non-zero vector $\bfw$ such that $A'\bfw=0$, and for every set $T \subseteq \{1, ..., n\}$ with $|T| \leq k$, it holds that
% \begin{equation}\nonumber
% \|\bfw_T\|_1 + |w_{n+1}| \leq \|\bfw_{T^c}\|_1,
% \end{equation}
% where $T^c=\{1,...,n\}\backslash T$.
% \end{lemma}

We consider one special construction of matrix $A^{m \times |S|}$ for a complete graph. % with $|S|$ nodes.
$A$ has `1' on every entry in the last row, and takes value `1' and `0' with equal probability independently for every other entry. $A'=[A \quad -\bm1_m]$, let $\hat{A}$ be the submatrix of the first $m-1$ rows of $A'$. Let $\bfy=\bfz-\hat{z}_0 \bm 1_m$, and let $\hat{\bfy}$ denote the first $m-1$ entries of $\bfy$.  We have,
\begin{equation}\nonumber %\label{eqn:recover11}
 (2\hat{A}-W^{(m-1)\times |S|})\bfx' = 2\hat{\bfy}-y_m.
\end{equation}
We recover $\bfx'$ by solving the  $\ell_1$-minimization problem,
\begin{equation}\label{eqn:recover1}
\min \|\bfx\|_1, \quad \textrm {s.t. } (2\hat{A}-W^{(m-1)\times |S|})\bfx = 2\hat{\bfy}-y_m.
\end{equation}
%where $\|\bfx\|_1=\sum_i |x_i|$. %We have

 \begin{theorem}\label{thm:recovery}
 With the above construction of $A$, when $m \geq C(k+1)\log |S|$ for some constant $C>0$ and $|S|$ is large enough, with probability at least $1-O(|S|^{-\alpha})$ for some constant $\alpha>0$, $\bfx'$ is the unique solution to (\ref{eqn:recover1}) for all $k+1$-sparse vectors $\bfx'$ in $\mathcal{R}^{|S|+1}$.
 \end{theorem}

Theorem \ref{thm:recovery} indicates that even though the hub measurement is erroneous, one can still identify $k$-sparse vectors associated with  $S$  with $O((k+1)\log |S|)$ measurements.

%  The recovery performance also varies for different measurement construction methods, here we consider the random measurement construction for complete graphs in which that every node is included in a measurement independently with probability 0.5, and every measurement is independent of each other.  Mathematically, $P(A_{ij}=1)=0.5$ and $P(A_{ij}=0)=0.5$ independently for every $i$ and $j$. Let the number of the randomly chosen measurements be $m=O(k\log n)$, and we will choose the scaling constant later. %We assume that $A$ has one row with all `1's. If this is not the case,
% We also add one row to $A$ with all `1's, and it only increases the number of measurements by one. For such a measurement matrix $A$ for complete graphs, we have the following result.
%
% \begin{theorem}\label{thm:recovery}
% Given $n$ nodes that can be measured freely via one hub, if the measurement matrix $A^{(m+1)\times n}$ with $m=O(k\log n)$ for a complete graph has one row of all `1's, and every other entry independently takes value `1' or `0' with equal probability, then with probability at least $1-O(n^{-\alpha})$ for some constant $\alpha>0$, $\ell_1$-minimization can successfully recover all $k$-sparse vectors in $\mathcal{R}^n$ even if the hub measurement is erroneous.
% \end{theorem}
The proof of Theorem \ref{thm:recovery} relies heavily on Lemma \ref{lem:rip}. % so we first state it as follows.

\begin{lemma}\label{lem:rip}
If matrix $\Phi^{p \times n}$ takes value $-1/\sqrt{p}$ on every entry in the last column and takes value $\pm 1/\sqrt{p}$ with equal probability independently on every other entry, then for any $\delta>0$, there exists some constant $C$ such that when $p\geq C(k+1) \log n$ and $n$ is large enough, with probability at least $1-O(n^{-\alpha})$ for some constant $\alpha>0$ it holds that for every   $U\subseteq \{1,...,n\}$ with $|U|\leq 2k+2$ and for every   $\bfx \in \mathcal{R}^{2k+2}$,
\begin{equation}\label{eqn:rip}
(1-\delta)\|\bfx\|^2_2 \leq \|\Phi_U \bfx \|_2^2 \leq (1+\delta)\|\bfx\|_2^2,
\end{equation}
where $\Phi_U$ is the submatrix of $\Phi$ with column indices in $U$.
%holds simultaneously.
\end{lemma}

\begin{proof}
Consider  matrix $\Phi'^{p\times n}$ with each entry taking value  $\pm 1/\sqrt{p}$ with equal probability independently. For every realization of matrix $\Phi'$, construct a matrix $\hat{\Phi}$ as follows. For every $i \in \{1,...,p\}$ such that $\Phi'_{in}=1/\sqrt{p}$, let $\hat{\Phi}_{ij}=-\Phi'_{ij}$ for all $j=1,...,n$. Let $\hat{\Phi}_{ij}=\Phi'_{ij}$ for every other entry. One can check that $\hat{\Phi}$ and $\Phi$ follow the same probability distribution. Besides, according to the construction of $\hat{\Phi}$, for any subset $U\subseteq \{1,...,n\}$,
\begin{equation}\label{eqn:flip}
{\Phi'_U}^{T} \Phi'_U={\hat{\Phi}_U}^{T}\hat{\Phi}_U.
\end{equation}

The Restricted Isometry Property \cite{CaT05} indicates that the statement in Lemma \ref{lem:rip} holds for $\Phi'$. %  states that for any $\delta>0$, if $p\geq C(k+1) \log n$ for some constant $C$ and $n$ is large enough, then with probability at least $1-O(n^{-\alpha})$ for some constant $\alpha>0$ such that for every set $U\subseteq \{1,...,n\}$ with $|T|\leq 2k+2$ and for every vector $\bfx \in \mathcal{R}^{2k+2}$,
%\begin{equation}\label{eqn:phi'}
%(1-\delta)\|\bfx\|^2_2 \leq \|\Phi'_U \bfx \|_2^2 \leq (1+\delta)\|\bfx\|_2^2
%\end{equation}
%holds simultaneously.
 From (\ref{eqn:flip}), and the fact that $\|\Phi'_U \bfx \|_2^2 =\bfx^T {\Phi'_U}^{T}\Phi'_U \bfx $,   the statement also holds for $\hat{\Phi}$. %still holds if we replace $\Phi'_U$ with $\hat{\Phi}_U$ in (\ref{eqn:phi'}).
 Since $\hat{\Phi}$ and $\Phi$ follow the same probability distribution, the lemma follows.
\end{proof}

\begin{proof}%{(of Theorem \ref{thm:recovery})}
(of Theorem \ref{thm:recovery}) From Lemma \ref{lem:rip}, when $m \geq C(k+1) \log |S|$ for some $C>0$ and $|S|$ is large enough, with probability at least $1-O(|S|^{-\alpha})$, matrix $(2\hat{A}-W^{(m-1)\times |S|})/\sqrt{m-1}$ satisfies (\ref{eqn:rip}) for some small enough $\delta$, say $\delta<\sqrt{2}-1$. Then from \cite{CaT06,FL09}, %existing compressed sensing theory,
(\ref{eqn:recover1}) can recover all $k+1$-sparse vectors correctly.
\end{proof}

%As the number of nodes  in a  network goes to infinity, if the number of groups of nodes that can be measured together via some hub remains constant, and the number of nodes in each group goes to infinity, (one example of such network is $\mathcal{G}^4$,) then by applying
%Theorem \ref{thm:recovery} with a simple union bound, we know that with high probability $\ell_1$-minimization can successfully recover all the sparse vectors even if the hub measurements are erroneous, provided that while designing the measurement matrix for a general graph based on hubs, we  randomly generate  matrices with each entry taking value `0' and `1' with equal probability independently as the measurement matrices for complete graphs. %then with high probability, provided that for every group of nodes that are  measured together via some  hub, the number of nodes grows to infinity as $n$ grows to infinity, and the number  of measurements safi 

%% file: simu_arxiv.tex
\section{Simulation} \label{sec:simu}
\begin{figure*}[ht]
\begin{minipage}{2.3in}
%\begin{tabular}{c}
%\begin{center}
\begin{psfrags}
\psfrag{Upper bound of number of measurements}[bl][bl][1.1]{\tiny Upper bound of number of measurements}
\psfrag{Number of measurements}[bl][bl][1.1]{\tiny Number of measurements}
\psfrag{Radius}[bl][bl][1.1]{\tiny Radius}
\psfrag{Number of edges}[cl][l][1.2]{\tiny Number of edges}
\includegraphics[width=1.1\linewidth,height=0.65\linewidth]{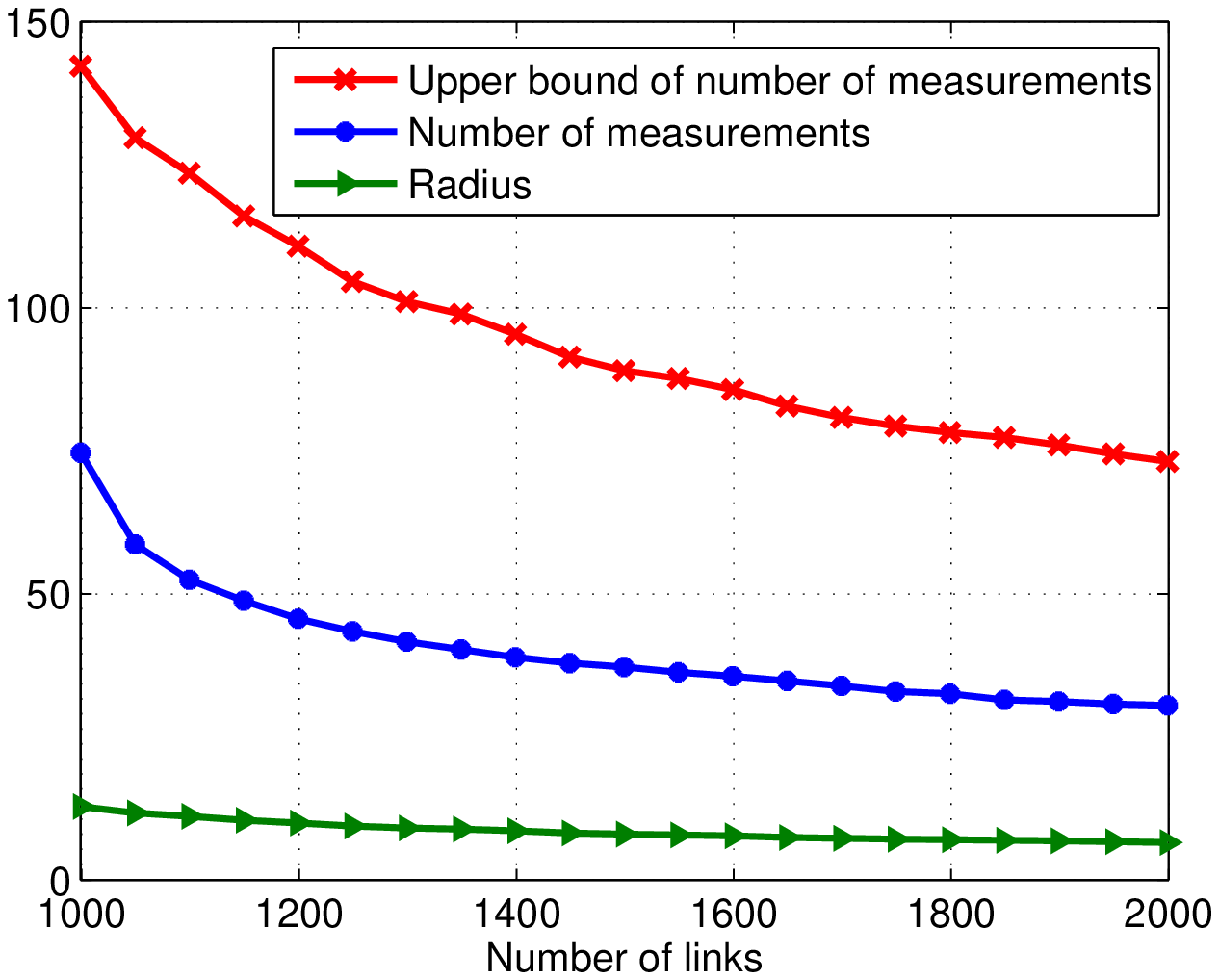}
\end{psfrags}
%\\
\caption{Random graph with $n=1000$}\label{fig:treeplus}
%\end{center}
%\end{tabular}
%\end{figure}
\end{minipage}
%\hfill
\begin{minipage}{2.3in}
%\begin{figure}
%\begin{center}
%\begin{tabular}{c}
\begin{psfrags}
\psfrag{m=1 m=1}[bl][bl][1.1]{\tiny $m=1$}
\psfrag{m=2 m=2}[bl][bl][1.1]{\tiny $m=2$}
\psfrag{m=3 m=3}[bl][bl][1.1]{\tiny $m=3$}
\psfrag{Number of nodes}[bl][l][1.2]{\tiny Number of nodes}
\psfrag{Number of measurements}[cl][l][1.2]{\tiny Number of measurements}
\includegraphics[width=1.1\linewidth,height=0.65\linewidth]{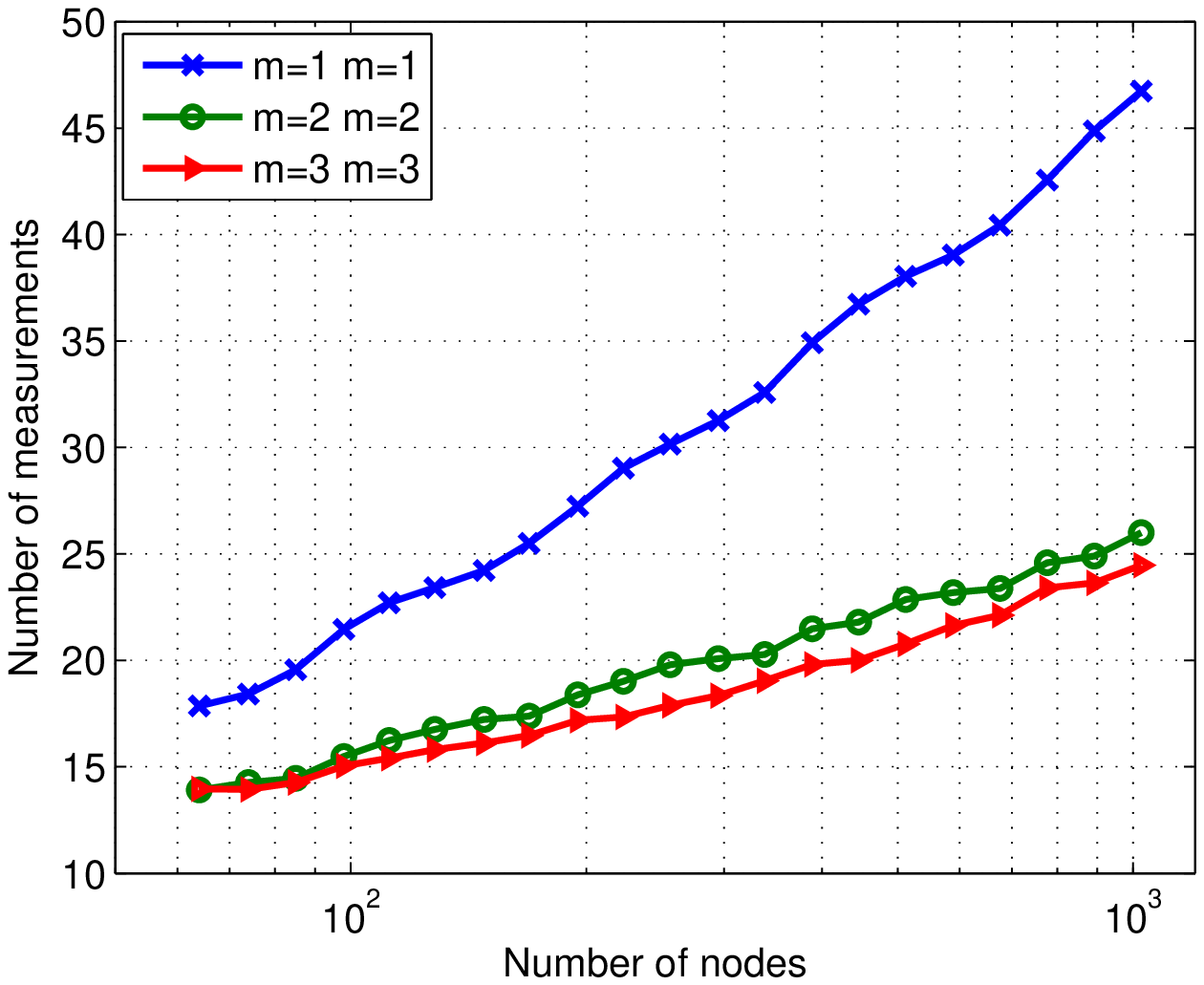}
\end{psfrags} %\\
%\end{figure}
\caption{BA model with increasing $n$}\label{fig:BA}
%\end{center}
%\end{tabular}
\end{minipage}
\begin{minipage}{2.3in}
%\begin{figure}
%\begin{center}
%\begin{tabular}{c}
\begin{psfrags}
\psfrag{ellell, with noise noise}[bl][bl][1.1]{\tiny $\ell_1$, with noise}
\psfrag{ellell, no noise noise}[bl][bl][1.1]{\tiny $\ell_1$, no noise}
\psfrag{Our method, with noise noise}[bl][bl][1.1]{\tiny Our method, with noise}
\psfrag{Our method, no noise noise}[bl][bl][1.1]{\tiny Our method, no noise}
\psfrag{Support size of the vectors}[cl][l][1.2]{\tiny Support size of the vectors}
\psfrag{xrxo2}[bc][l][1]{\tiny $\|\bfx_{\textrm{r}}-\bfx_0\|_2/\|\bfx_0\|_2$}
\includegraphics[width=1.1\linewidth,height=0.65\linewidth]{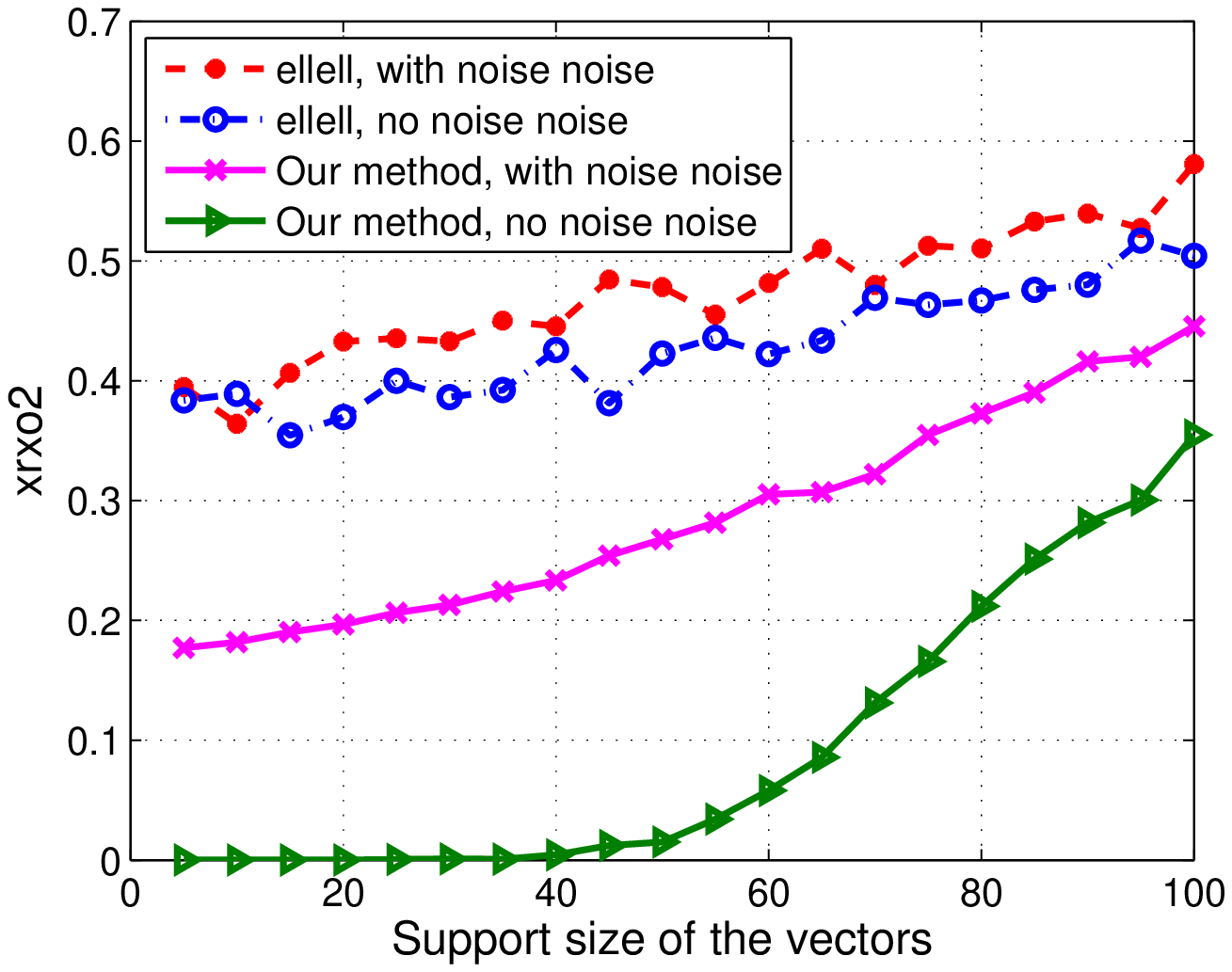}
\end{psfrags} %\\
%\end{figure}
\caption{Recovery performance with hub errors}\label{fig:simunoise}
%\end{center}
%\end{tabular}
\end{minipage}
\end{figure*}

\noindent\textbf{{Experiment 1 (Effectiveness of Algorithm 1):}} Given a graph $G$, we apply Algorithm 1 to  divide the nodes into groups such that each group (except the last one) can be measured   via some hub. % nodes in the same group (except the last group) can be measured freely via some hub that induces a connected subgraph. And
The last group   contains one node and can be measured directly.  %since in complete graphs,
 It is know that $M^C_{1,n}=\lceil \log (n+1) \rceil$, and the corresponding measurement matrix has the binary expansion of integer $i$ as column $i$ \cite{Dorfman43}. Also from (\ref{eqn:MC}) the number of measurements required to recovery $k$-sparse vectors is within a constant times $k M^C_{1,n}$. Therefore, here we design measurements to recover $1$-sparse vectors on $G$ as an example. %$M^C_{k, n} \leq O(k M^C_{1,n})$. % as an example.
%Note that $M^C_{1,n}=\lceil \log (n+1) \rceil$ %the minimum number of measurements needed to recover $1$-sparse vectors associated with a complete graph is $\lceil \log (n+1) \rceil$
%and the corresponding
%Since a matrix with  the binary expansion of interger $i$ as column $i$ of the measurement matrix can identify $1$-sparse vectors in a complete graph \cite{Dorfman43}, we construct such measurements for each group,
%The measurement matrix to recover $1$-sparse vectors in a complete graph has the binary expansion of $i$ as column $i$ \cite{Dorfman43}. Then
The total number of constructed measurements  is $\sum_i^{q-1} \lceil \log (n_i+1)\rceil+q$, where $n_i$ is the number of nodes in group $i$ and $q$ is the total number of groups. %To recover the sparse vectors from measurements, note that for $1$-sparse vectors, the corresponding subvector of measurements of each group equals to either the zero vector or some constant times the vector of the binary expansion of some integer $i$. Thus, one can easily recover the location and the value of the non-zero entry from measurements.

In Fig. \ref{fig:treeplus}, we gradually increase the number of edges in a graph with $n=1000$ nodes.  We start with a uniformly generated random tree, and in each step randomly add $25$ edges to the graph. %that do not already exist. %until the number of edges in the graph is $2n-1$.
All the results are averaged over one hundred realizations. The number of measurements constructed decreases from 73 to 30 when the number of edges increases from $n-1$ to $2n-1$. %Then number of measurements gradually . Note that
%Since here $M^C_{1, n}=\lceil \log (n+1)\rceil =10$, then %the number of measurements needed by a complete graph to recover $1$-sparse vectors is $\lceil \log (n+1)\rceil =10$. Therefore, the number of measurements by Algorithm 1 is always within $O(\log n)$, and when the average node degree is close to 4,
Note that the number of measurements is already within $3 M^C_{1, n}$ when the average node degree is close to 4. The radius $R$ of the graph decreases from 13 to 7, and we also plot the upper bound $R\lceil \log n \rceil+R+1$ %of the number of measurement %$R\lceil\log (n+1)\rceil +R+1$
 provided by Theorem \ref{prop:algo}. One can see that %the upper bound in general is quite loose, and
the number of measurements actually constructed is much less than the upper bound.

In Fig. \ref{fig:BA}, we consider the scale-free network with Barab\'asi-Albert (BA) model \cite{BA99} where the graph initially has $m_0$ connected nodes, and each new node connects to $m$ existing nodes with a probability that is proportional to the degree of the existing nodes. We start with a random tree of 10 nodes and increase the total number of nodes from 64 to 1024. Every result is averaged over one hundred realizations. Since the diameter of BA model is $O(\log n/\log \log n))$ \cite{BR04},  then by Theorem \ref{prop:algo}, the number of our constructed measurements is upper bounded by $O(\log^2 n/\log\log n ))$. %As the mixing time of BA model is $O(\log n)$ \cite{MPS06}, methods in \cite{CKMS12} and \cite{XMT11} require $O(\log^3 n)$ random measurements. %One can see that the number of measurements constructed is proportional to $\log n$, and decreases when $m$ increases.

\vspace{0.1in}

\noindent\textbf{{Experiment 2 (Recovery Performance with Hub Error):}} %We also simulate the sparse recovery performance with our constructed matrix.
%Compressed sensing theory indicates that if $A$ is a random 0-1 matrix, %the measurement matrix $A$ for the complete graph has i.i.d. entries with equal probability in '1' and '0', then
%with overwhelming probability we can recover the sparse vector $\bfx_0$ though $\ell_1$-minimization \cite{CaT06}. %which returns the vector with the least $\ell_1$-norm among all the vectors $\bfx$ such that $A\bfx=A\bfx_0$.
We generate a graph with $n=500$ nodes from BA model. Algorithm 1 divides nodes into four groups with 375, 122, 2 and 1 node respectively. %and finds a hub for each group.
For each of the first two groups with size $n_i$ ($i=1,2$), we generate $\lceil n_i/2\rceil$ random measurements each  measuring a random subset of the group together with its hub. Every node of the group is included in the random subset independently with probability 0.5. We also measure the two hubs directly. %or groups with size less than 150, we measure the nodes one by one. The total number of measurements is . In decoding,
Each of the three nodes in the next two groups is measured directly by one measurement. The generated matrix $A$ is 254 by 500. We generate a sparse vector $\bfx_0$ with i.i.d. zero-mean Gaussian entries on a randomly chosen support, and normalize $\|\bfx_0\|_2$ to 1.

To recover $\bfx_0$ from $\bfy=A\bfx_0$, one can run the widely used $\ell_1$-minimization \cite{CaT06} to recover the subvectors associated with the first two groups, and the last three entries of $\bfx_0$ can be obtained from measurements directly. However, as discussed in Section \ref{sec:huberror}, %note that every measurement for the first two groups passes through its hub, then
an error in a hub measurement degrades the recovery accuracy of subvectors associated with that group.
% every measurement for the group of nodes using this hub.
To address this issue, we   use a modified $\ell_1$-minimization in which the errors in the two hubs are treated as entries of an augmented vector to recover. Specifically, let the augmented vector $\bfz=[\bfx_0^T ,e_1, e_2]^T$ and the augmented matrix $\tilde{A}=[A  \  \bm\beta \  \bm \gamma]$, where $e_1$ (or $e_2$) denotes the error in the measurement of the first (second) hub, and the column vector $\bm \beta$ (or $\bm \gamma$) has  `1' in the row corresponding to the measurement of the first (or second) hub and `0' elsewhere. We then recover $\bfz$ (and thus $\bfx_0$) from $\bfy=\tilde{A}\bfz$ by running $\ell_1$-minimization on each group separately. %and let $\bfx_r$ denote the first $n$ entries of the recovered vector.

Fig. \ref{fig:simunoise} compares the recovery performance of our modified $\ell_1$-minimization and the conventional $\ell_1$-minimization, where the hub errors $e_1$ and $e_2$ are drawn from standard Gaussian distribution %$\mathcal{N}(0,1)$
with zero mean and unit variance. For every support size $k$, we randomly generate two hundred $k$-sparse vectors $\bfx_0$, and let $\bfx_{\textrm{r}}$ denote the recovered vector. Even with the hub errors, the average $\|\bfx_{\textrm{r}}-\bfx_0\|_2/\|\bfx_0\|_2$ is within $10^{-6}$ when $\bfx_0$ is at most 35-sparse by our method, while by $\ell_1$-minimization, the value is at least 0.35. We also consider the case that besides errors in  hub measurements,  every other measurement  has i.i.d.   Gaussian noise with zero mean and variance $0.04^2$. % Let $\bfw$ denote the noise vector and $\|\bfw\|_2$ is normalized to 2.
 The average $\|\bfx_{\textrm{r}}-\bfx_0\|_2/\|\bfx_0\|_2$ here is smaller with our method than that with $\ell_1$-minimization. %in the sense that the average $\|\bfx_r-\bfx_0\|_2$ is smaller in our method.